%% file: top.tex
\documentclass[journal,twocolumn]{IEEEtran}
\IEEEoverridecommandlockouts

\usepackage{ifthen}

\providecommand{\VersionLength}{long}
\providecommand{\ColumnNum}{2}

\newcommand{\col}{\ifthenelse{\equal{\ColumnNum}{2}}}
\newcommand{\ver}{\ifthenelse{\equal{\VersionLength}{long}}}

\input{arxiv_defs}

\allowdisplaybreaks

\begin{document}

\input{title}

\begin{abstract}
\input{abstract}
\end{abstract}

\section{Introduction}
\label{s:intro}
\input{intro}

\section{Notation}
\label{s:notation}
\input{notation}

\section{Common-Message Broadcast Channel Model}
\label{s:channel_model}
\input{channel_model}

\section{Single-User Scheme via Matrix Triangularization: Known Results}
\label{s:single_user}
\input{single_user}

\section{Multi-User Scheme via Matrix Triangularization}
\label{s:multi_user}
\input{multi_user}

\section{Space--Time Triangularization}
\label{s:space_time}
\input{space_time}

\section{Nearly-Optimal $K$-GMD}
\label{ss:space_time_nearly}
\input{space_time_nearly}

\section{Extensions}
\label{s:extentions}
\input{extentions}

\section{Discussion and Further Research}
\label{s:discussion}
\input{discussion}

\appendices

\section{Proof of \lemref{lem:det_adj_complex}}
\label{app:det_adj_complex_proof}
\input{det_adj_complex_proof}
% \flushright{\blacksquare}

\section{Proof of \lemref{lem:F_invariant}}
\label{app:F_invariant}
\input{F_invariant}

\ver{
\section{Reduction from $3 \times 3$ to $2 \times 2$ in the Rateless Problem}
\label{app:rateless_reduction}
\input{rateless_reduction}
}{}

\section{Proof of \thrmref{thm:extention_does_not_help}}
\label{app:extention_does_not_help}
\input{extention_does_not_help}

\section{Proof of  \lemref{lem:k_to_k_plus_one}}
\label{app:proof_k_to_k_plus_one}
\input{proof_gmd_jet}

\section{Proof of \thrmref{thm:n_n_asymptotical} for $n=2$, $K=3$, $N=4$}
\label{app:proof_K_3_n_2}
\input{proof_K_3_n_2}

\section{Proof of \thrmref{thm:n_n_asymptotical} for $n=2$ and General $K,N$}
\label{app:proof_n_2}
\input{proof_n_2}

\section{Proof of \thrmref{thm:n_n_asymptotical} for $K=2$ and general $n,N$}
\label{app:proof_K_2}
\input{proof_K_2}

\section{Proof of \thrmref{thm:n_n_asymptotical} for general $n,N,K$}
\label{app:proof_general}
\input{proof_general}

\section{Proof of \thrmref{thm:upper_lower_complex}}
\label{app:upper_lower_proof}
\input{upper_lower_proof}

\ver{
\section{Proof of \lemref{lem:GTD_with_multiplicities}}
\label{app:Proof_GTD_with_multiplicities}
\input{proof_GTD_with_multiplicities}
}{}

\bibliographystyle{IEEEtran}


\input{bios}

\end{document}

%% file: arxiv_defs.tex
\usepackage[tbtags]{amsmath} 
\usepackage{amssymb}  
\usepackage{verbatim}
\usepackage{amsxtra}
\usepackage{enumerate}
\usepackage{amsfonts}
\usepackage{color}
\usepackage{cite}
\usepackage{float}
\usepackage{epsfig, graphicx, psfrag, subfig}
\usepackage{units}
\usepackage[table]{xcolor}
\usepackage{calc}
\usepackage{amsfonts,amssymb}
\usepackage{epstopdf}
\usepackage[mathcal]{euscript}
\usepackage{multirow}
\usepackage{arydshln}
\usepackage{url}
\usepackage{caption}
\usepackage{cite}
\usepackage{mathtools}
\mathtoolsset{showonlyrefs=true}
\usepackage{pstool}

\DeclareFontFamily{OT1}{pzc}{}
\DeclareFontShape{OT1}{pzc}{m}{it}{<-> s * [1.10] pzcmi7t}{}
\DeclareMathAlphabet{\mathpzc}{OT1}{pzc}{m}{it}
\newtheorem{thm}{Theorem}
\newtheorem{lemma}{Lemma}
\newtheorem{remark}{Remark}

\newtheorem{corol}{Corollary}

\newtheorem{property}{Property}
\newtheorem{defn}{Definition}
\newtheorem{exmpl}{Example}

\newtheorem{assert}{Assertion}

\providecommand{\thrmref}[1]{Theorem~\ref{#1}}
\providecommand{\thmref}[1]{Theorem~\ref{#1}}
\providecommand{\exref}[1]{Example~\ref{#1}}
\providecommand{\defnref}[1]{Definition~\ref{#1}}
\providecommand{\secref}[1]{Section~\ref{#1}}
\providecommand{\lemref}[1]{Lemma~\ref{#1}}

\providecommand{\propertyref}[1]{Property~\ref{#1}}
\providecommand{\remref}[1]{Remark~\ref{#1}}
\providecommand{\figref}[1]{Figure~\ref{#1}}
\providecommand{\colref}[1]{Corollary~\ref{#1}}
\providecommand{\appref}[1]{Appendix~\ref{#1}}

\providecommand{\assertref}[1]{Assertion~\ref{#1}}

\newcommand{\ie}{i.e.}

\DeclareMathOperator{\R}{Re}
\DeclareMathOperator{\I}{Im}

\newcommand{\bm}[1]{\mbox{\boldmath{$#1$}}}

\newcommand{\SNR}{\text{SNR}}

\newcommand{\eff}{{\text{eff}}}

\newcommand{\mI}{\mathcal{I}}

\newcommand{\mT}{\mathcal{T}}

\newcommand{\mA}{\mathcal{A}}

\newcommand{\mV}{\mathcal{V}}

\newcommand{\Comment}[1]{}
\newcommand{\old}[1]{}
\newcommand{\rem}[1]{}

\newcommand{\inner}[2]{\left\langle #1, #2 \right\rangle}

\newcommand{\tQ}{\tilde Q}

\newcommand{\tH}{\tilde H}

\newcommand{\ty}{\tilde y}
\providecommand{\tx}{\tilde x}

\newcommand{\tR}{\tilde R}

\newcommand{\tA}{\tilde A}

\newcommand{\by}{\bm y}

\newcommand{\bx}{{\bm x}}

\newcommand{\bv}{{\bm v}}
\newcommand{\bz}{{\bm z}}
\newcommand{\ba}{{\bm a}}
\newcommand{\bb}{{\bm b}}
\newcommand{\br}{{\bm r}}

\newcommand{\bmu}{{\bm \mu}}

\newcommand{\bsigma}{{\bm \sigma}}

\providecommand{\bx}{{\bf x}}
\providecommand{\by}{{\bf y}}
\providecommand{\bv}{{\bf v}}
\providecommand{\bz}{{\bf z}}

\providecommand{\bp}{{\bf p}}

\newcommand{\cN}{{\mathcal N}}

\newcommand{\Norm}[1]{\left\| #1 \right\|}

\providecommand{\comment}[1]{}

\providecommand{\norm}[1]{\Norm{#1}}

\newcommand{\beqn}[1]{\begin{eqnarray}\label{#1}}
\newcommand{\eeqn}{\end{eqnarray}}
\newcommand{\beq}[1]{\begin{align}\label{#1}}
\newcommand{\eeq}{\end{align}}

\newcommand{\tby}{\tilde{\by}}
\newcommand{\tbx}{\tilde{\bx}}

\newcommand{\tba}{\tilde{\ba}}
\newcommand{\tbb}{\tilde{\bb}}
\newcommand{\cC}{{\cal C}}

\providecommand{\qqquad}{\quad \qquad}
\providecommand{\qqqquad}{\qquad \qquad}

\providecommand{\submat}[3]{#1 \left\lfloor #2, \, #3 \right\rceil}

\providecommand{\blkmat}[2]{\left\lceil   {#1}  \right\rfloor_{\otimes #2}}

\newcommand{\blath}[1]{\mbox{#1-th}}
\newcommand{\Capacity}[1]{I(H_{#1},C_\bx)}
\newcommand{\CC}{C_\bx}
\newcommand{\CCC}{C_\bx}
\newcommand{\UUU}{\mathcal{U}}
\newcommand{\VVV}{\mathcal{V}}
\newcommand{\GGG}{\mathcal{G}}
\newcommand{\HHH}{\mathcal{H}}
\newcommand{\QQQ}{\mathcal{Q}}
\newcommand{\OOO}{\mathcal{O}}
\newcommand{\AAA}{\mathcal{A}}
\newcommand{\TTT}{\mathcal{T}}
\newcommand{\RRR}{\mathcal{R}}
\newcommand{\SSS}{\mathcal{S}}
\providecommand{\trace}[1]{\mathop{\mathrm{tr}}\left( #1 \right)}
\newcommand{\adj}{\mathop{\mathrm{adj}}}

\DeclareMathAlphabet{\mathbsf}{OT1}{cmss}{bx}{n}% bold sans serif
\DeclareMathAlphabet{\mathssf}{OT1}{cmss}{m}{sl}% slanted sans serif
\DeclareMathAlphabet{\mathcsf}{OT1}{cmss}{sbc}{n}% condensed sans serif

\DeclareSymbolFont{bsfletters}{OT1}{cmss}{bx}{n}
\DeclareSymbolFont{ssfletters}{OT1}{cmss}{m}{n}
\DeclareMathSymbol{\bsfGamma}{0}{bsfletters}{'000}
\DeclareMathSymbol{\ssfGamma}{0}{ssfletters}{'000}
\DeclareMathSymbol{\bsfDelta}{0}{bsfletters}{'001}
\DeclareMathSymbol{\ssfDelta}{0}{ssfletters}{'001}
\DeclareMathSymbol{\bsfTheta}{0}{bsfletters}{'002}
\DeclareMathSymbol{\ssfTheta}{0}{ssfletters}{'002}
\DeclareMathSymbol{\bsfLambda}{0}{bsfletters}{'003}
\DeclareMathSymbol{\ssfLambda}{0}{ssfletters}{'003}
\DeclareMathSymbol{\bsfXi}{0}{bsfletters}{'004}
\DeclareMathSymbol{\ssfXi}{0}{ssfletters}{'004}
\DeclareMathSymbol{\bsfPi}{0}{bsfletters}{'005}
\DeclareMathSymbol{\ssfPi}{0}{ssfletters}{'005}
\DeclareMathSymbol{\bsfSigma}{0}{bsfletters}{'006}
\DeclareMathSymbol{\ssfSigma}{0}{ssfletters}{'006}
\DeclareMathSymbol{\bsfUpsilon}{0}{bsfletters}{'007}
\DeclareMathSymbol{\ssfUpsilon}{0}{ssfletters}{'007}
\DeclareMathSymbol{\bsfPhi}{0}{bsfletters}{'010}
\DeclareMathSymbol{\ssfPhi}{0}{ssfletters}{'010}
\DeclareMathSymbol{\bsfPsi}{0}{bsfletters}{'011}
\DeclareMathSymbol{\ssfPsi}{0}{ssfletters}{'011}
\DeclareMathSymbol{\bsfOmega}{0}{bsfletters}{'012}
\DeclareMathSymbol{\ssfOmega}{0}{ssfletters}{'012}
\newcommand{\UUUU}[2]{{  \left(\UUU_{#1}^{\left({#2}\right)}\right)^\dagger    }}
\newcommand{\Ul}[2]{{  \left(U_{#1}^{\left({#2}\right)}\right)^\dagger    }}
\newcommand{\VVVV}[1]{{  \VVV^{\left({#1}\right)}    }}
\newcommand{\Vl}[1]{{  V^{\left({#1}\right)}    }}
\providecommand{\subm}[2]{\left\lfloor #1, \, #2 \right\rceil}
\providecommand{\submm}[2]{\left\lfloor #1:\, #2 \right\rceil}
\providecommand{\submatt}[3]{#1 \left\lfloor #2: \, #3 \right\rceil}

\providecommand{\extt}[2]{\mI_{#1}^{\left[{#2}\right]}}
\providecommand{\embb}[3]{I_{#1}\left[#2\,;#3\right]}
\providecommand{\subunion}[3]{\bigcup_{#1}\subm{#2}{#3}}
\providecommand{\subunionn}[3]{\bigcup_{#1}\submm{#2}{#3}}
\newcommand{\Cptp}{C}
\newcommand{\nnn}{\tilde{n}}

\usepackage[T1]{fontenc}
\usepackage{frcursive}
\usepackage{calligra}
\newcommand{\setfont}[2]{{\fontfamily{#1}\selectfont #2}}

\newcommand{\xxx}{\text{\setfont{frc}{x}}}
\newcommand{\yyy}{\text{\setfont{frc}{y}}}
\newcommand{\zzz}{\text{\setfont{frc}{z}}}
\newcommand{\vvv}{\text{\setfont{frc}{v}}}

\providecommand{\diff}[1]{\ensuremath{\operatorname{d}\!{#1}}}
\usepackage{chngcntr}
\counterwithin{paragraph}{section}

%% file: title.tex
\title{Joint Unitary Triangularization for Gaussian Multi-User MIMO Networks}
\ver{
\author{Anatoly Khina, Idan Livni, Ayal Hitron, and Uri Erez \\ \emph{Technical Report, Dept.\ of EE-Systems, Tel Aviv University, March 27, 2015}

    \thanks{
      The authors are with the Department of Electrical Engineering,
      Tel Aviv university, Israel, email: \mbox{\{anatolyk,idanlivn,ayal,uri\}@eng.tau.ac.il}
        Parts of this work were presented at the International Symposium on Information Theory (ISIT) 2011
        in Saint Petersburg, Russia,
        and at the International Symposium on Information Theory (ISIT) 2012 in
        Cambridge, MA, USA.

        The work of Anatoly Khina was supported in part by the Feder family award, the David and Paulina Trotsky Foundation, and by the Clore Israel Foundation.
        The work of Anatoly Khina and Ayal Hitron was supported in part by the Yitzhak and Chaya
        Weinstein Research Institute for Signal Processing.
        The work of Uri Erez was supported in part by the Israel Science Foundation
under Grant No. 1557/12.
    }
}
}{
\author{Anatoly Khina, Idan Livni, Ayal Hitron, and Uri Erez, \emph{Member, IEEE}

    \thanks{
Copyright (c) 2014 IEEE. Personal use of this material is permitted.  However, permission to use this material for any other purposes must be obtained from the IEEE by sending a request to pubs-permissions@ieee.org.
      The authors are with the Department of Electrical Engineering,
      Tel Aviv university, Israel, email: \mbox{\{anatolyk,idanlivn,ayal,uri\}@eng.tau.ac.il}
        Parts of this work were presented at the International Symposium on Information Theory (ISIT) 2011
        in Saint Petersburg, Russia,
        and at the International Symposium on Information Theory (ISIT) 2012 in
        Cambridge, MA, USA.

	The work of Anatoly Khina was supported in part by the Feder family award, the David and Paulina Trotsky Foundation, and by the Clore Israel Foundation.
	The work of Anatoly Khina and Ayal Hitron was supported in part by the Yitzhak and Chaya
	Weinstein Research Institute for Signal Processing.
	The work of Uri Erez was supported in part by the Israel Science Foundation
under Grant No. 1557/12.
    }
}

}

\maketitle

%% file: abstract.tex
The problem of transmitting a common message to multiple users over the Gaussian multiple-input multiple-output broadcast channel is considered,
where each user is equipped with an arbitrary number of antennas.
A closed-loop scenario is assumed, for which a practical capacity-approaching scheme is developed.
By applying judiciously chosen unitary operations at the transmit and receive nodes, the channel matrices are triangularized so that the resulting matrices have equal diagonals, up to a possible multiplicative scalar factor. This, along with the utilization of successive interference cancellation, reduces the coding and decoding tasks to those of coding and decoding over the single-antenna additive white Gaussian noise channel. Over the resulting effective channel, any off-the-shelf code may be used.
For the two-user case, it was recently shown that such joint unitary triangularization is always possible.
In this paper, it is shown that for more than two users, it is necessary to carry out the unitary linear processing jointly over multiple channel uses, i.e., space--time processing is employed.
It is further shown that exact triangularization, where all resulting diagonals are equal, is still not always possible,
and appropriate conditions for the existence of such are established for certain cases.
When exact triangularization is not possible, an asymptotic construction is proposed,
that achieves the desired property of equal diagonals up to edge effects that can be made arbitrarily small, at the price of processing a sufficiently large number of channel uses together.

\begin{IEEEkeywords}
Matrix decompositions, space--time modulation, common-message broadcast, physical-layer multicast, Gaussian MIMO, successive interference cancellation.
\end{IEEEkeywords}

%% file: intro.tex
\PARstart{A} {\MakeUppercase recurring} theme in digital communications is the use of a standard ``off-the-shelf'' coding module in
combination with appropriate linear pre/post processing which is tailored to the specific channel model. Such
methods are appealing due to their low complexity of implementation as well as conceptually, since the tasks 
of coding and modulation are effectively decoupled.

The simplest example of the decoupling approach is provided by the singular-value decomposition (SVD) in communication 
for single-user (SU) Gaussian multiple-input multiple-output (MIMO) channels.
In this case, the MIMO channel is transformed into diagonal form, corresponding to parallel scalar channels.
If one allows pre- or post-interference cancellation, a much broader class of decompositions may be employed.
For SU MIMO communication, this includes the widely used schemes based on the QR decomposition, namely, V-BLAST/GDFE \cite{Wolniansky_V-BLAST,CioffiForneyGDFE}.
Further applicable decompositions that allow to approach capacity via decoupling, 
include the geometric-mean decomposition (GMD) \cite{GMD,QRS,UnityTriangularization} for the SU case, 
and its generalization~--- block diagonal GMD \cite{BlockDiagonalGMD}~--- for private-message broadcast (BC).

In the present work, we aim to extend the decoupling approach beyond the single-user Gaussian MIMO channel, to the more general problem of common-message BC.
That is, we consider a scenario in which a transmitter, equipped with multiple antennas, wishes to send simultaneously the same (``common'') message to multiple users, each of which equipped with (any number of) multiple antennas.

The capacity of this scenario, referred to as common-message BC (or ``physical-layer multicast''), is well known, and is given by the compound channel capacity  \cite{Dobrushin59,BlackwellBreimanThomasian59,Wolfowitz60}.
Unfortunately, whereas for the problem of transmitting private messages over the Gaussian MIMO BC channel, capacity can be achieved via decoupling (in conjunction with dirty-paper coding; see, e.g., \cite{Tejera05,BlockDiagonalGMD}), 
practical schemes that attain an analogous result for the common-message counterpart are not hitherto known.

Beyond being important in its own right, common-message BC  
serves as the basis for various communication settings,
since many communication scenarios can be transformed into an equivalent MIMO common-message BC setting.
\ver{
This is the case for rateless coding over SISO and MIMO Gaussian channels \cite{JET:RatelessITW2011} (see also \secref{ss:multi_user_rateless}), permuted channels \cite{JET:Permuted_ISIT2012} (see also \secref{ss:multi_user_perms}), half- and full-duplex SISO and MIMO relaying \cite{JET:RatelessITW2011,JET:FullDuplexRelayITW2012}, two-way MIMO relaying \cite{JET:TwoWayRelayISIT11,JET:TwoWayRelay:ISITA2014} and many others.
}{
This is the case for rateless coding over SISO and MIMO Gaussian channels \cite{JET:RatelessITW2011,ayal_thesis},
permuted channels \cite{JET:Permuted_ISIT2012} (see also \secref{ss:multi_user_perms}), half- and full-duplex SISO and MIMO relaying \cite{JET:RatelessITW2011,JET:FullDuplexRelayITW2012}, two-way MIMO relaying \cite{JET:TwoWayRelayISIT11,JET:TwoWayRelay:ISITA2014} and many others.
}

Extension of the decoupling approach, which is at the heart of single-user scalar systems, to the multiple-user MIMO common-message BC problem requires, however, overcoming a major hurdle: 
Not only is simultaneous diagonalization impossible, even the existence of appropriate joint triangularization for two users was not known to be possible until recently \cite{STUD:SP}.

Hence, different practical approaches have been proposed over the years for the problem of conveying a common message over Gaussian MIMO broadcast channels.
However, none of these approaches is capacity achieving in general, even for simple cases.
To illustrate this, we consider a simple three-user example.

\begin{exmpl}[Degrees-of-freedom mismatch] \label{DOF-example}
 Consider the following three-user channel:\footnote{Throughout this paper, vectors are denoted by boldface lower case letters, and matrices are denoted by upper case letters. Logarithms are taken to base $2$ and rates are given in bits.}
\begin{align}
    \by_k = H_k \bx + \bz_k \,, \qquad k=1,2,3 \,,
\end{align}
where $\bz_k$ is an additive white Gaussian noise (AWGN), specifically we assume to be circularly-symmetric Gaussian noise  with unit power for each element $\cC \cN(0,I)$, $\bx$ is the channel vector subject to an average power constraint $P$,
$H_k$ are the complex-valued channel matrices
\begin{align*}
    H_1 = \left(
            \begin{array}{cc}
              \alpha_1 & 0 \\
            \end{array}
          \right) \col{}{\,},
    \col{\:}{\quad}
    H_2 = \left(
            \begin{array}{cc}
              0 & \alpha_1 \\
            \end{array}
          \right) \col{}{\,},
    \col{\:}{\quad}
    H_3 = \left(
            \begin{array}{cc}
              \alpha_2 & 0 \\
              0 & \alpha_2 \\
            \end{array}
          \right)
    \col{}{\,},
\end{align*}
and $\alpha_1$ and $\alpha_2$ are chosen such that the WI capacities of all three channels are equal, viz.
\begin{align*}
  \Cptp^{\text{common}}=\Cptp_{\text{WI}}
  \triangleq
  \log(1 + |\alpha_1|^2 P/2)
  =
  2 \log(1 + |\alpha_2|^2 P/2) \,.
\end{align*}
This example models a three-user ``degrees-of-freedom-mismatch'' scenario, in which the first two users are equipped with a single antenna each (i.e., they have only one degree of freedom), whereas the third user is equipped with two antennas (i.e., has two degrees of freedom).

\end{exmpl}

Of course, from a purely information-theoretic viewpoint, a random i.i.d. Gaussian codebook over time and space is simultaneously good (i.e., capacity achieving) for all three users in the example. However when considering \emph{practical} codes, the situation is very different. 

To the best of our knowledge, known \emph{practical} schemes are limited to the smallest number of degrees of freedom (``multiplexing gain'') of the different users,
or incorporate time- or frequency-sharing, which again lose degrees of freedom.
Thus, these schemes achieve only a fraction of the available degrees of freedom.
Alternatively, maximal degrees-of-freedom open-loop techniques may be used (e.g., in the case of two transmit antennas as in the example, golden code modulation \cite{GoldenCodes:YaoWornell,GoldenCodes:BelfioreViterbo,GoldenCodes:Oggier,GoldenCodes:Kumar}). However, these are far from capacity-achieving at low to moderate transmission rates. 

By using single-stream communication, in the high SNR regime, the third user
is able to achieve only half of its individual capacity.
On the other hand, transmitting two streams across the two transmit antennas,
results in a loss of half of the capacity of users $1$ and $2$.
Another approach considered in the literature for this problem is that of using a ``pure open-loop'' approach,
namely Alamouti modulation \cite{Alamouti}~--- for the two-transmit antenna case, and orthogonal space--time block coding (OSTBC)~\cite{TarokhJafarkhaniCalderbank_STBC}~--- for more.
The performance of these schemes does not depend on the number of receivers.
However, this universality comes at the price of a substantial rate loss for MIMO channels having several receive antennas, as these schemes use only a single stream, thus failing to achieve the multiplexing gain
offered by the MIMO channel of user $3$ in the example.\footnote{Moreover, for more than two transmit antennas, the OSTBC of \cite{TarokhJafarkhaniCalderbank_STBC} attain strictly less than one degree of freedom.}
Also note that time/frequency sharing incur a great loss in performance (up to half of the capacity in this case).
Other techniques that can be applied for this scenario \cite{Gohary03,LopezPhD,TavildarViswanath_UniversalCodes} are also suboptimal in general.

The aim of the present work is to develop a practical capacity-achieving scheme for the Gaussian MIMO common-message broadcast MIMO setting via decoupling, 
allowing to utilize a ``black box'' approach to coding. 
Namely, this approach allows constructing a capacity-achieving scheme that utilizes only ``off-the-shelf'' encoders and decoders designed for scalar AWGN channels, together with simple signal processing tools.

We construct a capacity-approaching scheme that applies judiciously chosen unitary operations to the time-extended channel matrices at the transmitter and the receivers in conjunction with successive interference cancellation.
In contrast to the open-loop OSTBC structures, that strive for an ``orthogonal design'' structure, i.e., to diagonalize the channel matrices (see, e.g., \cite{TarokhJafarkhaniCalderbank_STBC}),
the space--time structure presented in this work results in triangular matrices, similar to those of V-BLAST/GDFE, but having \emph{equal diagonals}.
This gives rise to effective parallel scalar additive white Gaussian noise (AWGN) channels,
over which standard codes can be used to approach capacity.
Thus, the proposed scheme can be thought of as an ``interpolation'' between the open-loop OSTBC and the closed-loop SU V-BLAST/SVD ones.

The results of this paper generalize those of \cite{STUD:SP}, in which the case of only two users was considered,
for which it suffices to apply unitary transformations directly to the channel matrices.
For more users, on the other hand, we show that jointly processing multiple channel uses is necessary. That is, the unitary transformations are applied to  time-extended channel matrices.

The rest of the paper is organized as follows.
In \secref{s:notation} we present the notations that are used throughout the paper.
In \secref{s:channel_model} we define the Gaussian MIMO common-message BC channel model.
In \secref{s:single_user} we recall known schemes for the single-user case, relying on various forms of unitary matrix decompositions. In \secref{s:multi_user} we suggest a generalization of the SU schemes to the multi-user scenario, based on newly developed matrix decompositions  and derive necessary and sufficient conditions for the existence of such decompositions in some scenarios. Then, in \secref{s:space_time}, we generalize the multi-user scheme by employing space--time coding and discuss the existence of ``perfect'' decompositions needed for such a construction.
In \secref{ss:space_time_nearly} we utilize the space--time structure in order to develop a practical scheme, which is nearly optimal and asymptotically achieves the capacity for any number of users, even when ``perfect decompositions'' are not possible. Finally, in \secref{s:extentions} we present some extensions of the results and conclude in \secref{s:discussion}.

%% file: notation.tex
The following notation will be used throughout the paper:

\begin{itemize}
 \item Channel matrix of dimension $n_r \times n_t$: $H$, where $n_r$ and $n_t$ stand for the number of antennas at the receiver and at the transmitter, respectively.
\item Channel gain: $\alpha$.
\item Augmented channel matrix: $\tilde{H}$, see \defnref{def:augmented} in \secref{ss:p2p_scheme}.
\item Channel canonical matrix: $G$, see \defnref{def:canonical} in \secref{ss:p2p_scheme}.
\item General square complex matrix of dimensions $n \times n$: $A$.
\item Hermitian square matrix: $S$.
\item Upper triangular matrix with diagonal $\br$: $R$.
\item Upper triangular matrix with a constant diagonal: $T$.
\item Real-valued diagonal matrix: $D$.
\item Complex-valued matrices whose columns are orthonormal (which are unitary, in case these matrices are square): $U$,$V$,$Q$.
\item The Identity matrix: $I$.
\item Capital script letters denote time-extended matrices: $\HHH,\AAA,\SSS,\RRR,\TTT,\UUU,\VVV,\QQQ,\GGG$, see \secref{ss:space_time_intro}.
\item Number of users: $K$.
\item Number of time extensions: $N$.
\item Vectors are denoted by boldface lower case letters. For example, $\bx$ denotes the transmitted vector, $\by$~--- the received vector, and $\bz$~--- the noise vector.
\item Time-extended vectors are denoted by script lower case letters. For example,
$\xxx,\yyy$ and $\zzz$ denote extended transmit, received and noise vectors, respectively.
\item Indices: $j,k,l,m,p,q$.
\item Channel capacity: $C$.
\item All logarithms are taken to base $2$. All rates are given in bits per two dimensions (complex channel use).
\item Average power constraint: $P$.
\item Covariance matrix of the vector $\bx$: $C_{\bx}$.
\item Singular values and generalized singular values: $\bsigma$, $\bmu$.
\item Real and imaginary parts of a complex number: $\R \{\cdot\}$, $\I \{\cdot\}$.
\item Expected value of a random variable: $\mathbb{E}(\cdot)$.
\item Vector $\ell_2$ norm: $\left\| \cdot \right\|$.
\item Determinant of a matrix: $\det(\cdot)$.
\item Trace of a matrix: $\trace{\cdot}$.
\item Adjugate (the transpose of the cofactor) matrix: $\adj(\cdot)$.
\end{itemize}

%% file: channel_model.tex
The $K$-user Gaussian MIMO broadcast channel consists of one transmit and $K$ receive nodes, where each received signal is related to the transmitted signal through a MIMO link:\footnote{For ease of notation, in the case $K=1$ we denote the single channel matrix $H_1$ by $H$.}
\begin{align} \label{MIMO-BC}
  \by_k = H_k \bx + \bz_k \,, \qquad k = 1,\dots,K \,,
\end{align}

\noindent where $\bx$ is the channel input of dimensions $n_t \times 1$, and is subject to an average power constraint $P$;\footnote{Alternatively, one can consider any other input covariance constraint, e.g., individual power constraints, and covariance matrix constraints.
Given any covariance matrix, the approach described in the sequel may be applied to approach 
\eqref{MIMO_MI}.}
$\by_k$ is the channel output vector of receiver $k$ ($k=1,\dots,K$) of dimensions
$n_r^{(k)} \times 1$; $H_k$ is the channel matrix to user $k$ of dimensions $n_r^{(k)} \times n_t$; and $\bz_k$ is an additive circularly-symmetric Gaussian noise vector of dimensions $n_r^{(k)} \times 1$, where, without loss of generality, we assume that the noise elements are mutually independent and identically distributed with unit power.

The aim of the transmitter is to send the same (common) message to all the receivers.
The capacity of this scenario is well known to equal the (worst-case) capacity of the
compound channel \cite{Dobrushin59,BlackwellBreimanThomasian59,Wolfowitz60}, with the compound parameter being the channel matrix index:
\begin{align}
\label{eq:MIMO_BC_capacity}
    C \left( \left\{ H_k \right\}_{k=1}^K, P \right) = \max_{C_\bx} \min_{k=1,\dots,K} I(H_k,C_\bx) \,,
\end{align}

\noindent where $I(H_k,C_\bx)$ is the mutual information between the channel input $\bx$ and the channel output
$\by_i$, obtained by taking $\bx$ to be Gaussian with covariance matrix $C_\bx$:
\begin{align} \label{MIMO_MI}
  I(H,C_\bx) \triangleq \log \det \left( I + H \CCC H^\dagger \right) \,,
\end{align}

\noindent and the maximization is carried over all admissible input covariance matrices $C_\bx$, satisfying the power constraint
$\trace{ C_{\bx} } \leq P$.

For $K=1$ (SU), the capacity \eqref{eq:MIMO_BC_capacity} can be achieved via the decoupling approach in several ways, each corresponding to a different matrix decomposition.

%% file: single_user.tex
In this section we briefly recall some important matrix decompositions, and
the associated SU communication schemes.
In \secref{ss:decomposition} we recall the generalized triangular decomposition (GTD),
and some of its important special cases which include the SVD, QR, and GMD. A geometrical interpretation of these decompositions is provided in \secref{ss:GeometricInterpretationGTD}.
In \secref{ss:p2p_scheme}, we describe how the GTD can be used in order to construct a practical capacity-achieving communication scheme for the SU Gaussian MIMO communication problem.

\subsection{Generalized Triangular Decomposition}
\label{ss:decomposition}

We only consider the decomposition of \emph{square} invertible matrices throughout this work. As we show in the sequel, this does not impose any restriction on the communication problems addressed.

The next theorem uses the following definition:

\begin{defn}[Multiplicative Majorization (See \cite{PalomarJiang})]
\label{def:major}
    Let $\bx$ and $\by$ be two $n$-dimensional vectors of positive elements.
    Denote by $\tbx$ and $\tby$ the vectors composed of the entries of $\bx$ and $\by$, respectively, ordered non-increasingly.
    We say that $\bx$ majorizes $\by$ ($\bx \succeq \by$) if they have equal products:
\vspace{-.35\baselineskip}
    \[
        \prod_{j=1}^n  x_j  =  \prod_{j=1}^n y_j  \,,
    \]
    and their (ordered) elements satisfy,
    for any $1 \leq l < n$,
\vspace{-.35\baselineskip}
    \[
        \prod_{j=1}^l  \tx_j  \geq  \prod_{j=1}^l  \ty_j  \,.
    \]
\end{defn}

\begin{thm}[Generalized Triangular Decomposition]
\label{thm:gtd}
Let $A$ be an invertible matrix of dimensions $n \times n$ and $\br$ be an $n$-dimensional vector of positive
 elements.
A GTD of the matrix $A$ is given by:
    \begin{align} \label{eq:gtd}
        A &= U R V^\dagger \,,
    \end{align}
where $U$, $V$ are unitary matrices, and $R$ is an upper triangular matrix
with a prescribed set of diagonal values
$\br$, where $r_j = R_{jj}$. This decomposition
exists if and only if the vector $\br$
is majorized by the singular-values vector of $A$:
\begin{align}
\label{eq:majorization_gtd}
    \bsigma(A) \succeq \br \,.
\end{align}
\end{thm}

In other words, the singular values are an extremal case for the diagonal of all possible unitary triangularizations.

The necessity of the majorization condition was proven by Weyl \cite{WeylCondition},
and the sufficiency of this condition~--- by Horn \cite{WeylConditionInverse_ByHorn}.
Explicit constructions of the decomposition
were introduced in  \cite{QRS-GTD} and \cite{GTD}.

We now recall three important special cases of the GTD.
\subsubsection{SVD (See, .e.g., \cite{GolubVanLoan3rdEd})}
An important special case of the GTD is the SVD, in which
the resulting matrix $R$ in \eqref{eq:gtd} is a \emph{diagonal} matrix, such that the diagonal elements
of $R$ are equal to the singular values of the original matrix $A$.

\subsubsection{QR Decomposition (See, .e.g., \cite{GolubVanLoan3rdEd})}
Another important special case of the GTD is the QR decomposition, in which the matrix
$V$ in \eqref{eq:gtd} equals to the identity matrix and hence does not depend on the matrix $A$.
This decomposition can be constructed by performing Gram-Schmidt orthonormalization on the (ordered) columns of the
matrix $A$.

\subsubsection{GMD (See \cite{UnityTriangularization,QRS,GMD})}
A GMD of a square complex invertible matrix $A$ is given by:
    \begin{align} \label{eq:gmd}
        A &= U T V^\dagger \,,
    \end{align}
where $U$, $V$ are unitary matrices, and $T$ is an upper triangular matrix
such that all its diagonal values equal
to the geometric mean of the singular values of $A$, which is real and positive.

Note that this decomposition always exists if $A$ is invertible (since the vector of singular values of $A$ necessarily majorizes
the vector of diagonal elements of $T$), but is not unique.

\subsection{Geometric Interpretation of the GTD}
\label{ss:GeometricInterpretationGTD}

\input{gtd_interpretation.tex}

\subsection{SU MIMO transmission via Matrix Triangularization}
\label{ss:p2p_scheme}

We now review the capacity-approaching communication schemes that utilize the above matrix decompositions.
For the SU case (i.e., $K=1$ in \eqref{MIMO-BC}), a practical communication scheme can be obtained by applying the SVD to the channel matrix $H$:
\begin{align}
      H = U D V^\dagger \,.
\end{align}
By applying the pre-processing matrix $V$ at the transmitter and the post-processing matrix $U^\dagger$
at the receiver, the resulting effective channel matrix becomes \emph{diagonal}, and therefore the capacity can be achieved using off-the-shelf codes, designed for \emph{scalar} SU AWGN channels. The rates of those codes are determined by the SNRs of the independent scalar sub-channels, namely, by the diagonal elements of the diagonal matrix $D$
(after allocating power to the resulting sub-channels, via water-pouring).

We now review a more general scheme, applicable to any GTD rather than the special case of SVD.
This scheme is based upon the derivation
of the MMSE variant of Vertical Bell-Laboratories Space--Time coding (V-BLAST), see, e.g., \cite{HassibiVBLAST,UCD,QRS}.

\begin{defn}[Augmented Matrix]
\label{def:augmented}
Define the following \emph{augmented matrix}:\footnote{$\CCC^{1/2}$ is any matrix $B$ satisfying $B B^\dagger = \CCC$, and can be found, e.g., via the Cholesky decomposition.}
\begin{align}
\label{eq:augmented_matrix_def}
	\tH \triangleq
	\left(\begin{array}{c}
		  H \CCC^{1/2} \\ I_{n_t}
	      \end{array}
	\right)\,,
\end{align}
where $I_{n_t}$ is the $n_t \times n_t$ identity matrix.
Next, the matrix $\tH$ is transformed into a square matrix, by means of the QR decomposition.
\end{defn}

\begin{defn}[Channel Canonical Matrix]
\label{def:canonical}
Let $\tH$ be the augmented matrix \eqref{eq:augmented_matrix_def}, and let
\begin{align} \label{eq:G_matrix}
 \tH = Q G \,,
\end{align}
where $Q$ is an $(n_r+n_t) \times n_t$ matrix with orthonormal columns and $G$ is an $n_t \times n_t$ upper triangular matrix with real-valued positive diagonal elements.
The matrix $G$ will be referred to as a \emph{channel canonical matrix}, reminiscent of the system canonical response defined in
\cite{CDFE-PartI} for LTI scalar systems.
\end{defn}

Now the matrix G is decomposed according to the GTD:
\begin{align}
\label{eq:Ggtd}
    G &= U R V^\dagger \,,
\end{align}
where $R$ is upper triangular whose diagonal values are equal to the prescribed diagonal elements $r_1,\ldots,r_{n_t}$ (which satisfy the multiplicative majorization condition of \defnref{def:major}),
and $r_j^2-1$ are the effective signal-to-noise ratios of the scalar sub-channels.

\begin{remark}
\label{rem:greater_than_1_diag}
    Due to the presence of the identity matrix $I_{n_t}$ in \eqref{eq:augmented_matrix_def}, 
    it follows that the the diagonal elements of $G$ and $R$ are necessarily greater or equal to $1$, 
    and their determinants are greater than 1.\footnote{Assuming a ``canonical QR decomposition'' is used, i.e., 
    the one that results in positive diagonal entries in the triangular matrix.}
\end{remark}

The  transmission scheme is as follows:
\begin{enumerate}
\item
  Construct $n_t$ codewords, each from a codebook matched to a scalar AWGN channel of signal-to-noise ratio (SNR) $r_j^2-1$. That is, up to a rate of $\log r_j^2$.

\item
  In each channel use, an $n_t$-length vector $\tbx$ is formed using one sample from each codebook. The transmitted vector $\bx$ is then obtained using the following linear precoder:
  \begin{align} \label{eq:bx_tilde_to_bx}
	  \bx = \CCC^{1/2} V \tbx \,.
  \end{align}
\item
  The receiver calculates
  \begin{align} \label{eq:scheme_receiver}
	  \tby = U^\dagger \tQ^\dagger \by \,,
  \end{align}
  where $\tQ$ consists of the first $n_t$ rows of $Q$.

\item
  Finally, the codebooks are decoded using successive interference cancellation, starting from the \blath{$n_t$}  codeword and ending with the first one: The \blath{$n_t$} codeword is decoded first, using the \blath{$n_t$} element of $\tby$,
  treating the other codewords as AWGN. The  effect of the \blath{$n_t$} element of $\tbx$ is then subtracted out from the remaining elements of $\tby$. Next, the \blath{$(n_t-1)$} codeword is decoded, using the \blath{$(n_t-1)$} element of $\tby$ --- and so forth.
\end{enumerate}
The proof of optimality of this scheme, i.e., that it is capacity achieving, appears in~\cite[Lemma~III.3]{UCD}.

Note that each element of $\tbx$ should be understood to correspond to a symbol of a codebook of length $L$. Thus, the index time is suppressed. Similarly, the successive interference cancellation process of recovering the codebooks from $\tby$ should be understood, again, to correspond to a symbol of a codebook of length $L$.
Our analysis is not affected by the exact value of $L$, but rather only by the gap to capacity of the base code. Hence, in order to approach capacity, $L$ needs to be large. Throughout this paper, we assume capacity-achieving scalar (base) codes; any loss in these codes,
would translate in a straightforward manner to a loss in the overall scheme.

\begin{remark} \label{rem:gmd_same_codebook}
   If we take $V=I$ in \eqref{eq:Ggtd}, namely use the QR decomposition, we obtain a transmission scheme that requires no precoding at the transmitter. Since the QR decomposition is unique, we have no freedom in choosing the diagonal values $r_j$.
  Alternatively, the matrices $U$ and $V$ can be chosen according to the SVD. In this case, the resulting matrix $R$ in \eqref{eq:Ggtd} is \emph{diagonal}, and therefore the channel is transformed into parallel independent scalar sub-channels and there is no need to perform successive interference cancellation. As in the case of the QR decomposition, the SVD is unique, and there is no freedom in choosing the diagonal values $r_j$ (which, in this case, are the singular values of the matrix $G$).
  Finally,
  If the matrices $U$ and $V$ are chosen according to the GMD \eqref{eq:gmd}, then all the values $r_j$ are equal, meaning that all the codebooks in the scheme have the same rate.
  Moreover, in this case the \emph{same}  scalar codebook can be used over all the sub-channels.\footnote{In practice, the codebooks should not be identical, though they can, for example, be derived from a common base codebook via scrambling.}
  This special case is known as the uniform channel decomposition (UCD) \cite{UCD}.
\end{remark}

\begin{remark}[Decoding Order] \label{rem:order}
 In step 4 of the scheme, one could decode the codebooks in a different order. This corresponds to replacing the QR decomposition \eqref{eq:Ggtd} with Gram-Schmidt orthonormalization in a different order, e.g., QL decomposition. Alternatively, this could be represented in the notations of this section by retaining the QR decomposition, but performing it on a column-permuted matrix $G \Pi$, where $\Pi$ is some permutation matrix. This, in general, would alter the rate allocation between the different sub-streams.
\end{remark}

%% file: gtd_interpretation.tex
We give a geometric interpretation of the GTD of \thmref{thm:gtd},
for the special case of $2 \times 2$ real matrices.
A similar geometric interpretation can be devised for the general case.

In the real case, unitary matrices reduce to (real) orthogonal ones.
In the $2 \times 2$ case, these orthogonal matrices are merely
rotation matrices.\footnote{In general, reflection matrices need to be considered in conjunction with the rotation matrices. However, reflection matrices are not needed for
the construction of GTD, as will become clear in the sequel.}
Thus, the matrices $U$ and $V$ of \thmref{thm:gtd} are rotation matrices, namely,
\begin{align}
  V&=\left( \begin{array}{cc}
	    \cos\theta_r & -\sin\theta_r \\
	    \sin\theta_r &  \cos\theta_r
	  \end{array}
    \right)
\label{eq:gtd_interprtation:V}
  \\
  U&=\left( \begin{array}{cc}
	    \cos\theta_\ell & -\sin\theta_\ell \\
	    \sin\theta_\ell &  \cos\theta_\ell
	  \end{array}
    \right)\,,
\end{align}
where $\theta_r$ and $\theta_\ell$ are the rotation angles.

Denote the columns of the matrix to be decomposed, $A$, by $\ba$ and $\bb$:
\begin{align}
    &A \triangleq
    \begin{pmatrix}
        \ba & \bb
    \end{pmatrix}
    \triangleq 
    \left( \begin{array}{cc}
          a_x & b_x \\
       a_y & b_y
         \end{array}
    \right)
\end{align}
and assume, without loss of generality, $\det(A)=1$.

By multiplying $A$ by $V$ on the right, we obtain
\col{
\begin{align}
  &
\!\!\!\!
  AV
  =
  \begin{bmatrix}
      a_x \cos\theta_r + b_x \sin\theta_r & -a_x\sin\theta_r + b_x \cos\theta_r \\
      a_y \cos\theta_r + b_y \sin\theta_r & -a_y\sin\theta_r + b_y \cos\theta_r
  \end{bmatrix}
  \\
  &
\!\!\!\!
=
  \begin{bmatrix}
      a_x \cos\theta_r+b_x\sin\theta_r &
      a_x \cos(\theta_r+\frac{\pi}{2}) + b_x \sin(\theta_r+\frac{\pi}{2}) \\
      a_y \cos\theta_r+b_y\sin\theta_r &
      a_y \cos(\theta_r+\frac{\pi}{2}) + b_y \sin(\theta_r+\frac{\pi}{2})
  \end{bmatrix}
\label{eq:GTD_interp:AV}
\end{align}
}
{
\begin{align}
  AV
  &=
  \begin{bmatrix}
      a_x \cos\theta_r + b_x \sin\theta_r & -a_x\sin\theta_r + b_x \cos\theta_r \\
      a_y \cos\theta_r + b_y \sin\theta_r & -a_y\sin\theta_r + b_y \cos\theta_r
  \end{bmatrix}
\\ &=
  \begin{bmatrix}
      a_x \cos\theta_r+b_x\sin\theta_r &
      a_x \cos(\theta_r+\frac{\pi}{2}) + b_x \sin(\theta_r+\frac{\pi}{2}) \\
      a_y \cos\theta_r+b_y\sin\theta_r &
      a_y \cos(\theta_r+\frac{\pi}{2}) + b_y \sin(\theta_r+\frac{\pi}{2})
  \end{bmatrix}
\label{eq:GTD_interp:AV}
\end{align}
}

By varying the rotation angle $\theta_r$,
it is readily verified that the resulting column vectors in \eqref{eq:GTD_interp:AV},
move along an ellipse, centered at the origin.
This is illustrated in \figref{fig:right_rotation}, for a specific choice of $A$, 
where we define $\tA \triangleq AV$ and its columns~--- by $\tba$ and $\tbb$.
\col{
\begin{figure}[t]
    \centering
    \epsfig{file = 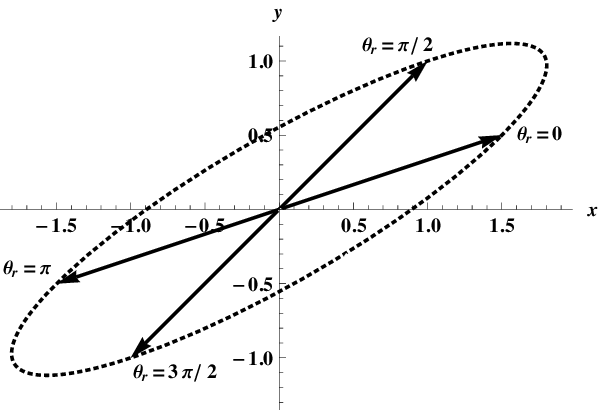, width = \columnwidth}
    \caption{All possible column vectors of $AV$, where $V$ is a rotation matrix, and
	     $\ba = (3/2, 1/2)^\dagger$ and $\bb = (1,1)^\dagger$. The arrows correspond to $\tba$ at different angles.}
    \label{fig:right_rotation}
\end{figure}
}
{
\begin{figure}[h]
    \centering
    \epsfig{file = 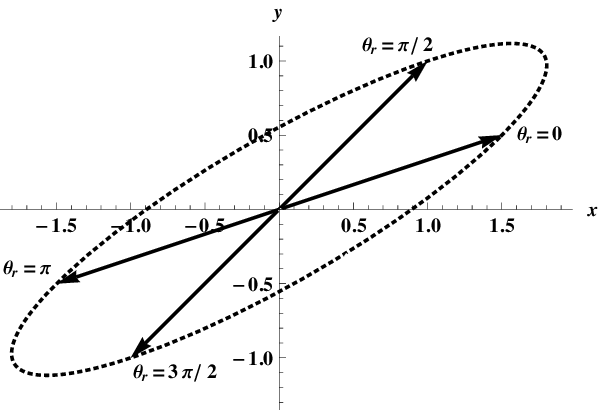, scale = 0.9}
    \caption{All possible column vectors of $AV$, where $V$ is a rotation matrix, and
	     $\ba = (3/2, 1/2)^\dagger$ and $\bb = (1,1)^\dagger$. The arrows corresponding to $\tba$ at the angles specified.}
    \label{fig:right_rotation}
\end{figure}
}

After applying $V$ on the right, we multiply the resulting matrix $\tA$
by a rotation matrix $U^\dagger$ on the left.
The latter operation rotates the column vectors $\tba$ and $\tbb$, by an angle $(-\theta_\ell)$
(the minus is due to the transposition of $U$ prior to multiplication).
The angle $\theta_\ell$ is chosen such that $U^\dagger \tba$
is aligned with the $x$-axis.
This is illustrated for a specific choice of $\tba$ and $\tbb$ in \figref{fig:qr}.

\col{
\begin{figure}[thp]
    \centering
    \epsfig{file = 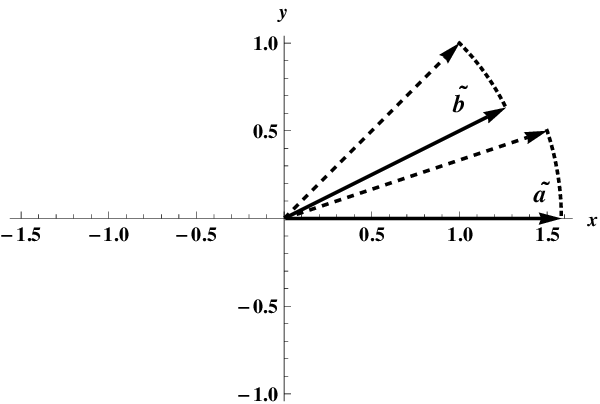, width = \columnwidth}
    \caption{Rotation by $U^\dagger$ of $\tba$ and $\tbb$, resulting from $U^\dagger \tA$ (or alternatively of $\ba$ and $\bb$ in the QR decomposition case), until $U^\dagger \tba$
	     is aligned with the $x$-axis, for $\tba = (3/2, 1/2)^\dagger$ and $\tbb = (1,1)^\dagger$.}
    \label{fig:qr}
\end{figure}
}
{
\begin{figure}[h]
    \centering
    \epsfig{file = 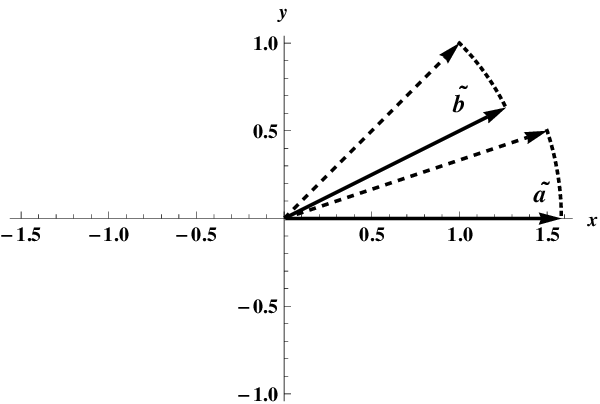, scale = 0.9}
    \caption{Rotation by $U^\dagger$ of $\tba$ and $\tbb$ $U^\dagger \tA$ (or alternatively of $\ba$ and $\bb$ in the QR decomposition case), until $U^\dagger \tba$
	     is aligned with the $x$-axis, for $\tba = (3/2, 1/2)^\dagger$ and $\tbb = (1,1)^\dagger$.}
    \label{fig:qr}
\end{figure}
}

\begin{remark}
\label{rem:unitary_operation_effects}
  Since the orthogonal matrix $V$ is applied on the right,
  the norms of the \emph{rows} of $A$ are not affected.
  Nevertheless, the \emph{columns} of $AV$ have different norms, in general,
  from those of the columns of $A$, as can be seen from \eqref{eq:GTD_interp:AV}.
  The multiplication on the left by $U^\dagger$, on the other hand,
  does not change the norms of the columns.
  As for the angle between the column vectors~--- multiplication by a unitary matrix $V$ on the right changes the relative angle between the two vectors,
  unlike a unitary operation applied on the left, which only rotates the two vectors
  together, but does not change the relative angle between the two.
\end{remark}

Since the norms of the columns are not affected by unitary operations applied on the left,
the possible values on the diagonal of the resulting triangular matrix
in the GTD, are fully determined by the norms (``lengths'') of the column vectors
resulting after applying $V$ on the right, which in turn,
vary together on an ellipse.

We next interpret geometrically the special cases of SVD, QR and GMD (for the real $2 \times 2$ case).

\subsubsection{SVD}
\label{sss:geometric_interp:SVD}

In this decomposition, the resulting columns, at the end of the process,
must be orthogonal.
This is established by choosing $\theta_r$ such that the relative angle between the resulting vectors, after the multiplication by $V$, is $\nicefrac{\pi}{2}$. As we show below, this is always possible.
Afterwards, the two vectors are rotated together via the left-multiplication by $U^\dagger$,
until they lie parallel to the axes.
This process is demonstrated in \figref{fig:svd}.
\col{
\begin{figure}[t]
\centering
\subfloat[Right rotation by $\theta_r = 3.865$, for which the vectors are orthogonal.]{\label{subfig:svd_v}
   \centering
   \epsfig{file = 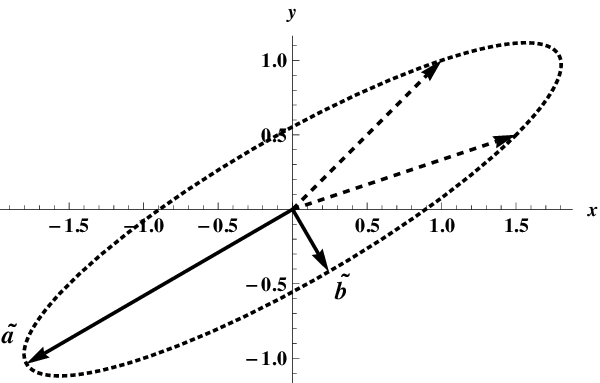, width = \columnwidth}
}
\\ 
\subfloat[Left rot.\ by $\theta_\ell=3.667$, for which the vectors are aligned with the axes.]{\label{subfig:svd_u}
   \centering
   \epsfig{file = 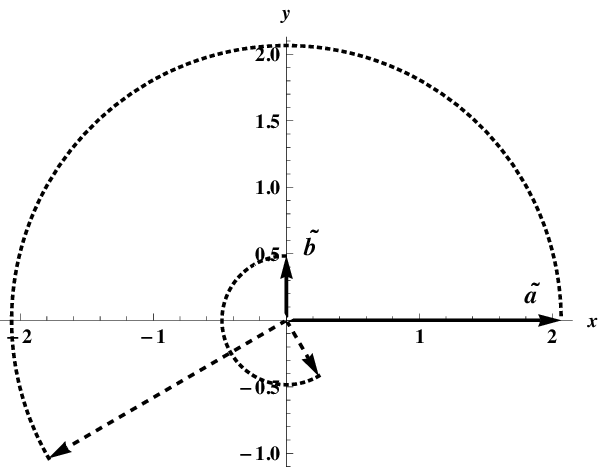, width = \columnwidth}
}
\caption{SVD for $\ba = (3/2,1/2)^\dagger$ and $\bb = (1,1)^\dagger$.}
\label{fig:svd}
\end{figure}
}
{
\begin{figure}[ht]
\centering
\subfloat[Right rotation by $\theta_r = 3.865$, for which the vectors are orthogonal.]{\label{subfig:svd_v}
   \centering
   \epsfig{file = 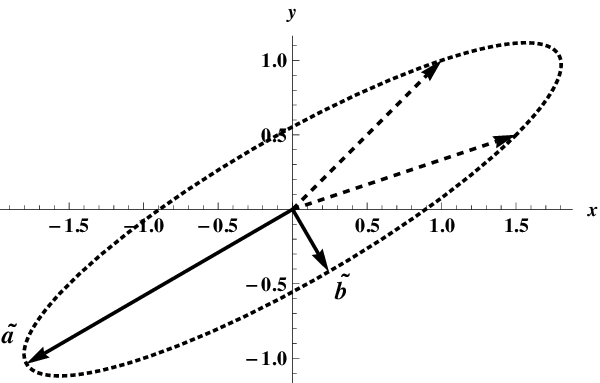, scale=0.8}
}
\qquad
\subfloat[Left rot.\ by $\theta_\ell=3.667$, for which the vectors are aligned with the axes.]{\label{subfig:svd_u}
   \centering
   \epsfig{file = 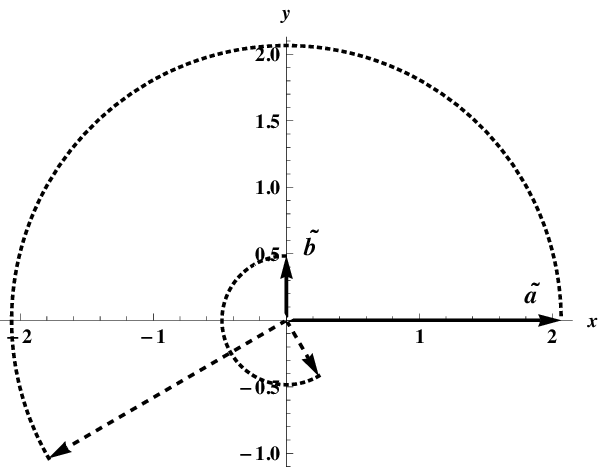, scale=0.8}
}
\caption{SVD for $\ba = (3/2,1/2)^\dagger$ and $\bb = (1,1)^\dagger$.}
\label{fig:svd}
\end{figure}
}
Moreover, the resulting orthogonal vectors correspond also to the longest and shortest (``extreme'') possible diagonal values achievable via the GTD.
This can also be seen in \figref{fig:svd} and is formally stated in the following lemma.
Note that this is a special ($2 \times 2$) case of the majorization property \eqref{eq:majorization_gtd} of the GTD. Here, we provide a geometric proof.

\begin{proof}
  The norm of $\tba$ after applying a rotation matrix $V$ on the right is
\col{
  \begin{align}
      \norm{\tba}^2 &= \left( a_x \cos \theta_r + b_x \sin \theta_r \right)^2
      + \left( a_y \cos \theta_r + b_y \sin \theta_r \right)^2
   \\ &= \frac{1}{2} \left( a_x^2 + a_y^2 + b_x^2 + b_y^2 \right)
   \\ &\quad + \frac{1}{2} \left( a_x^2 + a_y^2 - b_x^2 - b_y^2 \right) \cos 2\theta_r
   \\ &\quad + \left( a_x b_x + a_y b_y \right) \sin 2\theta_r \,.
  \end{align}
}
{
 \begin{align}
      \norm{\tba}^2 &= \left( a_x \cos \theta_r + b_x \sin \theta_r \right)^2
      + \left( a_y \cos \theta_r + b_y \sin \theta_r \right)^2
   \\ &= \frac{1}{2} \left( a_x^2 + a_y^2 + b_x^2 + b_y^2 \right)
   + \frac{1}{2} \left( a_x^2 + a_y^2 - b_x^2 - b_y^2 \right) \cos 2\theta_r
   + \left( a_x b_x + a_y b_y \right) \sin 2\theta_r \,.
  \end{align}
}
  Similarly, the norm of $\tbb$ is given by
\col{
  \begin{align}
      \norm{\tbb}^2
      &= \frac{1}{2} \left( a_x^2 + a_y^2 + b_x^2 + b_y^2 \right)
   \\ &\quad - \frac{1}{2} \left( a_x^2 + a_y^2 - b_x^2 - b_y^2 \right) \cos 2\theta_r
   \\ &\quad - \left( a_x b_x + a_y b_y \right) \sin 2\theta_r \,.
  \end{align}
}
{
\begin{align}
      \norm{\tbb}^2
      &= \frac{1}{2} \left( a_x^2 + a_y^2 + b_x^2 + b_y^2 \right)
   - \frac{1}{2} \left( a_x^2 + a_y^2 - b_x^2 - b_y^2 \right) \cos 2\theta_r
   - \left( a_x b_x + a_y b_y \right) \sin 2\theta_r \,.
  \end{align}
}
  The extreme values of $\norm{\tba}^2$ and $\norm{\tbb}^2$ are achieved at $\theta_r$ satisfying:
\col{
  \begin{subequations}
   \label{eq:GTD_interp:derivative}
   \noeqref{eq:GTD_interp:derivative:derivatives,eq:GTD_interp:derivative:first_half,eq:GTD_interp:derivative:second_half,eq:GTD_interp:derivative:zero}
  \begin{align}
        - \frac{\diff{\left(\norm{\tbb}^2 \right)}}{\diff{\theta_r}}
       &= \frac{\diff{\left(\norm{\tba}^2 \right)}}{\diff{\theta_r}}
\label{eq:GTD_interp:derivative:derivatives}
    \\ &= - \left( a_x^2 + a_y^2 - b_x^2 - b_y^2 \right) \sin 2\theta_r
\label{eq:GTD_interp:derivative:first_half}
   \\* &\quad + 2 \left( a_x b_x + a_y b_y \right) \cos 2\theta_r
\label{eq:GTD_interp:derivative:second_half}
   \\* &= 0 \,.
\label{eq:GTD_interp:derivative:zero}
  \end{align}
  \end{subequations}
}
{
 \begin{align}
        - \frac{\diff{\left(\norm{\tbb}^2 \right)}}{\diff{\theta_r}}
       &= \frac{\diff{\left(\norm{\tba}^2 \right)}}{\diff{\theta_r}}
    \\ &= - \left( a_x^2 + a_y^2 - b_x^2 - b_y^2 \right) \sin 2\theta_r
    + 2 \left( a_x b_x + a_y b_y \right) \cos 2\theta_r
   \\ &= 0 \,.
   \label{eq:GTD_interp:derivative}
  \end{align}
}
  On the other hand, the vectors $\tba$ and $\tbb$ are orthogonal for
  $\theta_r$ values satisfying
\col{
  \begin{subequations}
   \label{eq:GTD_interp:orthogonal}
    \noeqref{eq:GTD_interp:orthogonal:derivatives,eq:GTD_interp:orthogonal:second_half,eq:GTD_interp:orthogonal:zero}
  \begin{align}
      \inner{\ba}{\bb} &=
      -\frac{1}{2} \left( a_x^2 + a_y^2 - b_x^2 - b_y^2 \right) \sin 2\theta
\label{eq:GTD_interp:orthogonal:derivatives}
   \\ &+ \left( a_x b_x + a_y b_y \right) \cos 2\theta
\label{eq:GTD_interp:orthogonal:second_half}
   \\ &=0 \,.
\label{eq:GTD_interp:orthogonal:zero}
  \end{align}
  \end{subequations}
}
{
\begin{align}
      \inner{\ba}{\bb} &=
      -\frac{1}{2} \left( a_x^2 + a_y^2 - b_x^2 - b_y^2 \right) \sin 2\theta
   + \left( a_x b_x + a_y b_y \right) \cos 2\theta
   =0 \,.
   \label{eq:GTD_interp:orthogonal}
  \end{align}
}
  Observing that the requirements of \eqref{eq:GTD_interp:derivative} and \eqref{eq:GTD_interp:orthogonal}
  are the same, and that the second derivatives of $\norm{\tba}^2$ and $\norm{\tbb}^2$ are opposite,
  we conclude the desired result.
\end{proof}

\subsubsection{QR Decomposition}
\label{sss:geometric_interp:QR}

In this decomposition no right rotation V is applied, i.e.,
$V=I$ or equivalently $\theta_r = 0$.
Thus, a left rotation is applied to the columns of $A$,
until the first column vector is aligned with the $x$-axis.
This suggests that the first diagonal element is equal to the norm of the
first column of $A$ (prior to rotation);
the second diagonal element can be computed from the determinant and the first
diagonal vector, or alternatively by computing the norm of the orthogonal
component of the second column vector to the first one. See \figref{fig:qr}.

\subsubsection{GMD}
\label{sss:geometric_interp:GMD}

In this decomposition the angle $\theta_r$ is chosen such that the length (norm) of $\tilde{\ba}$
is equal to 1, or equivalently we seek for an angle $\theta_r$
for which the ellipse intersects with the unit circle.
Since both the ellipse and the unit circle
(which corresponds to the $2 \times 2$ identity matrix)
have determinants equal to 1 (i.e., have the same area) and both are centered at the origin, they must intersect at exactly 4 points,
unless the ellipse is itself the unit circle (in which case there is an infinite number of intersection points).
The operation on the left rotates the two vectors until the first is aligned with the $x$-axis.
Moreover, since unitary operations preserve volume (absolute value of the determinant),
the second diagonal element must be 1 as well. That is, the projection of the
second vector on the $y$-axis is equal to 1.
The remaining element may be found, e.g., via the Frobenius norm, which is again invariant under rotations on both sides, and its sign may be easily determined as well.
This is demonstrated in \figref{fig:GMD}.
\col{
\begin{figure}[t]
  \centering
  \subfloat[Right rotation by $\theta_r = 1.843$, for which the first vector has unit norm.]{\label{subfig:gmd_v}
    \centering
    \epsfig{file = 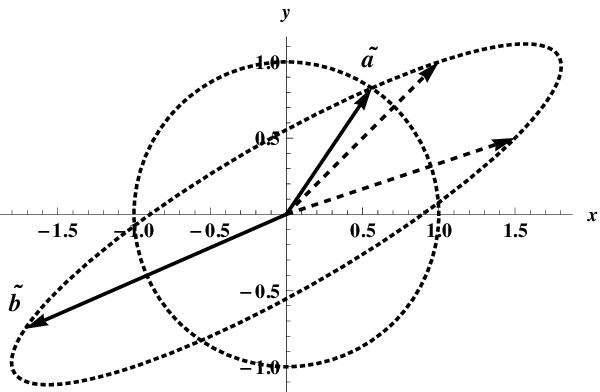, width = \columnwidth}
  }
\\
  \subfloat[Left rotation  by $\theta_\ell=0.977$, for which the first vector is aligned with the \textbf{\emph{x}} axis.]{\label{subfig:gmd_u}
    \centering
    \epsfig{file = 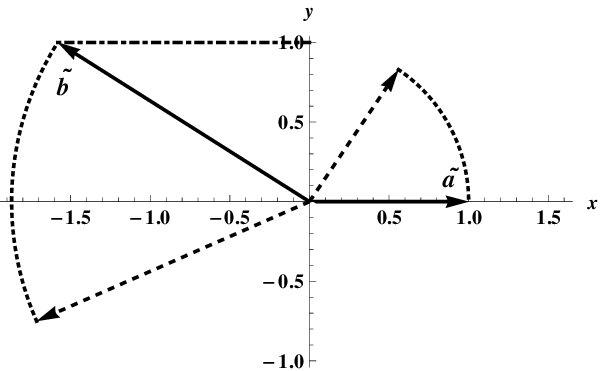, width = \columnwidth}
  }
  \caption{GMD for $\ba = (3/2,1/2)^\dagger$ and $\bb = (1,1)^\dagger$.}
\label{fig:GMD}
\end{figure}
}
{
\begin{figure}[h]
  \centering
  \subfloat[Right rotation by $\theta_r = 1.843$, for which the first vector has unit norm.]{\label{subfig:gmd_v}
    \centering
    \epsfig{file = 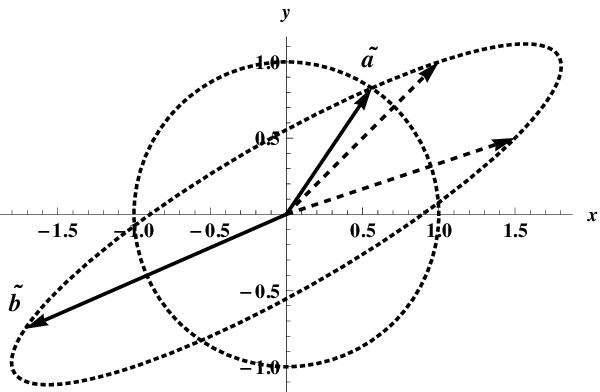, scale=0.8}
  }
\qquad
  \subfloat[Left rotation  by $\theta_\ell=0.977$, for which the first vector is aligned with the \textbf{\emph{x}} axis.]{\label{subfig:gmd_u}
    \centering
    \epsfig{file = 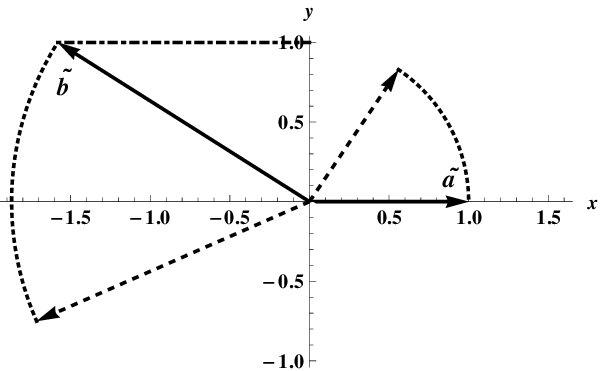, scale=0.8}
  }
  \caption{GMD for $\ba = (1.5,0.5)^\dagger$ and $\bb = (1,1)^\dagger$.}
\label{fig:GMD}
\end{figure}
}

%% file: multi_user.tex
The goal of this section is to generalize the point-to-point communication scheme, presented in
\secref{ss:p2p_scheme}, to the $K$-user BC channel defined in \secref{s:channel_model}.
This is a generalization of the two-user case ($K=2$) that was considered in \cite{STUD:SP}.

We start in \secref{ss:multi_user_intro} by defining some forms of joint decomposition
of $K$ matrices.
Namely, we define the $K$-user geometric mean decomposition ($K$-GMD)
and the $K$-user joint equi-diagonal triangularization ($K$-JET).
A communication scheme for the $K$-user common-message BC setting, based on these decompositions, is described in \secref{ss:multi_user_scheme}. Unfortunately, these decompositions do not always exist; In \secref{ss:multi_user_perfect_2x2} we provide necessary and sufficient conditions for the existence of
these decompositions, for a certain special case.

\subsection{$K$-JET and $K$-GMD}
\label{ss:multi_user_intro}
\input{multi_user_intro}

\subsection{Geometric Interpretation of the JET}
\label{ss:GeometricInterpretationJET}
\input{jet_interpretation.tex}

\col{
\subsection{MIMO Common-Message Broadcast Scheme via \\ Matrix Decomposition}
}{
\subsection{MIMO Common-Message Broadcast Scheme via Matrix Decomposition}
}
\label{ss:multi_user_scheme}
\input{multi_user_scheme}

\subsection{Perfect $2$-GMD for $2 \times 2$ Matrices}
\label{ss:multi_user_perfect_2x2}
\input{multi_user_perfect_two_by_two}

\ver{
\subsection{Example: ``Rateless'' Codes over the AWGN Channel}
\label{ss:multi_user_rateless}
\input{multi_user_rateless}

}{}

\subsection{Example: Arbitrarily Permuted Parallel Channels}
\label{ss:multi_user_perms}
\input{multi_user_perms}

%% file: multi_user_intro.tex
We now present the definitions of $K$-GMD and $K$-JET~---
decompositions of $K$ square matrices of the same dimensions and having the same determinant.

\begin{defn}[$K$-JET]
\label{defn:K-JET}
    Let  $A_1,\ldots,A_K$ be $K$ invertible complex matrices of dimensions $n \times n$,
    with equal determinants.
    A $K$-JET of these matrices is a decomposition
    \begin{align} \label{eq:K-JET}
      A_k &= U_k R_k V^\dagger \,, \qquad k=1,\ldots,K \,,
    \end{align}
    where $U_1,\ldots,U_K,V$ are $n \times n$ unitary matrices,  and $R_1,\ldots,R_K$ are upper triangular
     $n \times n$ matrices with \emph{the same} real, positive diagonal values, namely,
    \begin{align} \label{eq:R_equal_diagonals}
        \left[ R_1 \right]_{jj} = \cdots = \left[ R_K \right]_{jj} \,,\qquad j=1,\ldots,n \,.
    \end{align}
\end{defn}

\begin{remark}
     For $K=2$, $2$-JET will be simply referred to as JET.
     JET of two matrices was introduced in \cite{STUD:SP},
     where it was proved to always exist (for any two matrices $A_1$ and $A_2$ with equal determinants).
\end{remark}

\begin{remark}
\label{rem:JET:equal_det}
  The $K$-JET of \defnref{defn:K-JET} easily extends to matrices with non-equal determinants as follows.
  Define the normalized matrices
  \begin{align}
  \label{eq:rem:JET:equal_det:scaled_matrices}
      \tA_k \triangleq \left| \det(A_k) \right|^{-1/n} A_k \,.
  \end{align}
  These scaled matrices have unit determinants.\footnote{Up to a scalar phase which can be absorbed in the left-unitary matrices~$\{ U_k \}$.}
  Applying $K$-JET to the scaled matrices $\{ \tA_k \}$, 
  results in triangular matrices $\{\tR_k\}$ with equal diagonals, and a set of unitary matrices $\{U_k\}$ and $V$.
  This, in turn, suggests the following joint decomposition of the matrices $\{A_k\}$:
  \begin{align}
      A_k &= U_k R_k V^\dagger \,, \qquad k = 1, \ldots, K \,,
  \end{align}
  where 
  \begin{align}
	R_k &\triangleq \left| \det(A_k) \right|^{1/n} \tR_k \,.
  \end{align}
  Thus, $K$-JET applied to matrices having non-equal determinants, gives rise to triangular matrices having \emph{proportional} diagonals (instead of the equal diagonals, in the equal-determinant case).
  This is illustrated in the following example.
\end{remark}

\begin{exmpl}
\label{ex:non_equal_dets}
    Consider the following two matrices having 
    \emph{non-equal} determinants:
    \begin{align}
    A_1 &=
    \begin{pmatrix}
        2 & 1 \\
        0 & 8
    \end{pmatrix}
    = 4
    \underbrace{
    \begin{pmatrix}
        0.5 & 0.25 \\
        0 & 2
    \end{pmatrix}
    }_{\tA_1}
    \,,&
    \det \left( A_1 \right) &= 16 \,,\\
    A_2 &=
    \begin{pmatrix}
        5 & -2 \\
        0 & 5
    \end{pmatrix}
    = 5
    \underbrace{
    \begin{pmatrix}
        1 & -0.4 \\
        0 & 1
    \end{pmatrix}
    }_{\tA_2}
    \,,&
    \det \left( A_2 \right) &= 25 \,.
    \end{align}
By applying JET to ${\tA_1}$ and ${\tA_2}$, we obtain the following triangular matrices:\ver{}{\footnote{See \cite{JET:TechReport:SeveralUsers2013} for the corresponding matrices $U_1$, $U_2$ and $V$.}}
     \begin{align}
        \tR_1  &\approx
	\left( \begin{array}{cc}
	  1.20 & -1.48 \\
	  0 &  0.84
        \end{array} \right) 
	\Rightarrow
	R_1 = 4
	\left( \begin{array}{cc}
	  1.20 & -1.48 \\
	  0 &  0.84
        \end{array} \right) 
        \\
	\tR_2 &\approx
	\left( \begin{array}{cc}
	  1.20 & -0.17 \\
	  0 &  0.84
        \end{array} \right) 
	\Rightarrow
	R_2 = 5
	\left( \begin{array}{cc}
	  1.20 & -0.17 \\
	  0 &  0.84
        \end{array} \right) 
	\,,
    \end{align}
\ver{by applying the unitary matrices
    \begin{align}
        U_1  &\approx
        \left( \begin{array}{cc}
          -0.22 & -0.98 \\
          0.98 &  -0.22
        \end{array} \right) 
    \,, \quad
        U_2  \approx
        \left( \begin{array}{cc}
          -0.87 & -0.49 \\
          0.49 &  -0.87
        \end{array} \right) 
    \\ V  &\approx
        \left( \begin{array}{cc}
          -0.81 & -0.58 \\
          0.58 &  -0.81
        \end{array} \right) 
        \,.
    \end{align}
}{}

    Hence, the original matrices $A_1$ and $A_2$ can be simultaneously triangularized as follows
     \begin{align}
        A_1  &\approx 4 
	U_1 
	\left( \begin{array}{cc}
	  1.20 & -1.48 \\
	  0 &  0.84
        \end{array} \right) 
	V^\dagger
	\approx
	U_1
	\begin{pmatrix}
	  4.79 &  -5.91
	 \\ 0  &  3.34
	\end{pmatrix}
	V^\dagger
	\,,
        \\
	A_2 &\approx 5
	U_2
	\left( \begin{array}{cc}
	  1.20 & -0.17 \\
	  0 &  0.84
        \end{array} \right) 
	V^\dagger
	\approx U_2
	\begin{pmatrix}
	5.99 &  -0.85
	 \\ 0  &  4.18
	\end{pmatrix}
	V^\dagger
	\,.
    \end{align}
\end{exmpl}

\begin{defn}[$K$-GMD]
\label{defn:K-GMD}
    The $K$-GMD is a special case of the $K$-JET where the entries on the diagonal are constant, namely
    \begin{align}
         \left[ R_k \right]_{jj} = \sqrt[n]{\det A_k} \,,&& \begin{array}{l}
                                                           	k=1,\ldots,K \\
    					  		j=1,\ldots,n
         \,.
                                                           \end{array}
    \end{align}
    In this case the resulting upper triangular matrices will be denoted by $T_k$ (instead of $R_k$ for the general $K$-JET):
    \begin{align}
    \label{eq:K-GMD}
      A_k &= U_k T_k V^\dagger \,, \quad k=1,\ldots,K \,.
    \end{align}
\end{defn}
\begin{remark}
     For $K=1$, $1$-GMD reduces to the GMD of \eqref{eq:gmd}.
\end{remark}
The proof of the existence of a JET of two matrices $A_1$ and $A_2$ \cite{STUD:SP} is based
upon applying the GMD \eqref{eq:gmd} to the (single) matrix  $A_1 A_2^{-1}$.
This technique is generalized for more matrices in the next lemma.

\begin{lemma}[Equivalence of Square K-GMD and (K+1)-JET]
\label{lem:k_to_k_plus_one_square}
Let $A_1,\ldots,A_{K+1}$ be $n \times n$ full-rank complex-valued matrices with equal determinants, and
define the $K$ matrices:
\begin{align} \label{eq:definition_of_a_square}
			B_k = A_k A_{K+1}^{-1}  \,, \qquad k=1,\ldots,K \,.
\end{align}
Then the following two statements are equivalent:
\begin{enumerate}
\item There exist $K+1$
unitary matrices $U_1,\ldots,U_{K},U_{K+1}$, of dimensions $n \times n$, such that
\begin{align} \label{eq:jet_of_a_square}
	U_k^\dagger B_k U_{K+1} = T_k \,,\qquad  k=1,\ldots,K \,,
\end{align}
where $\left\{ T_k \right\} $ are $n \times n$ upper triangular with all diagonal entries equal to $1$.
\item There exist $K+2$ unitary matrices
$U_1,\ldots,U_{K+1},V$, of dimensions $n \times n$, such that
\begin{align}
	U_k^\dagger A_k V = R_k \,,\qquad k=1,\ldots,K+1 \,,
\end{align}
where $\left\{ R_k \right\}$ are $n \times n$ upper triangular with equal diagonals,
as in \eqref{eq:R_equal_diagonals}.
\end{enumerate}
\end{lemma}

\begin{proof}
First, assume that statement 2 holds.
Thus, there exist $K+2$ unitary matrices $U_1,\ldots,U_{K+1},V$, of dimensions $n \times n$, such that
\begin{align}
	U_k^\dagger A_k V = R_k \,,\qquad k=1,\ldots,K+1 \,,
\end{align}
where $\left\{ R_k \right\}$ are $n \times n$ upper triangular with equal diagonals.
This implies that 
\begin{align}
	U_k^\dagger B_k U_{K+1} &=
      U_k^\dagger A_k A_{K+1}^{-1} U_{K+1}
\\
      &=
      U_k^\dagger A_k V V^\dagger A_{K+1}^{-1} U_{K+1}
\\
      &=
      R_k  R_{K+1}^{-1}
\\
      &=
	T_k \,,
\end{align}
where $T_k$ is upper triangular with all the diagonal elements equal to $1$,
which results in statement 1.
\\ \indent
Now, assume that statement 1 holds.
Perform the QR decomposition on the matrix $A_{K+1}^{-1} U_{K+1}$:
\begin{align}
	A_{K+1}^{-1} U_{K+1} = V R  \,,
\end{align}
where $V$ is a unitary matrix of dimensions $n \times n$, and $R$ is an \mbox{$n \times n$}  upper triangular matrix.
Thus, substituting \eqref{eq:definition_of_a_square}, we obtain the following equalities:
\begin{align}
  U_k^\dagger A_k V R & = U_k^\dagger A_k A_{K+1}^{-1} U_{K+1}
\\*
 & = U_k^\dagger B_k U_{K+1}  \,, \qqquad k=1,\ldots,K \,,
\end{align}
which, according to \eqref{eq:jet_of_a_square}, is equal to
\begin{align} \label{eq:miffy1_square}
&   U_k^\dagger A_k V R = T_k  \,, \qquad \qquad k=1,\ldots,K \,.
\end{align}
On the other hand, we have
\begin{subequations}
\label{eq:miffy2_square}
\noeqref{eq:miffy2_square:1,eq:miffy2_square:2}
\begin{align}
 U_{K+1}^\dagger A_{K+1} V R    &= U_{K+1}^\dagger A_{K+1} A_{K+1}^{-1} U_{K+1} 
\label{eq:miffy2_square:1}
\\*
 &= U_{K+1}^\dagger U_{K+1}  =  I \,.
\label{eq:miffy2_square:2}
\end{align}
\end{subequations}
Multiplying \eqref{eq:miffy1_square} and \eqref{eq:miffy2_square} by $R^{-1}$ on the right yields:
\begin{align}
&    U_k^\dagger A_k V         = T_k R^{-1}  \,, \qquad k=1,\ldots,K \\*
&    U_{K+1}^\dagger A_{K+1} V = R^{-1} \,.
\end{align}
Since $T_k$ are upper triangular with only $1$s on the diagonal, the matrices
$ R_k  \triangleq T_k R^{-1} $  ($k=1,\ldots,K$) and
$ R_{K+1} \triangleq R^{-1} $
have equal diagonals, which completes the proof.
\end{proof}

\begin{remark}
As a consequence of \lemref{lem:k_to_k_plus_one_square}, if it is possible to perform $K$-GMD
on \emph{any} $K$ full rank square matrices having the same determinant, then it is also possible to perform $(K+1)$-JET
on \emph{any} $K+1$ full rank square matrices of the same dimensions and the same determinant, and vice versa.
In particular, since $1$-GMD is always possible, it is also always possible
to perform $2$-JET on \emph{any} two full rank square matrices of the same dimensions and equal determinants.
\end{remark}

\begin{remark}
  The condition of equal determinants in Definitions \ref{defn:K-JET} and \ref{defn:K-GMD}
  may be replaced with a slightly weaker condition of equal absolute values of the determinants, i.e.,
  \begin{align}
    \left| \det(A_1)  \right| = \left| \det(A_2)  \right| = \cdots = \left| \det(A_K)  \right| \,.
  \end{align}
  This is easily achieved by multiplying by additional diagonal phase matrices on the left in \eqref{eq:K-JET} and \eqref{eq:K-GMD}.
\end{remark}

%% file: jet_interpretation.tex
Following the geometric interpretation of the GTD in \secref{ss:GeometricInterpretationGTD},
we give a geometric interpretation of the JET for the special case of $2 \times 2$ matrices:
\begin{align}
    &A_1 \triangleq
    \begin{pmatrix}
        \ba^{(1)} & \bb^{(1)}
    \end{pmatrix}
    =
    \left( \begin{array}{cc}
            a_x^{(1)} & b_x^{(1)} \\
            a_y^{(1)} & b_y^{(1)}
         \end{array}
    \right)
 \\
    &A_2 \triangleq
    \begin{pmatrix}
        \ba^{(2)} & \bb^{(2)}
    \end{pmatrix}
    =
    \left( \begin{array}{cc}
            a_x^{(2)} & b_x^{(2)} \\
            a_y^{(2)} & b_y^{(2)}
         \end{array}
    \right)
    \,,
\end{align}
where $\ba^{(i)}$ and $\bb^{(i)}$ are the first and second columns of $A_i$ ($i=1,2$), respectively.
The interpretation for the general case is a simple extension of the $2 \times 2$ case.
As in \secref{ss:GeometricInterpretationGTD}, we assume, without loss of generality, that 
$\det(A_1) = \det(A_2) = 1$.

By multiplying both matrices $A_1$ and $A_2$ on the right by the same rotation matrix $V$~\eqref{eq:gtd_interprtation:V},
we obtain ($i=1,2$)
\col{
\begin{align}
  &A_i V
  \triangleq 
  \begin{pmatrix}
    \tba^{(i)} & \tbb^{(i)}
  \end{pmatrix}
  =
  \\
  &\!\!\begin{bmatrix}
      a_x^{(i)} \cos\theta_r+b_x\sin\theta_r &
      a_x^{(i)} \cos(\theta_r+\frac{\pi}{2}) + b_x^{(i)} \sin(\theta_r+\frac{\pi}{2}) \\
      a_y^{(i)} \cos\theta_r+b_y\sin\theta_r &
      a_y^{(i)} \cos(\theta_r+\frac{\pi}{2}) + b_y^{(i)} \sin(\theta_r+\frac{\pi}{2})
  \end{bmatrix}
  \!\!.
\label{eq:JET_interp:AV}
\end{align}
}
{
\begin{align}
  A_i V 
  &\triangleq \begin{pmatrix}
    \tba^{(i)} & \tbb^{(i)}
  \end{pmatrix}
  \\
  &=
  \begin{bmatrix}
      a_x^{(i)} \cos\theta_r+b_x\sin\theta_r &
      a_x^{(i)} \cos(\theta_r+\frac{\pi}{2}) + b_x^{(i)} \sin(\theta_r+\frac{\pi}{2}) \\
      a_y^{(i)} \cos\theta_r+b_y\sin\theta_r &
      a_y^{(i)} \cos(\theta_r+\frac{\pi}{2}) + b_y^{(i)} \sin(\theta_r+\frac{\pi}{2})
  \end{bmatrix}
   \,.
\label{eq:JET_interp:AV}
\end{align}
}
That is, 
we obtain two ellipses of equal area (absolute value of determinant), 
centered at the origin (see \figref{subfig:fig_jet_o}).
The norms of the first column vectors in \eqref{eq:JET_interp:AV}, 
$\tba^{(1)}$ and $\tba^{(2)}$, 
are $2\pi$-cyclic continuous functions of $\theta_r$.
Thus, using the intermediate value theorem, 
there exists an angle $\theta_r$ (and in fact, four such angles per cycle)
for which the norms of $\tba^{(1)}$ and $\tba^{(2)}$ are equal, 
as illustrated in \figref{subfig:fig_jet_v}.

Multiplying each of the resulting matrices, $A_i V$,
on the left, by an appropriate rotation matrix $U_i^\dagger$, 
where
\begin{align}
  U_i&=\left( \begin{array}{cc}
	    \cos\theta^{(i)}_\ell & -\sin\theta^{(i)}_\ell \\
	    \sin\theta^{(i)}_\ell &  \cos\theta^{(i)}_\ell
	  \end{array}
    \right)\,,
\end{align}
rotates both column vectors of $A_i V$ by the same 
angle, $\left( -\theta_\ell^{(i)} \right)$, without altering their norms.
Thus, by choosing $\theta_\ell^{(i)}$, such that 
$U_i^\dagger \tba^{(i)}$ are aligned with the $x$-axis, for both $i=1,2$, 
we achieve the desired decomposition, 
as depicted in Figures \ref{subfig:fig_jet_u1} and \ref{subfig:fig_jet_u2}.

\col{
\begin{figure*}[th]
  \centering
  \subfloat[All possible column vectors of $A_i V$, where $V$ is a rotation matrix; the original column vectors (for $V=I$) are depicted explicitly.]
  {\label{subfig:fig_jet_o}
    \centering
    \epsfig{file = 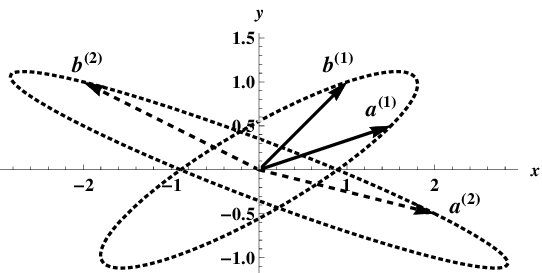, width = .75\columnwidth}
  } \qquad
  \subfloat[Right rotation by $\theta_r \approx 1.34$, for which the resulting vectors $\tba^{(1)}$ and $\tba^{(2)}$ have equal norms.]{\label{subfig:fig_jet_v}
    \centering
    \epsfig{file = 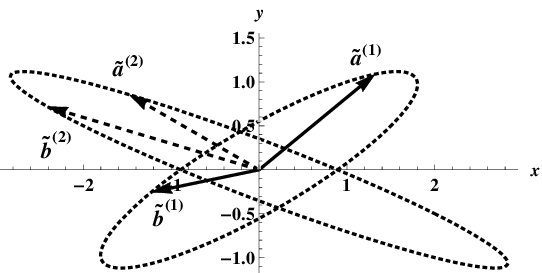, width = .75\columnwidth}
  }
\\
  \subfloat[Left rotation by $\theta^{(1)}_\ell \approx 0.69$ of  the first matrix, for which the first vector is aligned with the $x$-axis.]{\label{subfig:fig_jet_u1}
    \centering
    \epsfig{file = 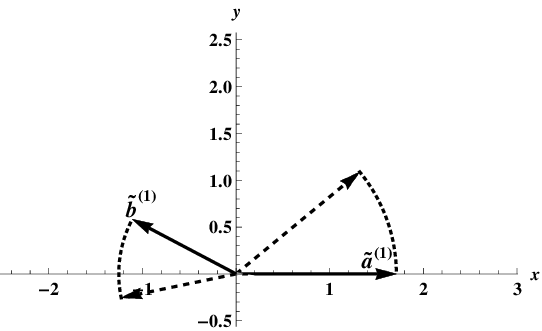, width = .75\columnwidth}
  } \qquad
  \subfloat[Left rotation by $\theta^{(2)}_\ell \approx 2.62$ of  the first matrix, for which the first vector is aligned with the $x$-axis.]{\label{subfig:fig_jet_u2}
    \centering
    \epsfig{file = 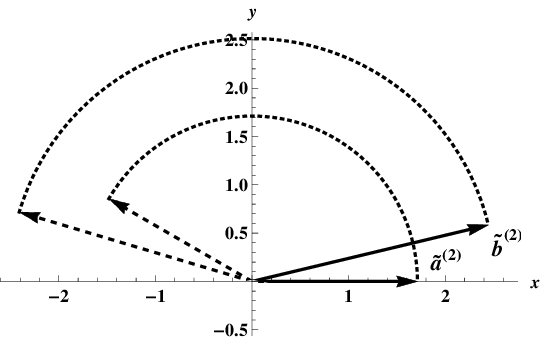, width = .75\columnwidth}
  }
  \caption{JET for $\ba^{(1)} = (3/2,1/2)^\dagger\,, \bb^{(1)}=(1,1)^\dagger\,, \ba^{(2)}=(2,-0.5)^\dagger\,,\bb^{(2)}=(-2,1)^\dagger$}
\label{fig:jet}
\end{figure*}
}
{
\begin{figure}[t]
  \centering
  \subfloat[All possible column vectors of $A_i V$, where $V$ is a rotation matrix; the original column vectors (for $V=I$) are depicted explicitly.]
  {\label{subfig:fig_jet_o}
    \centering
    \epsfig{file = 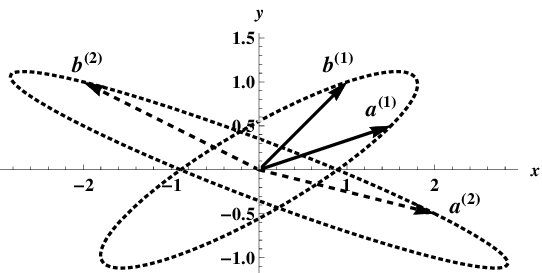, scale=0.8}
  } \qquad
  \subfloat[Right rotation by $\theta_r \approx 1.34$, such that the resulting vectors $\tba^{(1)}$ and $\tba^{(2)}$ have equal norms.]{\label{subfig:fig_jet_v}
    \centering
    \epsfig{file = 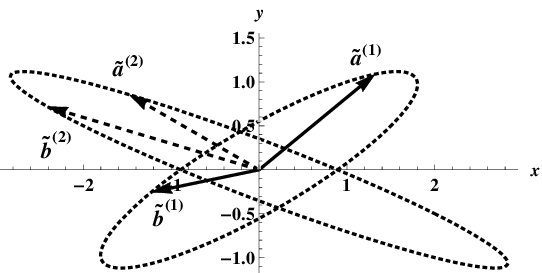, scale=0.8}
  }
\\
  \subfloat[Left rotation by $\theta^{(1)}_\ell \approx 0.69$ of  the first matrix, such that the first vector is aligned with the $x$-axis.]{\label{subfig:fig_jet_u1}
    \centering
    \epsfig{file = 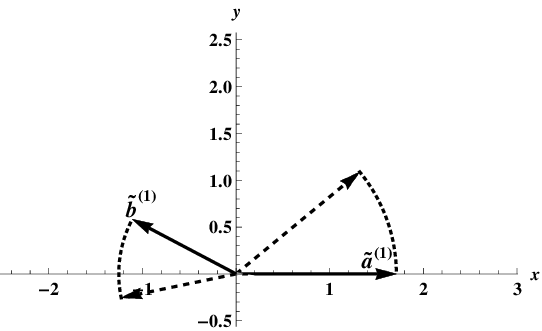, scale=0.8}
  } \qquad
  \subfloat[Left rotation by $\theta^{(2)}_\ell \approx 2.62$ of  the first matrix, such that the first vector is aligned with the $x$-axis.]{\label{subfig:fig_jet_u2}
    \centering
    \epsfig{file = 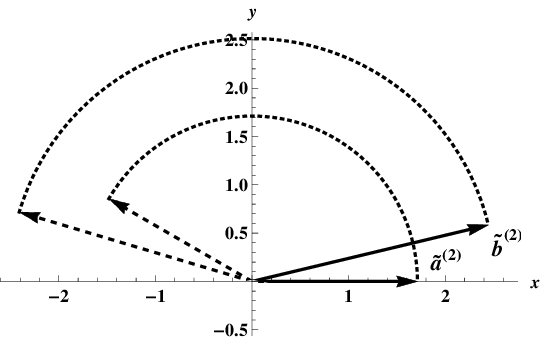, scale=0.8}
  }
  \caption{JET for $\ba^{(1)} = (3/2,1/2)^\dagger\,, \bb^{(1)}=(1,1)^\dagger\,, \ba^{(2)}=(2,-1)^\dagger\,,\bb^{(2)}=(-1,1)^\dagger$}
\label{fig:jet}
\end{figure}
}

\begin{remark}
    JET of more than two matrices is not possible, in general.
    This may be seen in the $2 \times 2$ case, that while every two ellipses must intersect for some value of $\theta_r$, due to the intermediate value theorem, 
    there is no hope for simultaneous intersection of more trajectories.
\end{remark}

%% file: multi_user_scheme.tex
The scheme of \secref{ss:p2p_scheme} can be generalized for the $K$-user BC channel
\eqref{MIMO-BC}  in a straightforward manner,
by replacing the GTD \eqref{eq:gtd} with the $K$-JET \eqref{eq:K-JET}.

Let $\CCC$ be an admissible covariance matrix.
As will be explained in \remref{rem:equal_det},
 we can assume without loss of generality that
\mbox{$\Capacity{1}=\cdots=\Capacity{K}$}.
The following scheme achieves the rate
$\Capacity{i}$.
Therefore, the common-message BC capacity \eqref{eq:MIMO_BC_capacity} can be achieved by an
appropriate choice of the matrix $\CCC$.

Applying Definitions \ref{def:augmented}  and \ref{def:canonical} we
define
\begin{align}
\label{eq:tilde_G_K_users}
	\tH_k &\triangleq
	\left(\begin{array}{c}
		  H_k \CCC^{1/2} \\ I_{n_t}
	      \end{array}
	\right)\,,
    \\
    \tH_k &= Q_k G_k \,, \qquad k=1,\ldots,K \,,
\label{eq:G_matrix_K_users}
\end{align}
where $I_{n_t}$ is the $n_t \times n_t$ identity matrix, $\CCC$ is any admissible covariance matrix,
the matrices $\tH_k$ are the augmented channel matrices,
$Q_k$ are $(n_r^{(k)}+n_t) \times n_t$ matrices with orthonormal columns, and $G_k$ are the canonical channel matrices of dimensions $n_t \times n_t$ and are upper triangular with real positive diagonal elements.

Now, assume that there exists a $K$-JET of the matrices $G_k$:
\begin{align}
    G_k &= U_k R_k V^\dagger \,, \qquad k=1,\ldots,K \,,
\end{align}
where $R_k$ are upper triangular matrices whose diagonal values are equal to $r_1,\ldots,r_{n_t}$.
Then, the same transmission scheme as in \secref{ss:p2p_scheme} may be employed, where in step 3 the $k$-th receiver uses the matrices $Q_k$ and $U_k$ in \eqref{eq:scheme_receiver}.

\begin{remark}
As in \remref{rem:gmd_same_codebook}, if the $K$-JET in the above scheme is also a $K$-GMD \eqref{eq:K-GMD}, then the capacity \eqref{eq:MIMO_BC_capacity} can be achieved using \emph{the same} scalar codebook over all scalar sub-channels.
\end{remark}

\begin{remark}
\label{rem:equal_det}
    Consider the case where, for the optimal input covariance matrix $C_\bx$, the mutual informations to the different users, $\{ I(H_k, C_\bx) \}$, are not all equal.
    In this case, the common-message BC capacity \eqref{eq:MIMO_BC_capacity} is limited to the minimum of these mutual informations:
    \begin{align}
        C \left( \left\{ H_k \right\}_{k=1}^K, P \right) &= \max_{C_\bx} \min_{k=1,\dots,K} I(H_k,C_\bx) \,.
    \end{align}
    Rewriting these mutual informations in terms of the channel canonical matrices $\{G_k\}$:
    \begin{align}
        I(H_k,C_\bx) &= \log \det \left( I + H_k \CCC H_k^\dagger \right)  \\
        &= \log \det \left( I + H_k^\dagger \CCC^{1/2 \dagger} \CCC^{1/2} H_k \right)  \\
        &= \log \det \left(\tH_k^\dagger \tH_k \right) \\
        &= \log \det \left( \left(Q_k G_k\right)^\dagger Q_k G_k  \right) \\
        &= \log \det \left( G_k^\dagger Q_k^\dagger Q_k G_k  \right) \\
        &= \log \det \left( G_k^\dagger  G_k  \right) \\
        &=2 \log |\det{(G_k)}|
        \,,
    \end{align}
    we have
    \begin{align}
        C \left( \left\{ H_k \right\}_{k=1}^K, P \right)
    &=  2 \log \min_{k=1,\dots,K}  \det{(G_k)} \,,
    \end{align}
    where the absolute value operation may be dropped as explained in \remref{rem:greater_than_1_diag}.

    Thus, the common-message BC capacity is dictated by the user having the minimal $\det(G_k)$.

    Applying K-JET to the matrices $\{ G_k \}$, results in \emph{proportional} diagonal elements (in contrast to the equal diagonals resulting when all mutual informations are equal; see \remref{rem:JET:equal_det}). Since these effective diagonal entries correspond to the effective SNRs of the effective scalar sub-channels observed by each user,
    this implies, in turn, that the users having larger mutual information have larger effective SNRs.
    However, since the common-message BC capacity is limited to the minimum of the mutual informations, the excess SNRs of the users with larger mutual informations (and $\det(G_k)$) has no effect on achievable rate. 

    This ``bottleneck phenomenon'' is illustrated in the following example.
\end{remark}
\vspace{.5\baselineskip}

\begin{exmpl}[\exref{ex:non_equal_dets} Continued]
    Consider the two channel canonical matrices $G_1$ and $G_2$ (replacing $A_1$ and $A_2$ in \exref{ex:non_equal_dets}).
    \begin{align}
    G_1 &=
    \begin{pmatrix}
        2 & 1 \\
        0 & 8
    \end{pmatrix}
    = 4
    \begin{pmatrix}
        0.5 & 0.25 \\
        0 & 2
    \end{pmatrix}
    \,,&
    \det \left( G_1 \right) &= 16 \,,\\
    G_2 &=
    \begin{pmatrix}
        5 & -2 \\
        0 & 5
    \end{pmatrix}
    = 5
    \begin{pmatrix}
        1 & -0.4 \\
        0 & 1
    \end{pmatrix}
    \,,&
    \det \left( G_2 \right) &= 25 \,.
    \end{align}
By applying JET to ${G_1}$ and ${G_2}$ we obtain
     \begin{align}
	R_1 = 4
	\left( \begin{array}{cc}
	  1.20 & -1.48 \\
	  0 &  0.84
        \end{array} \right) 
        \\
	R_2 = 5
	\left( \begin{array}{cc}
	  1.20 & -0.17 \\
	  0 &  0.84
        \end{array} \right) 
    \end{align}
    The corresponding common-message BC capacity is, therefore, 
    \begin{align}
    \label{eq:rem:equal_det:capacity}
        C &= 2 \log \min \big\{ \det \left( G_1 \right), \det \left( G_2 \right) \big\} \\
           &\approx \underbrace{2 \log \left( 4 \times 1.20 \right)}_{\text{Rate of stream 1}}+ \underbrace{2 \log \left( 4 \times 0.84 \right)}_{\text{Rate of stream 2}}
        \approx 8 \,.
    \end{align}
Thus, the rates of the two streams are dictated by user 1, whereas user 2 has excess effective SNR in each of the streams.
\end{exmpl}

\begin{remark}[Decoding Order]
 Recall that in the single-user case, there is no loss (in terms of achievable rates) in restricting attention to upper triangular decomposition at the receiver, since any ordering can be represented as a permutation of the matrix $R$ in \eqref{eq:Ggtd}, namely,
\begin{align}
\label{eq:Ggtd_with_perms}
      G = U \Pi R \Pi^\dagger V^\dagger \,,
\end{align}
where $\Pi$ is a permutation matrix.
Since permutation matrices are unitary, \eqref{eq:Ggtd_with_perms} falls under the
framework \eqref{eq:Ggtd} without permutations.
In the multi-user case, on the other hand, each receiver can choose a different decoding order,
which implies that the different permutation matrices cannot be absorbed in the (single) matrix $V$.
 Hence, there is a loss of generality in the proposed scheme. This restriction is removed in \secref{ss:upper_lower}.
\end{remark}

%% file: multi_user_perfect_two_by_two.tex
In this section we provide necessary  and sufficient conditions for the existence of $2$-GMD for $2 \times 2$ matrices.
The conditions are stated in the following theorem.
As explained in \remref{rem:equal_det}, we can assume without loss of generality that both matrices have determinants equal to $1$.
According to \lemref{lem:k_to_k_plus_one_square}, this also provides a necessary and sufficient condition for the existence of a $3$-JET for $2 \times 2$ matrices.

\begin{thm}[2-GMD for $2 \times 2$ Matrices]
\label{thm:perfect_two_on_two_complex}
    Let $A_1$ and $A_2$ be \emph{complex-valued} $2 \times 2$ matrices with determinants equal to $1$.
    Then, there exist complex-valued $2 \times 2$ unitary matrices $U_1,U_2,V$ such that:
    \begin{align} \label{eq:two_by_two_gmd}
          U_k^\dagger A_k V =  \left( \begin{array}{cc}
    		      1 & * \\ 0 & 1
                       \end{array} \right) \,, \quad k=1,2 \,,
    \end{align}
    if and only if the following inequality is satisfied:
    \begin{align} \label{eq:condition3}
          F_1 \left(  A_1^\dagger A_1 - I ,  A_2^\dagger A_2 - I  \right) \geq 0 \,,
    \end{align}
    where
    \begin{align} \label{eq:definition_of_F}
        F_1(S_1,S_2)  \triangleq \det \big( S_1 \adj (S_2) - S_2 \adj (S_1)    \big) \,,
    \end{align}
    and $*$ represents an arbitrary value value (which may differ between the two matrices).
\end{thm}

\begin{remark}
    Even in the case where the  matrices $A_1$ and $A_2$ are real-valued, the resulting unitary matrices $U_1,U_2,$ and $V$ are, in general, \emph{complex-valued}. In fact, if $A_1,A_2$ are real valued, then it can be easily shown that the matrices $U_1,U_2$ and $V$ are real-valued if and only if \eqref{eq:condition3} holds \emph{with equality}.
    In  \secref{ss:restatement} we show how to obtain a communication scheme that involves only real-valued orthogonal transformations, under the same condition \eqref{eq:condition3}, using a space--time structure.
\end{remark}

\begin{remark}
    Using this theorem, a sufficient and necessary condition for the existence of a $2$-GMD for two $3 \times 3$ \emph{diagonal} matrices can be derived. The method of this derivation
    is demonstrated via an example of the ``rateless'' problem with three rates in 
    \ver{\secref{sss:rateless_three_rates}}{\cite[Sec.~V-E2]{JET:TechReport:SeveralUsers2013}}.
\end{remark}

The following lemma, the proof of which is given in \appref{app:det_adj_complex_proof}, 
will be used in the proof of the theorem.
\begin{lemma} \label{lem:det_adj_complex}
  Let $S_1$ and $S_2$ be complex-valued Hermitian \mbox{$2 \times 2$} matrices.
Then, there exists a complex-valued vector \mbox{$\bv \in \mathbb{C}^2$}, such that
\begin{align}
 \bv^\dagger S_1 \bv &= 0 \\
 \bv^\dagger S_2 \bv &= 0 \\
 \bv &\neq 0 \,,
\end{align}
if and only if the following conditions hold:
\begin{subequations}
\label{eq:lemma_tnayal}
\noeqref{eq:lemma_tnayal:0,eq:lemma_tnayal:1,eq:lemma_tnayal:2}
\begin{align}
\label{eq:lemma_tnayal:0}
      \det(S_1) & \leq 0 \\
\label{eq:lemma_tnayal:1}
      \det(S_2) & \leq 0 \\
\label{eq:lemma_tnayal:2}
      F_1 \left(  S_1 , S_2  \right) & \geq 0 \,,
\end{align}
\end{subequations}
where $F_1$ is defined as in \eqref{eq:definition_of_F}.
\end{lemma}
\vspace{\baselineskip}

\begin{proof}[Proof of \thrmref{thm:perfect_two_on_two_complex}]
\input{two_by_two_theorem_proof}
\end{proof}

\begin{corol}
\thmref{thm:perfect_two_on_two_complex} can easily be generalized as follows: for any $r>0$, there exist three complex-valued $2 \times 2$ unitary matrices $U_1,U_2,V$ such that:
\begin{align}
      U_k^\dagger A_k V =  \left( \begin{array}{cc}
		      r & * \\ 0 & 1/r
                   \end{array} \right) \,, \quad k=1,2 \,,
\end{align}
if and only if the following conditions are satisfied:
\begin{align}
      \det \left( A_1^\dagger A_1 - r^2 I \right)  & \leq 0 \\
      \det \left( A_2^\dagger A_2 - r^2 I \right)  & \leq 0 \\
      F_1 \left(  A_1^\dagger A_1 - r^2 I ,  A_2^\dagger A_2 - r^2 I  \right) & \geq 0 \,.
\end{align}
The proof of the corollary follows along the same line as that of
\thrmref{thm:perfect_two_on_two_complex} with obvious modifications.
\end{corol}

\ver{}{
We next present a special class of channel matrices, for which perfect JET of more than three matrices is possible.
Another interesting special class of matrices stemming from the Gaussian rateless problem can be found in \cite[Sec.~V-E]{JET:TechReport:SeveralUsers2013} and \cite{ETW_GaussianRateless,JET:RatelessITW2011,ayal_thesis}.\footnote{A numerical derivation of the precoding matrix $V$ in the case of a rateless code (even for parameters for which a perfect decomposition is not possible) is available at \cite{ETW_Matlab}.}
}

%% file: two_by_two_theorem_proof.tex
Let $V$ be a $2 \times 2$ unitary matrix, and denote by
$\bv_1$ and $\bv_2$ the first and second columns of $V$, respectively.
Note that
\begin{align}
	A_k V = \left( A_k \bv_1 \middle| A_k \bv_2   \right)  \,, \quad k=1,2 \,.
\end{align}
We now perform the QR decomposition on the above matrices:
\begin{subequations}
\label{eq:av_is_triangular}
\noeqref{eq:av_is_triangular:1,eq:av_is_triangular:2}
\begin{align}
\label{eq:av_is_triangular:1}
  A_1 V &= U_1 T_1 \\
\label{eq:av_is_triangular:2}
  A_2 V &= U_2 T_2  \,,
\end{align}
\end{subequations}
where $U_1,U_2$ are unitary and $T_1,T_2$ are upper triangular.
Since we have
\begin{align}
    T_k = U_k^\dagger A_k V = \left( \begin{array}{c|c}
				     U_k^\dagger A_k v_1 & U_k^\dagger A_k v_2
                              \end{array} \right) \,,
\end{align}
and the norm of $A_k v_1$ equals that of $U_k^\dagger A_k v_1$,
the upper-left element of $T_1$ and $T_2$ is equal to $1$,
\begin{align} \label{eq:T_up_left}
      T_k = \left( \begin{array}{cc} 1 & * \\ 0 & * \end{array}
\right) \,, \quad k=1,2 \,,
\end{align}
 \emph{if and only if}:
\begin{align}
  \left\| A_1 \bv_1 \right\| &= 1 \\*
  \left\| A_2 \bv_1 \right\| & = 1 \,.
\end{align}
Also, since $V$ is required to be unitary, the norm of $\bv_1$ must equal $1$:
\begin{align}
      \left\|  \bv_1 \right\| = 1 \,.
\end{align}
Note that for every $\bv_1$,
we can choose a unit-norm vector $\bv_2$ that spans the subspace orthogonal to $\bv_1$,
thus constructing a unitary matrix $V$.
Also, since $V$ is unitary, $\det(A_1 V) = \det(A_2 V) = 1$, and therefore from \eqref{eq:T_up_left}
it follows that the bottom-right element also equals $1$.

Combining the above observations, it follows that 
there exists a $2 \times 2$ unitary matrix $V$
such that the decomposition \eqref{eq:av_is_triangular} is possible,
where $T_1,T_2$ have only $1$s on their diagonals,
if and only if the first column of $V$, denoted by $\bv_1$, satisfies the
following three equations:
\begin{align}
  \bv_1^\dagger A_1^\dagger A_1 \bv_1 &= 1 \\
  \bv_1^\dagger A_2^\dagger A_2 \bv_1 &= 1 \\
  \bv_1^\dagger \bv_1 & = 1 \,,
\end{align}
or equivalently,
\begin{align} \label{eq:vector_tnayal}
\begin{aligned}
  \bv_1^\dagger (A_1^\dagger A_1 - I) \bv_1 &= 0 \\
  \bv_1^\dagger (A_2^\dagger A_2 - I) \bv_1 &= 0 \\
  \bv_1^\dagger \bv_1 & = 1 \,.
\end{aligned}
\end{align}
Note that since $\det(A_1)=\det(A_2)=1$, we have
\begin{align}
  \det \left( A_1^\dagger A_1 - I \right) & \leq 0 \\
  \det \left( A_2^\dagger A_2 - I \right) & \leq 0 \,.
\end{align}
Using this result along with the result of \lemref{lem:det_adj_complex} with \mbox{$S_k = A_k^\dagger A_k - I$}, proves the theorem.

%% file: multi_user_rateless.tex
We now consider the problem of constructing scalar Gaussian rateless codes, treated in \cite{ETW_GaussianRateless}.\footnote{A numerical derivation of the precoding matrix $V$ in the case of a rateless code (even for parameters for which a perfect decomposition is not possible) is available in \cite{ETW_Matlab}.}
The constructed codes are designed for a complex AWGN channel,
\begin{align} \label{eq:rateless_blocks}
      \by_l = \alpha \bx_l + \bz_l\,, \qquad l = 1,2,\ldots \,,
\end{align}
where $\alpha$ is a channel gain that varies from receiver to receiver, $\bx_l$
is the channel input vector of $M$ symbols, $\bz_l$ is a noise vector of $M$ i.i.d.\ complex Gaussian random variables, each of variance $1$, and $\by_l$ is the vector of $M$ channel output symbols.
The channel input is average-power limited, without loss of generality, to power $1$.

We assume that $\alpha$ can take one of $K$ possible values, such that  a gain of $\alpha_k$ implies that the
message should be decodable using only the first $k$ received blocks.\footnote{Alternatively, this can be viewed as a scheme that works for every value of $\alpha$, but designed to be optimal only for $K$ specific values.}
The gains are such that, for any value of $k$, the total capacity is the same:
\begin{align}
	C = k \log (1 + | \alpha_k |^2 ) \,, \qquad k=1,2,\ldots,K \,.
\end{align}
This implies that the compound capacity is achieved by a white input distribution.

The scheme proposed in \cite{ETW_GaussianRateless} consists of dividing the information message into $L$ sub-messages (``layers''), encoding each sub-message using a (fixed-block) codebook, designed for a scalar AWGN channel, and sending in each block some linear combination of those codewords. In the sequel we will consider only the case where $K=L$, i.e., the number of codewords used by the scheme is equal to the highest possible number of blocks received by the receiver.

Alternatively, this problem can be viewed as a $K$-user MIMO common-message BC problem, as follows:
the $K$ transmission blocks \eqref{eq:rateless_blocks} can be considered as a single transmission
over a Gaussian MIMO channel, with channel matrix
\begin{align}
	H = \left(
  \begin{array}{cccc}
	\alpha_k & 0 & \cdots & 0 \\
	0 & \alpha_k & \cdots & 0 \\
	\vdots & \vdots & \ddots & \vdots \\
	0 & 0 & \cdots & \alpha_k
  \end{array}
\right) \,.
\end{align}
Since the $k$-th user is allowed to use only the first $k$ blocks,
this is equivalent to removing the last $K-k$ rows from the corresponding channel matrix, namely, the channel
matrix of the $k$-th user becomes:
\begin{align} \label{eq:H_k}
	H_k = \left(
\overbrace{
  \begin{array}{cccc}
	\alpha_k & 0 & \cdots & 0\\
	0 & \alpha_k & \cdots & 0\\
	\vdots & \vdots & \ddots & \vdots \\
	0 & 0 & \cdots & \alpha_k
  \end{array}}^k
\overbrace{
  \begin{array}{ccc}
	0 & \cdots & 0\\
	0 & \cdots & 0\\
	\vdots & \vdots & \vdots \\
	0 & \cdots & 0
  \end{array}}^{K-k}
\right) \,.
\end{align}
Since the capacity-achieving distribution in this problem is white, this translates to
an input covariance matrix which is a scaled identity matrix. Namely, $\CC=I$.

Alternatively, the channel matrix of the $k$-th user can be viewed as a square $K \times K$ diagonal matrix,
where the last $K-k$ diagonal elements are forced to be zeros:
\begin{align} \label{eq:H_k_alternative}
	H_k = \left(
\overbrace{
  \begin{array}{cccc}
	\alpha_k & 0 & \cdots & 0\\
	0 & \alpha_k & \cdots & 0\\
	\vdots & \vdots & \ddots & \vdots \\
	0 & 0 & \cdots & \alpha_k \\
	0 & 0 & \cdots & 0 \\
\vdots & \vdots & \ddots & \vdots \\
0 & 0 & \cdots & 0
  \end{array}}^k
\overbrace{
  \begin{array}{ccc}
	0 & \cdots & 0\\
	0 & \cdots & 0\\
	\vdots & \vdots & \vdots \\
	0 & \cdots & 0 \\
	0 &  \cdots & 0 \\
 \vdots & \ddots & \vdots \\
0 & \cdots & 0
\end{array}}^{K-k}
\right) \,.
\end{align}
This alternative representation yields the same results as the representation \eqref{eq:H_k}.

We now recover the results of \cite{ETW_GaussianRateless}, giving explicit constructions for $K=2,3$.

\subsubsection{Two Rates ($K=2$)}

Specializing the problem to the case of one (possible) incremental redundancy block ($K=2$), the
two channel matrices are (same as $H_1$ and $H_3$ in \exref{DOF-example})
        \begin{align}
            H_1 = \left(  \begin{array}{cc} \alpha_1 & 0  \end{array}\right)
        	\,,\quad
            H_2 = \left(  \begin{array}{cc} \alpha_2 & 0  \\ 0  & \alpha_2 \end{array}\right) \,,
        \end{align}	
where $\alpha_1,\alpha_2$ are values satisfying
\begin{align}
\log(1 + \left| \alpha_1 \right|^2) = 2\log(1 + \left| \alpha_2 \right|^2) = C \,.
\end{align}
Applying the scheme of \secref{ss:multi_user_scheme} yields the following precoding matrix \cite{JET:RatelessITW2011}:
        \begin{align}
            V =
            \sqrt{\frac{1}{2^{C/2}+1}}
            \left(
              \begin{array}{cc}
                1 & 2^{C/4} \\
                2^{C/4} & -1 \\
              \end{array}
            \right)  \,,
        \end{align}
which coincides with the result in~\cite[Sec.~III]{ETW_GaussianRateless}.

\subsubsection{Three Rates}
\label{sss:rateless_three_rates}
The case of $K=3$ was also treated in \cite{ETW_GaussianRateless},
where a condition for which a ``perfect'' scheme exists was derived.
We will now shed light on this condition.

Again, representing the problem as a MIMO common-message BC one, the three possible channel matrices are:
\begin{align}
H_1 &= \left( \begin{array}{ccc}
	  \alpha_1 & 0 & 0
       \end{array}
\right)
\,,
\\
H_2 &= \left( \begin{array}{ccc}
	  \alpha_2 & 0 & 0 \\
	  0 &  \alpha_2  & 0
       \end{array}
\right)
\,, \\
H_3 & = \left( \begin{array}{ccc}
	  \alpha_3 & 0 & 0 \\
	  0 &  \alpha_3  & 0 \\
	  0 & 0 &  \alpha_3
       \end{array}
\right)\,,
\end{align}
where $\alpha_1,\alpha_2,\alpha_3$ are values satisfying
\begin{align}
\log(1 + \left| \alpha_1 \right|^2) = 2\log(1 + \left| \alpha_2 \right|^2) = 3 \log(1 + \left| \alpha_3 \right|^2) = C \,.
\end{align}
The channel canonical matrices, as defined in  \eqref{eq:G_matrix_K_users}, are:
\begin{align}
G_1 &=  \left( \begin{array}{ccc}
	  2^\frac{C}{2} & 0 & 0 \\
	  0 &  1  & 0 \\
	  0 & 0 & 1
       \end{array}
\right)
\,,
\\
G_2 &=  \left( \begin{array}{ccc}
	  2^\frac{C}{4} & 0 & 0 \\
	  0 &  2^\frac{C}{4}  & 0 \\
	  0 & 0 & 1
       \end{array}
\right)
\,,
\\
G_3 &=  \left( \begin{array}{ccc}
	  2^\frac{C}{3} & 0 & 0 \\
	  0 &  2^\frac{C}{3}  & 0 \\
	  0 & 0 & 2^\frac{C}{3}
       \end{array}
\right)\,.
\end{align}
Since $G_3$ is a s scaled identity matrix, we are in fact seeking a $2$-GMD of the remaining two matrices.
Thus, denoting $b = 2^{\frac{C}{12}}$, we need to perform a $2$-GMD on the following two $3 \times 3$ matrices,
\begin{align}
G_1 =  \left( \begin{array}{ccc}
	  b^6 & 0 & 0 \\
	  0 &  1  & 0 \\
	  0 & 0 & 1
       \end{array}
\right)
\,,\quad
G_2 =  \left( \begin{array}{ccc}
	  b^3 & 0 & 0 \\
	  0 &  b^3  & 0 \\
	  0 & 0 & 1
       \end{array}
\right)\,.
\end{align}
Equivalently, dividing both matrices by  $b^2$, we are seeking a $2$-GMD of the following two $3 \times 3$ matrices, both having a determinant equal to $1$:
\begin{align}
A_1 & =  \left( \begin{array}{ccc}
	  b^4 & 0 & 0 \\
	  0 &  b^{-2}   & 0 \\
	  0 & 0 & b^{-2}
       \end{array}
\right)
\\
A_2 & =  \left( \begin{array}{ccc}
	  b & 0 & 0 \\
	  0 &  b  & 0 \\
	  0 & 0 & b^{-2}
       \end{array}
\right)\,.
\end{align}
As shown in \appref{app:rateless_reduction}, this reduces to performing $2$-GMD on the following two $2 \times 2$ matrices:
\begin{align}
\tilde{A}_1 &= \left( \begin{array}{cc}	
	  \frac{\sqrt{1-b^2+b^8}}{b^2} & \frac{b^6-1}{b\sqrt{(1-b^2+b^8)(1+b^2+b^4)}} \\
	  0 & \frac{b^2}{\sqrt{1-b^2+b^8}}	
       \end{array}
\right)
\\
\tilde{A}_2 &=
\left( \begin{array}{cc}
	   b  & 0 \\
	  0 & b^{-1}
       \end{array}
\right)\,.
\end{align}
We have:
\begin{align}
& F_1(\tilde{A}_1^\dagger \tilde{A}_1-I,\tilde{A}_2^\dagger \tilde{A}_2 - I) =
\\
&  - \frac{(b^2-1)^4(b^2+1)^2(1+b^2+b^4)(1-3b^2+b^4)}{b^{12}} \,,
\end{align}
where $F_1$ is defined in \eqref{eq:definition_of_F}.
According to \thrmref{thm:perfect_two_on_two_complex}, there exists a solution if and only if
this value is non-negative, namely,
\begin{align}
	 1 - 3 \cdot 2^\frac{C}{6} + 2^\frac{C}{3}  \leq 0\,.
\end{align}
This condition is satisfied if and only if:
\begin{align}
C \leq 6 \log \left( \frac{3 + \sqrt{5}}{2}      \right) \approx 8.331\,,
\end{align}
which coincides with the result that was obtained in \cite{ETW_GaussianRateless},
where arduous algebraic manipulations were used to obtain this condition.

Finally, we note that there exists a similar result for four rates ($K=4$). In this case, it is shown in  \cite{ayal_thesis} that there exists a perfect solution if and only if the rate $C$ does not exceed a critical rate, which equals approximately $10.55$.

%% file: multi_user_perms.tex
The problem of transmitting information over arbitrarily permuted parallel channels was studied by Willems and Gorokhov \cite{WillemsGorokhov_PermutedChannels} and by Hof et al.~\cite{Hof_Permuted}.
In this point-to-point scenario, the transmitter is connected to the receiver via $M$ parallel memoryless channels,
sharing the same input alphabet,
the transition matrices of which are known at the transmitter but not their order.
Namely, at each time instant, the transmitter generates $M$ input symbols to be sent over the $M$ parallel channels,
and these symbols are then \emph{permuted} by a one-to-one-mapping (permutation) $\pi \in S_M$ from $\left\{1,\ldots,M\right\}$ onto itself.

The permutation $\pi$ is arbitrary, yet constant throughout the transmission block,
and is known to the receiver but not to the transmitter.
The aim of the receiver is to recover the transmitted message with arbitrarily small error probability.
This channel model is of relevance in scenarios where the gains of the channels are generated according
to an i.i.d.\ distribution, and one may choose the ``design gains'' so as to minimize the outage probability;
for details see \cite[Sec.~VII]{WillemsGorokhov_PermutedChannels}.

In this section we describe a practical capacity-achieving scheme for the \emph{Gaussian} case, described by
\begin{align}
\label{eq:perm_channel_model}
     y_m = \alpha_m x_m + z_m  \,, \quad m=1,2,\ldots,M \,,
\end{align}
where
$x_m$ is the input to the $m$-th channel and is subject to a power constraint\footnote{Alternatively,
the individual power constraints can be replaced by a sum-power constraint.
However, both cases reduce to the same result.}
\begin{align} \label{eq:perms_power_constraint}
\mathbb{E} \left( |x_m|^2 \right) \leq 1 \,,
\end{align}
$y_m$ is the output of the \mbox{$m$-th} channel,
and $\left\{ z_m \right\} $ are i.i.d.\ circularly-symmetric Gaussian variables with unit variance,
independent of $\left\{ x_m \right\}$.
The gains $\left\{ \alpha_m \right\}$ are known to the receiver,
whereas the transmitter knows the gains up to an unknown permutation. Namely, the transmitter knows
the gains but not their \emph{order}.

The $M$ parallel channels~\eqref{eq:perm_channel_model} may be regarded as a single MIMO channel,
\begin{align}
      \by = H\bx + \bz \,,
\end{align}
where $\bx$ is the channel input vector of length $M$, and
$\bz$ is a circularly-symmetric white Gaussian random vector of length $M$ and identity covariance matrix.
The channel matrix $H$ is an $M \times M$ diagonal matrix, which is known at the receiver:
\begin{align} \label{eq:h_permuted}
      H = \left( \begin{array}{cccc}
	  \alpha_1 & 0 & \cdots & 0 \\
	  0 & \alpha_2 & \cdots & 0 \\
	\vdots & \vdots & \ddots & \vdots \\
	  0 & 0 & \cdots & \alpha_M
\end{array}
\right) \,.
\end{align}
The transmitter knows the matrix $H$, up to the unknown order of the diagonal elements.

The latter is, in turn, equivalent to broadcasting the same (common) message to
$K=M!$ receivers simultaneously, where the channel matrix to user $k$ is
\begin{align}
      H_k \triangleq \left( \begin{array}{cccc}
	  \alpha_{\pi_k(1)} & 0 & \cdots & 0 \\
	  0 & \alpha_{\pi_k(2)} & \cdots & 0 \\
	\vdots & \vdots & \ddots & \vdots \\
	  0 & 0 & \cdots & \alpha_{\pi_k(M)}
\end{array}
\right) \,,
\end{align}
and $\pi_k \in S_K$ is a permutation which is different for each user.
As a consequence, this transmission problem may be regarded as a special case of the common-message Gaussian MIMO broadcast one.
Under the power constraint \eqref{eq:perms_power_constraint},
the capacity of this common-message BC scenario is obtained by
taking $\CC=I$ in \eqref{eq:MIMO_BC_capacity}, namely,
\begin{align} \label{eq:capacity}
      C = \sum_{m=1}^{M} \log \left( 1 + |\alpha_m|^2 \right) \,.
\end{align}
We now show how the same transmission schemes as described in the previous sections can be used in this scenario for $M=2 \rightarrow K=2$ and $M=3 \rightarrow K=6$. We give here only the results without proofs. The full details are given in \cite{JET:Permuted_ISIT2012}.

For the case of $M=2$, the channel can be in one of two~``states'':
\begin{align}
    H_1 = \left( \begin{array}{ccc}
                               \alpha_1 & 0  \\
                               0 & \alpha_2  \\
                             \end{array} \right) \,,
\end{align}
\begin{align}
    H_2 = \left( \begin{array}{ccc}
                               \alpha_2 & 0  \\
                               0 & \alpha_1  \\
                             \end{array} \right) \,,
\end{align}
where $\alpha_1,\alpha_2 \geq 0$ are known.

Since there are only two options for the channel matrix $H$,
the capacity in this case can be achieved using JET, as described in \secref{ss:multi_user_intro}.
Specifically, capacity is achieved
by choosing the precoding matrix to be the (scaled) Hadamard matrix
(which coincides with the $2 \times 2$ DFT matrix):
\begin{align}
	V = \frac{1}{\sqrt{2}} \left(
	  \begin{array}{cc}
	     1 & 1 \\
	     1 & -1 \\
	  \end{array}
\right) \,.
\end{align}
Similarly, in the case of three parallel channels ($M=3$), we have:
\begin{align} \label{eq:abc}
    H = \left( \begin{array}{ccc}
                               \alpha_1 & 0 & 0 \\
                               0 & \alpha_2 & 0 \\
                               0 & 0 & \alpha_3
                             \end{array} \right) \,,
\end{align}
where $\alpha_1,\alpha_2,\alpha_3 \geq 0$ are known, up to an unknown permutation.
In this case, capacity is achieved by the following precoding matrix, which is the $3 \times 3$ DFT matrix:
\begin{align}
    V = \frac{1}{\sqrt{3}} \left( \begin{array}{ccc}
                               1 & 1 & 1 \\
                               1 & e & e^{-1} \\
                               1 & e^{-1} & e
                             \end{array} \right) \,,
\end{align}
where $e \triangleq e^{2 \pi i / 3}$.

For $M \geq 4$, capacity is no longer achieved using a DFT precoding matrix.
Nevertheless, extension of the above scheme to $4 \leq M \leq 6$ is possible \cite{JET:Permuted_ISIT2012} by utilizing algebras of higher dimensions, such as the quaternion algebra. These algebras can be materialized using a
space--time structure over the complex or real fields. Moreover, the complex field may be represented over the reals by incorporating time extensions, as is explained in the sequel~--- in \secref{ss:restatement}.

In the next section we describe the space--time structure that is used for the construction of joint triangularization of more than two matrices.

%% file: space_time.tex
\subsection{Introduction}
\label{ss:space_time_intro}
\input{space_time_intro}

\subsection{Space--Time Common-Message BC Scheme}
\label{ss:space_time_scheme}
\input{space_time_scheme}

\subsection{Space--Time $2$-GMD for $2 \times 2$ Matrices}
\label{ss:restatement}
\input{space_time_two_by_two}

%% file: space_time_intro.tex
 As indicated by \thrmref{thm:perfect_two_on_two_complex}, joint (unitary) triangularization with constant diagonal values
($K$-GMD)
 is not always possible. However, even when the condition for joint triangularization does not hold, it is possible to gain more mathematical degrees of freedom  by utilizing multiple uses of the same channel realization. 
The idea of mixing the same symbols between multiple channel uses has much in common with OSTBC~\cite{Alamouti,TarokhJafarkhaniCalderbank_STBC}.
However, whereas space--time processing has traditionally been applied to an open-loop communication scenario,
in the present work it will be applied to the closed-loop common-message BC problem.

We first recall the idea of linear space--time codes, also known as linear dispersion codes (see, e.g., \cite{HassibiSTC}), which will be used as a building block for the proposed communication scheme.
For this, we consider the point-to-point MIMO Gaussian channel, with an $n_r \times n_t$ channel matrix $H$,
\begin{align}
      \by = H \bx + \bz \,.
\end{align}
We now utilize transmission over $N$ consecutive blocks, assuming that the channel matrix $H$ does not change between these blocks.
This is equivalent to sending time-extended symbols over the following \emph{time-extended channel}:
\begin{align} \label{eq:extended_channel}
    \yyy = \HHH \xxx + \zzz \,.
\end{align}
The time-extended vectors $\xxx,\yyy,\zzz$ are composed of $N$ ``physical'' (concatenated) input, output, and noise vectors, respectively, and $\mathcal{H}$ is the $(N n_r) \times (N n_t)$  \emph{time-extended channel matrix} defined as
\begin{align}
\label{eq:extended_channel_def}
      \mathcal{H} = \blkmat{H}{N} \,,
\end{align}
where $ \blkmat{A}{N} $ denotes the Kronecker product $I_N \otimes A$, viz.\ a block-diagonal matrix with $N$ blocks of $A$ on its diagonal:
\begin{align}
    \blkmat{A}{N} \triangleq \left( \begin{array}{cccc} A & 0  & \cdots & 0 \\ 0 & A & \cdots & 0 \\ \vdots & \vdots & \ddots & \vdots \\ 0 & 0 & \cdots & A
                   \end{array} \right)  \,.
\end{align}
In linear space--time modulation (also known as ``space--time coding'') the extended input vector $\xxx$ is obtained by 
linearly combining independent streams of data symbols.\footnote{The transformation may, more generally, be taken to be linear over the reals. Nevertheless, for the purposes of this paper it suffices to consider only linear transformations over the complex numbers.}
Of special interest are modulations that possess a certain structure with the aim of 
facilitating decoding. Such a family includes OSTBCs, and in particular Alamouti modulation \cite{Alamouti}.
When using an OSTBC, the transmitter applies a unitary transformation, \emph{which does not depend on the channel matrix $H$}, to the data symbols, and the receiver applies another orthogonal transformation to the channel output, such that the effective channel matrix is transformed into a \emph{diagonal} form, over which communication is possible using off-the-shelf codes designed for \emph{scalar} AWGN channels.
Thus, \emph{simultaneous} diagonalization of all possible channel matrices, is attained.

Unfortunately, OSTBCs that universally achieve the white-input capacity of every channel, as is the case for Alamouti modulation, do not exist for MISO channels with more than $2$ transmit antennas, let alone for MIMO channels \cite{TarokhJafarkhaniCalderbank_STBC,LiangSTC}.

In this work, we use the idea of space--time modulation, but instead of diagonalizing the channel matrices,
we are content with \emph{triangularization}.
This, in turn, requires the employment of another ingredient to the communication scheme, namely,
successive interference cancellation at the receivers.
Further, in contrast to OSTBC, where the same transformation is applied to a continuum of channels,
the proposed approach is applicable to only a finite number of channel matrices.

%% file: space_time_scheme.tex
We now introduce the space-time common-message BC scheme.
Recall the common-message broadcast MIMO channel \eqref{MIMO-BC} with $K$ users and $n_t$ transmit antennas. We now utilize transmission over $N$ consecutive blocks, assuming that the channel matrices do not change between these blocks.
This is equivalent to sending extended symbols over the following \emph{time-extended channels}:
\begin{align}
	\yyy_k = \HHH_k \xxx + \zzz_k \,, \qquad k=1,\ldots,K  \,,
\end{align}
where the time-extended vectors $\xxx, \yyy_k, \zzz_k$ and time-extended matrices $\HHH_k$ are defined as in \eqref{eq:extended_channel} and \eqref{eq:extended_channel_def}.
\!\!\footnote{This technique can be extended to the case where the channel matrices are time-varying. In this case, the time-extended channel matrices of \eqref{eq:extended_channel_def} are replaced by the block-diagonal matrices
\begin{align}
    \HHH_k = \left(
  \begin{array}{cccc}
	H_k^{(1)} & 0 & \cdots & 0 \\
	0 & H_k^{(2)} & \cdots & 0 \\
	\vdots & \vdots & \ddots & \vdots \\
	0 & 0 & \cdots & H_k^{(N)} \\
  \end{array}
\right)
    \,.
\end{align}}
The power constraint now becomes $\mathbb{E} \left[ \xxx^\dagger \xxx \right] \leq N P$.

Let $\CC$ be an $n_t \times n_t$ covariance matrix satisfying \mbox{$\trace \CC \leq P$}.
As explained in \remref{rem:equal_det}, we can assume
without loss of generality that
\begin{align}
  \Capacity{1}=\cdots=\Capacity{K}=C \,.
\end{align}

Define the  matrices $\tilde{H}_k$, $Q_k$, and $G_k$ as in
\eqref{eq:tilde_G_K_users} and {\eqref{eq:G_matrix_K_users}.
Further define the following time-extended channel canonical matrices:
\begin{align} \label{eq:extended_two}
    \GGG_k \triangleq  \blkmat{G_k}{N}  \,, \quad k=1,\ldots,K \,.
\end{align}
Now, assume that there exists a $K$-JET of the matrices $\GGG_k$:
\begin{align}
    \GGG_k &= \UUU_k \RRR_k \VVV^\dagger \,,
\end{align}
where $\RRR_k$ are upper triangular matrices whose diagonal values are equal to $r_1,\ldots,r_{n_t N}$.
Then, the same transmission scheme as in \secref{ss:p2p_scheme} can be employed,
with the following replacements:
\begin{itemize}
 \item
The transmitted vector $\bx$ is replaced by the time-extended vector $\xxx$
\item
The received vector $\by$ is replaced by the time-extended vector $\yyy_k$
\item
In step 3, the $k$-th user uses the matrix $\UUU_k$ instead of $U$ in \eqref{eq:scheme_receiver},
and the matrix $\tilde{Q}$ is replaced with its time-extended version,
$ \blkmat{ \tilde{Q}_k }{N} $, where $\tilde{Q}_k$ consists of the first $n_t$ rows of $Q_k$.
\end{itemize}

%% file: space_time_two_by_two.tex
We now consider the special case where the transmitter is equipped with 2 antennas,
and we are interested in performing
$2$-GMD, or alternatively, $3$-JET, on the extended matrices.

As we saw in \secref{ss:multi_user_perfect_2x2}, 2-GMD of $2 \times 2$ matrices is not always possible.
This raises the question whether we can exploit the space--time structure to
perform $2$-GMD on the extended matrices,
even in cases where 2-GMD of the original (not time-extended) matrices is not possible.

For a general number of antennas $n_t$, we know that space--time structures can sometimes enable GMD in cases where it is not possible without time extensions (see, e.g., \cite{JET:Permuted_ISIT2012}).
However, in some cases, space--time structures cannot help. Such is the case for $n_t=2$, as
implied by the following theorem which is proved in \appref{app:extention_does_not_help}.

\begin{thm} \label{thm:extention_does_not_help}
Let $A_1$ and $A_2$ be \emph{complex-valued} $2 \times 2$ matrices with determinants equal to $1$,
such that condition \eqref{eq:condition3} does not hold (namely, there does not exist a 2-GMD of the matrices $A_1$ and $A_2$).
Let $N \in \mathbb{N}$, and define the following extended matrices:
\begin{align}
      \AAA_k \triangleq \blkmat{A_k}{N} \,, \quad k=1,2 \,.
\end{align}
Then, there also does not exist 2-GMD of the matrices $\AAA_1,\AAA_2$, for any value of $N \in \mathbb{N}$.
\end{thm}

Consider now the case where the channel matrices are real-valued, and we allow the use of
only orthogonal real-valued matrices $U_k,V$ in the communication scheme.
Then, if condition \eqref{eq:condition3} holds,
a space--time structure with $N=2$ enables 2-GMD.
This is explained in the following corollary.

\begin{corol}
If condition \eqref{eq:condition3} holds, then according to \thrmref{thm:perfect_two_on_two_complex}
we can perform $2$-GMD on $A_1,A_2$ \eqref{eq:two_by_two_gmd} with complex-valued unitary matrices $U_1,U_2,V$.
In particular, we can assume that the
three matrices $U_1,U_2,V$ are of the following form:
\begin{align} \label{eq:abcd_form}
      \left( \begin{array}{cc}
	    a+bi & \phantom{-}c+di  \\
	    c-di & -a+bi \\		
                   \end{array} \right)\,.
\end{align}
This implies that there exists a $2$-GMD of the extended matrices with $N=2$, $\AAA_1$ and $\AAA_2$,
where the corresponding real-valued orthogonal matrices $\UUU_1,\UUU_2,\VVV$ are derived from $U_1,U_2,V$  \eqref{eq:abcd_form} as follows:
\begin{align}
      \left( \begin{array}{cccc}
		      \phantom{-}a & -b & \phantom{-}c & -d \\
   \phantom{-}c &  \phantom{-}d & -a & -b \\		
		    \phantom{-}b & \phantom{-}a & \phantom{-}d & \phantom{-}c \\		
			  -d & \phantom{-}c & \phantom{-}b & -a
                   \end{array} \right) \,.
\end{align}
However, more extensions, i.e., $N \geq 3$, cannot help to construct (perfect) 2-GMD,
due to \thmref{thm:extention_does_not_help}.
\end{corol}

%% file: space_time_nearly.tex
As indicated by \thrmref{thm:perfect_two_on_two_complex}, joint triangularization with constant diagonal values ($K$-GMD) is not always possible even if we consider time-extended channel matrices, 
as in \thrmref{thm:extention_does_not_help}.

The question is whether we may use the transmission scheme, presented in \secref{ss:multi_user_scheme}, for the general multi-user problem. 
We now demonstrate that although perfect decomposition is not possible in general, we can still perform \emph{nearly-optimal} triangularization, by utilizing multiple uses of the same channel realization.

There are many ways to define ``nearly optimal''. Commonly, this term refers to a problem with some optimization criterion, or some error criterion, where the optimization solution or the error are bounded, based on some statistical assumptions. Here, we refer to a different meaning.
We strive for an explicit lower bound on the communication rate (without any statistical assumption on the generation processes of the channel matrices), which is asymptotically optimal, in the number of time extensions utilized.
These two goals are achieved by defining ``nearly optimal $K$-GMD'', in which the resulting matrices are as in ``perfect $K$-GMD'' form~--- upper triangular matrices with equal and constant diagonal elements~--- up to a small number of diagonal elements, which becomes negligible as the number of time extensions grows. 
This is defined formally as follows. 

\begin{defn}[Nearly-Optimal $K$-GMD]
  Let $A_1,\dots,A_K$ be complex-valued $n \times n$ matrices with determinants equal to $1$. 
  Consider a sequence of decompositions (for each $N$) of the following form.
  For each $N$, define the following $nN \times nN$ extended matrices:
  \begin{align}
  \label{eq:extended_n}
    \AAA_k \triangleq \blkmat{A_k}{N}
  \,, \quad k=1,\dots,K \,,
  \end{align}
  and the $(K+1)$ matrices $\UUU_1,\ldots,\UUU_K,\VVV$
  of dimensions \mbox{$nN \times \nnn$},
  with orthonormal columns, such that:
  \begin{align}
    \UUU_k^\dagger \AAA_k \VVV = \left( \begin{array}{ccccc}
		    1 & * & \cdots & * & * \\
		    0 & 1 & \cdots & * & *\\
		    \vdots & \vdots & \ddots & \vdots & \vdots \\
		    0 & 0 & \cdots & 1 & * \\
		    0 & 0 & \cdots & 0 & 1
		\end{array} \right) 
    \triangleq \TTT_k \,,
    \quad k=1, \dots, K \,,
 \label{eq:nearly:Tk}
  \end{align}
  where $*$ represents some value (which may differ within each matrix as well as between different ones).

  We say that the sequence of decompositions is \emph{nearly-optimal $K$-GMD}, if 
  \begin{align}
      \lim_{N \rightarrow \infty} \frac{ \nnn }{n N} = 1 \,.
  \end{align}
\end{defn}
\vspace{.5\baselineskip}

\begin{thm}[Existence of nearly-optimal K-GMD]
\label{thm:n_n_asymptotical}
  For any $K$ complex-valued $n \times n$ matrices $A_1,\dots,A_K$ with determinants equal to $1$, there exists a sequence of \emph{nearly-optimal K-GMD} with 
  \mbox{$\nnn = n \big(N - ( n^{K-1} - 1 ) \big)$}, 
  where $N \geq n^{K-1}$.
\end{thm}

\vspace{\baselineskip}

Note that again, as was explained in \remref{rem:equal_det},
we assume, w.l.o.g., that all matrices have determinants equal to~1.

The proof of the theorem is given in the form of a constructive algorithm.
The algorithm for the general case is presented in \appref{app:proof_general}.
Also, implementations of the algorithm in Matlab and Python are available in \cite{MATLAB} and \cite{PYTHON}, respectively.
In order to simplify the understanding of the algorithm, we demonstrate the algorithm for some special cases, each of which illustrates a different aspect of the general case. 
In \secref{ss:space_time_proof_K_2_n_2} we present the algorithm for the simplest case of $2$-GMD of extended $2 \times 2$ matrices, with any number of time extensions.
In \appref{app:proof_K_3_n_2} we present the algorithm for the case of $3$-GMD of extended $2 \times 2$ matrices with only $N=4$ extensions. In \appref{app:proof_n_2} we
generalize this for general $K$-GMD of extended $2 \times 2$ matrices.
Finally, in \appref{app:proof_K_2}, we present the algorithm for $2$-GMD of extended $n \times n$ matrices.

We note that, similarly to the case of
perfect triangularization, nearly optimal $K$-GMD is equivalent to nearly optimal $(K+1)$-JET. This
is formally stated in the following lemma, which is a generalization of \lemref{lem:k_to_k_plus_one_square} to the non square-matrix case, 
and is proved in \appref{app:proof_k_to_k_plus_one}.

\begin{lemma}[Equivalence of K-GMD and (K+1)-JET]
\label{lem:k_to_k_plus_one}
Let $A_1,\ldots,A_{K+1}$ be $n \times n$ full-rank complex-valued matrices with equal determinants, and
define the $K$ matrices:
\begin{align} \label{eq:definition_of_a}
			B_k = A_k A_{K+1}^{-1}  \,, \qquad k=1,\ldots,K \,.
\end{align}
Then, the following two statements are equivalent:
\begin{enumerate}
\item There exist $K+1$
matrices with orthonormal columns $U_1,\ldots,U_{K},U_{K+1}$, of dimensions $n \times \nnn$, such that
\begin{align} \label{eq:jet_of_a}
	U_k^\dagger B_k U_{K+1} = T_k \,,\qquad  k=1,\ldots,K \,,
\end{align}
where $\left\{ T_k \right\} $ are $\nnn \times \nnn$ upper triangular with all diagonal entries equal to $1$.
\item There exist $K+2$ matrices with orthonormal columns
$U_1,\ldots,U_{K+1},V$, of dimensions $n \times \nnn$, such that
\begin{align}
	U_k^\dagger A_k V = R_k \,,\qquad k=1,\ldots,K+1 \,,
\end{align}
where $\left\{ R_k \right\}$ are $\nnn \times \nnn$ upper triangular with equal diagonals,
as in \eqref{eq:R_equal_diagonals}.
\end{enumerate}
\end{lemma}

Nearly optimal $K$-GMD is readily applied for K-user common-message BC:
Transmission is carried over the equal sub-channel gains whereas the non-equal ones are discarded.

\begin{corol}[Achievable Rates via Nearly-Optimal $K$-GMD]
\label{corol:nearly:achievable_rate}
    Let $H_1, \ldots, H_K$ be complex-valued channel matrices of dimensions $n_{r}^{(1)} \times n_t, \ldots, n_{r}^{(K)} \times n_t$, respectively, and $C_\bx$ be an $n_t \times n_t$ covariance matrix satisfying the power constraint $\trace{C_\bx} \leq P$.
    Define $\{ \HHH_k \}$, $\{ G_k \}$, and $\{ \GGG_k \}$ as in \secref{ss:space_time_scheme} with $N \geq n_t^{K-1}$ time extensions. Without loss of generality, assume that 
    \begin{align*}
        \Capacity{1}=\cdots=\Capacity{K}=C \triangleq n_t \log(1 + \SNR_\eff) \,.
    \end{align*}
    Then, the following common-message BC 
    rate is achieved:
    \begin{subequations}
    \noeqref{eq:nearly:achievable_rate:power_compensate,eq:nearly:achievable_rate:naive}
    \begin{align}
        R &= 
        \left[1 - \frac{n_t^{K-1} - 1}{N} \right] n_t \log \left(1 + \frac{N}{N - ( n_t^{K-1} - 1 )} \SNR_\eff \right)
    \label{eq:nearly:achievable_rate:power_compensate}
     \\ &\geq \left[1 - \frac{n_t^{K-1} - 1}{N} \right] C
     \,,
    \label{eq:nearly:achievable_rate:naive}
    \end{align}
    \end{subequations}
    using equal-rate capacity-achieving scalar AWGN codes.
    By taking $N \rightarrow \infty$, the achievable rate $R$ achieves capacity.
\end{corol}

\vspace{.5\baselineskip}

\begin{proof}[Proof of \colref{corol:nearly:achievable_rate}]
    Apply \thmref{thm:n_n_asymptotical} to $\{ \GGG_k \}$ to obtain the square upper triangular matrices
    $\{ \TTT_k \}$ of dimensions \mbox{$n_t \big(N - ( n_t^{K-1} - 1 ) \big)$} with constant diagonals.
    By using the transmission scheme of \secref{ss:space_time_scheme} over $\{ \TTT_k \}$ a rate of \eqref{eq:nearly:achievable_rate:naive} is achieved.
    By allocating power and rate only to the $n_t \big(N - ( n_t^{K-1} - 1 ) \big)$ non-discarded streams corresponding to the (constant) diagonal values in $\{\TTT_k\}$ in \eqref{eq:nearly:Tk}, 
    the improved rate of \eqref{eq:nearly:achievable_rate:power_compensate} is achieved.
\end{proof}

\vspace{.5\baselineskip}

\begin{remark}
    Any nearly optimal $K$-GMD sequence (not necessarily the one specified in \thmref{thm:n_n_asymptotical}) allows to approach capacity in the limit of $N \rightarrow \infty$.
\end{remark}

\vspace{.5\baselineskip}

We now demonstrate \colref{corol:nearly:achievable_rate} for two special cases.
\begin{exmpl}[\exref{DOF-example} Revisited]
\label{DOF-example_3users}
    We reexamine the three-user degrees-of-freedom mismatch setting that was introduced in \exref{DOF-example} in \secref{s:intro}, which we reproduce here for convenience.
We have three users with the following channel matrices:
\begin{align*}
    H_1 = \left(
            \begin{array}{cc}
              \alpha_1 & 0 \\
            \end{array}
          \right) \col{}{\,},
    \col{\:}{\quad}
    H_2 = \left(
            \begin{array}{cc}
              0 & \alpha_1 \\
            \end{array}
          \right) \col{}{\,},
    \col{\:}{\quad}
    H_3 = \left(
            \begin{array}{cc}
              \alpha_2 & 0 \\
              0 & \alpha_2 \\
            \end{array}
          \right)
    \col{}{\,},
\end{align*}
such that their WI capacities are equal.

    For this specific case, since the third channel matrix is a scaled identity matrix,  $3$-JET and $3$-GMD coincide. Therefore, the number of channel uses needed to achieve $3$-GMD is identical to that of $3$-JET.

    Table \ref{table:PortionOfCapacityInfiniteSNRExample} summarizes achievable fractions of capacity corresponding to different numbers of time extensions. We note that in the table we do not apply power compensation 
    as appears in \eqref{eq:nearly:achievable_rate:power_compensate}. Thus, the achievable rates according to \eqref{eq:nearly:achievable_rate:naive} are tabulated.
    For comparison, with $P \rightarrow \infty$, time-sharing between the users achieves $33\%$ of the capacity, whereas both Alamouti modulation and beamforming achieve $50\%$.\footnote{In all the schemes, we assume that the scalar codes used are capacity-achieving.} 
    We note that Alamouti modulation falls under the framework of space--time triangularization (in this case diagonalization) with two time extensions, see \cite[Ch. 1.7.3]{livni_thesis}.
    By using more than two time extensions, the proposed scheme achieves a larger fraction of capacity.

    \begin{table}[ht!]
      \centering
      \begin{tabular}{||c||c|c|c|c|c|c|c|c|c||}
        \hline
        \# Time extensions &  2 &  3 & 4  & 5  & 6  & 7 & 8 & 10 \\ \hline
        \% Capacity    & 50 & 66 & 75 & 80 & 83 & 85 & 87 & 90 \\
      \hline
      \end{tabular}
      \caption{Fraction of capacity achievable for different numbers of channel uses processed together, when using $3$-GMD and $3$-JET (without power compensation) in \exref{DOF-example_3users}.}
    \label{table:PortionOfCapacityInfiniteSNRExample}
    \end{table}

    \begin{remark}
     Note that all the schemes considered here impose a decoding order which is shared among all the users. We will see in \secref{ss:upper_lower} that in this particular example, removing this restriction enables to attain 100\% efficiency (with no time extensions!).
    \end{remark}
\end{exmpl}

\begin{exmpl}[A General Three-User $2 \times 2$ Case]
\label{ex:nearly:3user_general}
    We assume now three general $n_r^{(k)} \times 2$ channel matrices.
    The resulting channel canonical matrices \eqref{eq:G_matrix_K_users} are of dimensions
    \mbox{$2 \times 2$}.
    To be optimal for all three users simultaneously, we need to use $3$-JET (which can be done using the same parameters of $2$-GMD, as explained in \remref{rem:nearly_k_jet}). If we further wish to have the same SNR for all the scalar sub-channels, then we need to use $3$-GMD.
    Table \ref{table:PortionOfCapacityInfiniteSNR} summarizes achievable fractions of capacity corresponding to different numbers of time extensions. Again, the achievable rates tabulated are according to  \eqref{eq:nearly:achievable_rate:naive}.
    For comparison, with $P \rightarrow \infty$, time-sharing between the users achieves $33\%$ of the capacity, whereas both Alamouti modulation and beamforming achieve $50\%$.\footnote{In all the schemes, we assume that the scalar codes used are capacity-achieving.} 

    \begin{table}[ht!]
      \centering
      \begin{tabular}{||c||c|c|c|c|c|c|c|c|c||}
        \hline
        \# Time extensions &  2 &  3 & 4  & 5  & 6  & 10 & 15 & 30 \\ \hline
        GMD \% Capacity & -- & -- & 25 & 40 & 50 & 70 & 80 & 90 \\ \hline
        \hspace{1mm} JET \% Capacity & 50 & 66 & 75 & 80 & 83 & 90 & 93 & 96 \\
      \hline
      \end{tabular}
      \caption{Fraction of capacity achievable for different numbers of channel uses processed together, when using $3$-GMD and $3$-JET (without power compensation) in \exref{ex:nearly:3user_general}.}
    \label{table:PortionOfCapacityInfiniteSNR}
    \end{table}
\end{exmpl}

\subsection{Preliminaries for the Proof of \thrmref{thm:n_n_asymptotical}}
\label{ss:space_time_proof_preliminaries}

We now introduce some definitions and properties that
will be used in the proof of \thrmref{thm:n_n_asymptotical} in \appref{app:proof_general},
as well as in its demonstration for the simple $2 \times 2$ matrix case in \secref{ss:space_time_proof_K_2_n_2} 
and the demonstrations in Appendices \ref{app:proof_K_3_n_2}--\ref{app:proof_K_2}.

 \begin {defn}
    Define by $j:m$ the list of consecutive indices between $j$ and $m$:
    \begin{align}
      j:m \triangleq  (j,j+1,j+2,\ldots,m)
      \,.
    \end{align}
 \end{defn}

 \begin{defn}
   Define the operation of ``extraction'' of multiple \emph{ordered} indices
\mbox{$ n_1\,,n_2\,,\ldots\,,n_k$}  from a matrix $A$ by:
   \begin{align}
     \submat{A}{n_1}{n_2\,,\ldots\,,n_k} \triangleq   \left(
 				\begin{array}{cccc}
 				  A_{n_1n_1} & A_{n_1n_2} & \cdots & A_{n_1n_k} \\
 				  A_{n_2n_1} & A_{n_2n_2} & \cdots & A_{n_2n_k} \\
                  \vdots     & \cdots     & \ddots & \vdots     \\
                  A_{n_kn_1} & A_{n_kn_2} & \cdots & A_{n_kn_k}
 				\end{array}
 				\right)
     \,.
   \end{align}

   For example, if

   \begin{align}
   A=
   \left(
     \begin{array}{cccccc}
       1  & 2  & 3  & 4  & 5  & 6 \\
       7  & 8  & 9  & 10 & 11 & 12 \\
       13 & 14 & 15 & 16 & 17 & 18 \\
       19 & 20 & 21 & 22 & 23 & 24 \\
       25 & 26 & 27 & 28 & 29 & 30 \\
       31 & 32 & 33 & 34 & 35 & 36 \\
     \end{array}
   \right) \,,
   \end{align}
   then,
   \begin{align}
   \submat{A}{2}{5} &=
        \left(
          \begin{array}{cc}
            8 & 11 \\
            26 & 29 \\
          \end{array}
        \right) \,,\\
   \submatt{A}{3}{5} &=
        \left(
          \begin{array}{ccc}
            15 & 16 & 17 \\
            21 & 22 & 23 \\
            27 & 28 & 29 \\
          \end{array}
        \right)\,, \\
   \submat{A}{1}{6\,,2} &=
        \left(
          \begin{array}{ccc}
            1 & 6 & 2 \\
            31 & 36 & 32 \\
            7 & 12 & 8 \\
          \end{array}
        \right)\,.
   \end{align}
   \end{defn}

   \begin{defn}

   Define the ``embedding'' operation $\embb{n}{A}{\bigcup_j \subm{m_j}{n_j}}$ as the replacement of the elements in the identity matrix $I_n$ in the index-pairs contained in $\subm{m_1}{n_1}\subm{m_2}{n_2}\subm{m_3}{n_3}\ldots\subm{m_k}{n_k}$,\footnote{The notation $\subm{j}{m}\subm{p}{q}$ stands for $\subm{j}{m} \cup \subm{p}{q}$.} with the elements of the $2 \times 2$ matrix $A$.

   For example, the embedding $\embb{4}{B}{\subm{1}{3}\subm{2}{4}}$ of
   \begin{align}
     B = \left( \begin{array}{cc}
   	      11 & 2 \\
   	      3  & 4
   	     \end{array}
         \right)
   \end{align}
   into the four-dimensional identity matrix $I_4$ is
  \begin{align}
    \left(
    \begin{array}{cccc}
      \cellcolor[gray]{0.8} 11 & 0 & \cellcolor[gray]{0.8} 2 & 0 \\
      0  & \cellcolor[gray]{0.5} 11 & 0 & \cellcolor[gray]{0.5} 2 \\
      \cellcolor[gray]{0.8} 3  & 0 & \cellcolor[gray]{0.8} 4 & 0 \\
      0  & \cellcolor[gray]{0.5} 3 & 0 & \cellcolor[gray]{0.5} 4 \\
    \end{array}
    \right)\,.
  \end{align}
   \end{defn}

   \begin{defn} \label{def:extracted_identity}
   Define the matrix $\extt{n}{\{n_j\}_{j=1}^{k}}$ as an $n \times k$
matrix,
whose columns are the ${\{n_j\}_{j=1}^{k}}$ vectors of the standard basis:
   \begin{align}
   \extt{n}{\{n_j\}}=
   \left(
     \begin{array}{c|c|c|c}
       e_n^{n_1} & e_n^{n_2} & \cdots & e_n^{n_k} \\
     \end{array}
   \right) \,,
   \end{align}
    where $e_n^{n_j}$ is a column-vector of length $n$ with all entries $0$ except for the $n_j$-th entry which equals $1$.

   For example,
   \begin{align}
   \extt{5}{4,1,5} =
   \left(
   \begin{array}{ccc}
        0 & 1 & 0\\
        0 & 0 & 0\\
        0 & 0 & 0\\
        1 & 0 & 0\\
        0 & 0 & 1\\
   \end{array}
   \right)\,.
   \end{align}
   Note that $\left(\extt{n}{\{n_j\}_{j=1}^k}\right)^\dagger \extt{n}{\{n_j\}_{j=1}^k}=I_k$.
   \end{defn}

    \begin{remark}\label{remark:extract}
    For any matrix $A$, ``extraction'' can be materialized via multiplication by a matrix
$ \mI_n^{\{[n_j]\}} $ of Definition \ref{def:extracted_identity}:
    \begin{align}
    \submat{A}{n_1}{n_2\,,\ldots\,,n_k}=\left(\mI_n^{[\{n_j\}]}\right)^\dagger A \mI_n^{[\{n_j\}]} \,.
    \end{align}
    An important special case is the extraction operation of a submatrix:
    \begin{align}
        \submatt{A}{j}{m} \triangleq \left(\extt{n}{j:m}\right)^\dagger A \extt{n}{j:m}\,.
    \end{align}
    \end{remark}

We now introduce a simple key property that will serve as the main idea in our proofs.

\begin {property} \label{prop:lenI}
Let $A$ be a scaled identity matrix, namely, \mbox{$A = cI$}, for some scalar $c$.
The QR decomposition of the matrix $A$ is invariant to multiplications by unitary matrices on the right.
This means that for any unitary matrix $V$, 
the resulting triangular matrix after applying the QR decomposition to the matrix $AV$ is the matrix $A$, and further $Q=V^\dagger$:
\begin{align}
cI_n=V^\dagger c I_n V \qqquad \forall c,n \,.
\end{align}
\end{property}

\subsection{Proof of \thrmref{thm:n_n_asymptotical} for $n=2,K=2$ and General $N$}
\label{ss:space_time_proof_K_2_n_2}

We now demonstrate the algorithm for the special case of $n=2$, $K=2$, and general $N$.
The proof is based on $K=2$ steps. \\
\underline{Step 1:} \\
We start by performing $1$-GMD on the matrix $A_1$:
\begin{align} \label{eq:stage_a}
	\Ul11 A_1 \Vl1 & =  \left( \begin{array}{cc}
	                     1 & x_1 \\ 0 & 1
	                 \end{array} \right)\,,
\end{align}
where the superscripts denote the step number and the subscripts denote the user index.
We now apply the decomposition \eqref{eq:stage_a} to each block separately, using:
\begin{align}
\UUUU11 &\triangleq\embb{2N}{\Ul11}{\subm{1}{2}\subm{3}{4}\cdots\subm{2N-1}{2N}} \,,\\
\VVVV1  &\triangleq\embb{2N}{\Vl1} {\subm{1}{2}\subm{3}{4}\cdots\subm{2N-1}{2N}} \,,
\end{align}
which yields the following $2N \times 2N$ extended triangular matrix:
\col{
\begin{align}
    \TTT_1^{(1)} &=   \UUUU11 \AAA_1 \VVVV1
\\
&= \left(
\begin{array}{ccccccc} \cline{1-2}
	 \multicolumn{1}{|c}{1} & \multicolumn{1}{c|}{x_1} & 0 & 0 & \cdots & 0 & 0 \\ \cdashline{2-3}
	 \multicolumn{1}{|c}{0}    &  \multicolumn{1}{:c|}{1} & 0 & \multicolumn{1}{:c}{0} & \cdots & 0 & 0 \\ \cline{1-4}
	  0 & \multicolumn{1}{:c}{0} & \multicolumn{1}{|c}{1} & \multicolumn{1}{:c|}{x_1}  & \cdots& 0 & 0  \\ \cdashline{2-3}
	  0 & 0 &  \multicolumn{1}{|c}{0}    & \multicolumn{1}{c|}{1} &\cdots &  0 & 0  \\ \cline{3-4}
	  \vdots  &  \vdots & \vdots  &  \vdots   & \ddots &  \vdots  &  \vdots   \\	    \cline{6-7}
	  0 & 0 & 0 & 0 &\cdots & \multicolumn{1}{|c}{1} & \multicolumn{1}{c|}{x_1} \\
	  0 & 0 & 0 & 0 &\cdots &  \multicolumn{1}{|c}{0}    & \multicolumn{1}{c|}{1} \\ \cline{6-7}
\end{array}
\right)\,.
\end{align}
}
{
\begin{align}
    \TTT_1^{(1)} &=   \UUUU11 \AAA_1 \VVVV1
\\*
&=\left(
\begin{array}{ccccccc} \cline{1-2}
	 \multicolumn{1}{|c}{1} & \multicolumn{1}{c|}{x_1} & 0 & 0 & \cdots & 0 & 0 \\ \cdashline{2-3}
	 \multicolumn{1}{|c}{0}    &  \multicolumn{1}{:c|}{1} & 0 & \multicolumn{1}{:c}{0} & \cdots & 0 & 0 \\ \cline{1-4}
	  0 & \multicolumn{1}{:c}{0} & \multicolumn{1}{|c}{1} & \multicolumn{1}{:c|}{x_1}  & \cdots& 0 & 0  \\ \cdashline{2-3}
	  0 & 0 &  \multicolumn{1}{|c}{0}    & \multicolumn{1}{c|}{1} &\cdots &  0 & 0  \\ \cline{3-4}
	  \vdots  &  \vdots & \vdots  &  \vdots   & \ddots &  \vdots  &  \vdots   \\	    \cline{6-7}
	  0 & 0 & 0 & 0 &\cdots & \multicolumn{1}{|c}{1} & \multicolumn{1}{c|}{x_1} \\
	  0 & 0 & 0 & 0 &\cdots &  \multicolumn{1}{|c}{0}    & \multicolumn{1}{c|}{1} \\ \cline{6-7}
\end{array}
\right)\,.
\end{align}
}
Note that the same matrix $\mV^{(1)}$ has to be applied also to the matrix of the second user (since the encoder is shared by all users).
  We next decompose the resulting matrix (after multiplying it by $\mV^{(1)}$ on the right) according to the QR decomposition,
  resulting in a unitary matrix $\UUUU21$ such that:
\col{
\begin{align}
   & \TTT_2^{(1)} =   \UUUU21 \AAA_2 \VVVV1 \\*
   & = \left(
\begin{array}{ccccccc} \cline{1-2}
	 \multicolumn{1}{|c}{r_1} & \multicolumn{1}{c|}{x_2} & 0 & 0 & \cdots & 0 & 0 \\ \cdashline{2-3}
	 \multicolumn{1}{|c}{0}    &  \multicolumn{1}{:c|}{r_2} & 0 & \multicolumn{1}{:c}{0} & \cdots & 0 & 0 \\ \cline{1-4}
	  0 & \multicolumn{1}{:c}{0} & \multicolumn{1}{|c}{r_1} & \multicolumn{1}{:c|}{x_2}  & \cdots& 0 & 0  \\ \cdashline{2-3}
	  0 & 0 &  \multicolumn{1}{|c}{0}    & \multicolumn{1}{c|}{r_2} &\cdots &  0 & 0  \\ \cline{3-4}
	  \vdots  &  \vdots & \vdots  &  \vdots   & \ddots &  \vdots  &  \vdots   \\	    \cline{6-7}
	  0 & 0 & 0 & 0 &\cdots & \multicolumn{1}{|c}{r_1} & \multicolumn{1}{c|}{x_2} \\
	  0 & 0 & 0 & 0 &\cdots &  \multicolumn{1}{|c}{0}    & \multicolumn{1}{c|}{r_2} \\ \cline{6-7}
\end{array}
\right)\,.
\end{align}
}
{
\begin{align}
    \TTT_2^{(1)} &=   \UUUU21 \AAA_2 \VVVV1 \\*
   & = \left(
\begin{array}{ccccccc} \cline{1-2}
	 \multicolumn{1}{|c}{r_1} & \multicolumn{1}{c|}{x_2} & 0 & 0 & \cdots & 0 & 0 \\ \cdashline{2-3}
	 \multicolumn{1}{|c}{0}    &  \multicolumn{1}{:c|}{r_2} & 0 & \multicolumn{1}{:c}{0} & \cdots & 0 & 0 \\ \cline{1-4}
	  0 & \multicolumn{1}{:c}{0} & \multicolumn{1}{|c}{r_1} & \multicolumn{1}{:c|}{x_2}  & \cdots& 0 & 0  \\ \cdashline{2-3}
	  0 & 0 &  \multicolumn{1}{|c}{0}    & \multicolumn{1}{c|}{r_2} &\cdots &  0 & 0  \\ \cline{3-4}
	  \vdots  &  \vdots & \vdots  &  \vdots   & \ddots &  \vdots  &  \vdots   \\	    \cline{6-7}
	  0 & 0 & 0 & 0 &\cdots & \multicolumn{1}{|c}{r_1} & \multicolumn{1}{c|}{x_2} \\
	  0 & 0 & 0 & 0 &\cdots &  \multicolumn{1}{|c}{0}    & \multicolumn{1}{c|}{r_2} \\ \cline{6-7}
\end{array}
\right)\,.
\end{align}

}

\noindent
\underline{Step 2:}\\
Note that the submatrix $\submat{\mathcal{T}_1^{(1)}}23$ is
$
\left(
      \begin{array}{:cc:} \hdashline
	1 & 0 \\
	  0 & 1 \\ \hdashline
      \end{array} \right)
$. Thus, according to \propertyref{prop:lenI} we can perform 1-GMD on the corresponding elements of the matrix of user 2, $\submat{\TTT_2^{(1)}}23$, without changing $\submat{\TTT_1^{(1)}}23$ :
  \begin{align}
    \Ul22\left(
    \begin{array}{cc}
      r_2 & 0 \\
      0	    & r_1 \\
    \end{array}
    \right)
    \Vl2
    =
    \left(
    \begin{array}{cc}
      1 & x_2^{(2)} \\
      0	    & 1 \\
    \end{array}
    \right)\,.
  \end{align}
Hence, by defining
\col{
\begin{align}
&\UUUU22 \\*
&\triangleq \embb{2N}{\Ul22} {\subm{2}{3}\subm{4}{5}\cdots\subm{2N-2}{2N-1}} \,,\\
&\VVVV2 \\*
&\triangleq \embb{2N}{\Vl2} {\subm{2}{3}\subm{4}{5}\cdots\subm{2N-2}{2N-1}}  \,,
\end{align}
}
{
\begin{align}
\UUUU22 &\triangleq \embb{2N}{\Ul22} {\subm{2}{3}\subm{4}{5}\cdots\subm{2N-2}{2N-1}} \,,\\
\VVVV2  &\triangleq \embb{2N}{\Vl2} {\subm{2}{3}\subm{4}{5}\cdots\subm{2N-2}{2N-1}}  \,,
\end{align}
}
and applying them to $\mathcal{T}_1^{(1)}$ and $\mathcal{T}_2^{(1)}$,
we attain:
\col{
\begin{align}
   \TTT_1^{(2)} & =   \left(\VVVV2\right)^\dagger \TTT_1^{(1)} \VVVV2
\col{\\}{\\*}
&  =
    \left(\VVVV2\right)^\dagger \UUUU11 \AAA_1 \VVVV1\VVVV2 \col{\\}{\\*}
    & = \left(
\begin{array}{ccccccc}
	 \multicolumn{1}{c}{1} & \multicolumn{1}{c}{\tx_1}     & \multicolumn{1}{c}{*}     & \multicolumn{1}{c}{0}      & \cdots & \multicolumn{1}{c}{0} & 0 \\\cline{2-6}
	 \multicolumn{1}{c}{0}       & \multicolumn{1}{|c}{1} & \multicolumn{1}{c}{0}     & \multicolumn{1}{c}{*}      & \cdots & \multicolumn{1}{c|}{0} & 0 \\
	  0                          & \multicolumn{1}{|c}{0}       & \multicolumn{1}{c}{1} & \multicolumn{1}{c}{\tx_1}  & \cdots&  \multicolumn{1}{c|}{0} & 0  \\
	  0                          & \multicolumn{1}{|c}{0}       & \multicolumn{1}{c}{0}       & \multicolumn{1}{c}{1} &\cdots &   \multicolumn{1}{c|}{0} & 0  \\
	  \vdots                     & \multicolumn{1}{|c}{\vdots}  & \vdots  &  \vdots   & \ddots &   \multicolumn{1}{c|}{\vdots}  &  \vdots   \\	
	  0                          & \multicolumn{1}{|c}{0} & 0 & 0 &\cdots & \multicolumn{1}{c|}{1} & \multicolumn{1}{c}{\tx_1} \\\cline{2-6}
	  0                          & \multicolumn{1}{c}{0} & 0 & 0 &\cdots &  \multicolumn{1}{c}{0}    & \multicolumn{1}{c}{1} \\
\end{array}
\right)\,,
\\
    \TTT_2^{(2)} &=   \UUUU22 \TTT_2^{(1)} \VVVV2
\\ & =
    \UUUU22\UUUU21 \AAA_2 \VVVV1\VVVV2 \\
    & =\left(
\begin{array}{ccccccc}
	 \multicolumn{1}{c}{r_1} & \multicolumn{1}{c}{\tx_2}     & \multicolumn{1}{c}{*}     & \multicolumn{1}{c}{0}      & \cdots & \multicolumn{1}{c}{0} & 0 \\\cline{2-6}
	 \multicolumn{1}{c}{0}       & \multicolumn{1}{|c}{1} & \multicolumn{1}{c}{{x_2^{(2)}}}     & \multicolumn{1}{c}{*}      & \cdots & \multicolumn{1}{c|}{0} & 0 \\
	  0                          & \multicolumn{1}{|c}{0}       & \multicolumn{1}{c}{1} & \multicolumn{1}{c}{\tx_2}  & \cdots&  \multicolumn{1}{c|}{0} & 0  \\
	  0                          & \multicolumn{1}{|c}{0}       & \multicolumn{1}{c}{0}       & \multicolumn{1}{c}{1} &\cdots &   \multicolumn{1}{c|}{0} & 0  \\
	  \vdots                     & \multicolumn{1}{|c}{\vdots}  & \vdots  &  \vdots   & \ddots &   \multicolumn{1}{c|}{\vdots}  &  \vdots   \\	
	  0                          & \multicolumn{1}{|c}{0} & 0 & 0 &\cdots & \multicolumn{1}{c|}{1} & \multicolumn{1}{c}{\tx_2} \\\cline{2-6}
	  0                          & \multicolumn{1}{c}{0} & 0 & 0 &\cdots &  \multicolumn{1}{c}{0}    & \multicolumn{1}{c}{r_2} \\
\end{array}
\right)\,.
\end{align}
}
{
\begin{align}
   \TTT_1^{(2)} & =   \left(\VVVV2\right)^\dagger \TTT_1^{(1)} \VVVV2
\col{\\}{\\*}
&  =
    \left(\VVVV2\right)^\dagger \UUUU11 \AAA_1 \VVVV1\VVVV2 \col{\\}{\\*}
    & = \left(
\begin{array}{ccccccc}
	 \multicolumn{1}{c}{1} & \multicolumn{1}{c}{x_1}     & \multicolumn{1}{c}{*}     & \multicolumn{1}{c}{0}      & \cdots & \multicolumn{1}{c}{0} & 0 \\\cline{2-6}
	 \multicolumn{1}{c}{0}       & \multicolumn{1}{|c}{1} & \multicolumn{1}{c}{0}     & \multicolumn{1}{c}{*}      & \cdots & \multicolumn{1}{c|}{0} & 0 \\
	  0                          & \multicolumn{1}{|c}{0}       & \multicolumn{1}{c}{1} & \multicolumn{1}{c}{x_1}  & \cdots&  \multicolumn{1}{c|}{0} & 0  \\
	  0                          & \multicolumn{1}{|c}{0}       & \multicolumn{1}{c}{0}       & \multicolumn{1}{c}{1} &\cdots &   \multicolumn{1}{c|}{0} & 0  \\
	  \vdots                     & \multicolumn{1}{|c}{\vdots}  & \vdots  &  \vdots   & \ddots &   \multicolumn{1}{c|}{\vdots}  &  \vdots   \\	
	  0                          & \multicolumn{1}{|c}{0} & 0 & 0 &\cdots & \multicolumn{1}{c|}{1} & \multicolumn{1}{c}{x_1} \\\cline{2-6}
	  0                          & \multicolumn{1}{c}{0} & 0 & 0 &\cdots &  \multicolumn{1}{c}{0}    & \multicolumn{1}{c}{1} \\
\end{array}
\right)\,,
\\
    \TTT_2^{(2)} &=   \UUUU22 \TTT_2^{(1)} \VVVV2
\\* & =
    \UUUU22\UUUU21 \AAA_2 \VVVV1\VVVV2 \\*
    & =\left(
\begin{array}{ccccccc}
	 \multicolumn{1}{c}{r_1} & \multicolumn{1}{c}{x_2}     & \multicolumn{1}{c}{*}     & \multicolumn{1}{c}{0}      & \cdots & \multicolumn{1}{c}{0} & 0 \\\cline{2-6}
	 \multicolumn{1}{c}{0}       & \multicolumn{1}{|c}{1} & \multicolumn{1}{c}{{x_2^{(2)}}}     & \multicolumn{1}{c}{*}      & \cdots & \multicolumn{1}{c|}{0} & 0 \\
	  0                          & \multicolumn{1}{|c}{0}       & \multicolumn{1}{c}{1} & \multicolumn{1}{c}{x_2}  & \cdots&  \multicolumn{1}{c|}{0} & 0  \\
	  0                          & \multicolumn{1}{|c}{0}       & \multicolumn{1}{c}{0}       & \multicolumn{1}{c}{1} &\cdots &   \multicolumn{1}{c|}{0} & 0  \\
	  \vdots                     & \multicolumn{1}{|c}{\vdots}  & \vdots  &  \vdots   & \ddots &   \multicolumn{1}{c|}{\vdots}  &  \vdots   \\	
	  0                          & \multicolumn{1}{|c}{0} & 0 & 0 &\cdots & \multicolumn{1}{c|}{1} & \multicolumn{1}{c}{x_2} \\\cline{2-6}
	  0                          & \multicolumn{1}{c}{0} & 0 & 0 &\cdots &  \multicolumn{1}{c}{0}    & \multicolumn{1}{c}{r_2} \\
\end{array}
\right)\,.
\end{align}
}
Now, to get the desired decomposition we need to ``extract'' the middle submatrices (by multiplying on both sides
by $\extt{2N}{2:2N-1}$,
as explained in \remref{remark:extract}).

Thus, by defining
\begin{align}
\qqquad\VVV &\triangleq \VVVV1\VVVV2\extt{2N}{2:2N-1} \\
\left(\UUU_1\right)^\dagger &\triangleq \left(\extt{2N}{2:2N-1}\right)^\dagger\left(\VVVV2\right)^\dagger \UUUU11 \\
\left(\UUU_2\right)^\dagger &\triangleq \left(\extt{2N}{2:2N-1}\right)^\dagger\UUUU22\UUUU21
\end{align}
we arrive at the desired result.

\begin{remark}
 \label{rem:nearly_k_jet}
  It was shown in \lemref{lem:k_to_k_plus_one} that $K$-GMD is equivalent to $(K+1)$-JET.
  Hence, nearly-optimal $(K+1)$-JET can be obtained with the same parameters as in \thrmref{thm:n_n_asymptotical}.
    Alternatively, an explicit algorithm for $(K+1)$-JET can be obtained by performing
the $K$-GMD algorithm as in \appref{app:proof_general}, where in the first step,
instead of performing $1$-GMD on the matrix $A_1$, $2$-JET on the matrices
$A_1$ and $A_2$ is performed, and similarly, in step $\ell$
instead of performing $1$-GMD on the matrix $\submatt{\mT_{\ell}^{({\ell})(1)}}{1}{n}$,
$2$-JET on the matrices
$\submatt{\mT_{\ell}^{({\ell})(1)}}{1}{n}$
and
$\submatt{\mT_{l+1}^{({\ell})(1)}}{1}{n}$ is performed.
\end{remark}

\vspace{\baselineskip}

%% file: extentions.tex
\subsection{Time-Varying Channel}
\label{ss:time_varying}
\input{time_varying}

\subsection{Different Decoding Orders}
\label{ss:upper_lower}
\input{upper_lower}

\subsection{Block GTD}
\label{ss:block_gmd}
\input{block_gmd}

%% file: time_varying.tex
Throughout this paper, we have considered the problem of broadcasting the same information to $K$ different users
over \emph{static} Gaussian MIMO channels, described by the matrices $H_k$.
As mentioned in \secref{s:channel_model}, this problem is equivalent to the
problem of transmission over a compound channel~\cite{Dobrushin59,BlackwellBreimanThomasian59,Wolfowitz60},
where a transmitter wishes to convey information to a single receiver over a MIMO channel,
which can take one out of $K$ realizations, the set of which is known at the transmitter,
but the exact realization is known only to the receiver (but not to the transmitter) and remains constant throughout the whole transmission.

For this problem, the schemes of \secref{ss:multi_user_scheme} and \secref{ss:space_time_nearly}
may be readily used.
These schemes may further be extended to the case where the channel varies in time.
For $K=2$, using the JET-based scheme, any arbitrary sequence of channel realizations (within the set $\{H_1,H_2\}$) may be accommodated,
provided that this sequence is known to the receiver. The transmitter, in this case, is identical to the one in the ``compound scenario'',
whereas the receiver needs to apply to its received signal, at each time instant, $U_1^\dagger$ or $U_2^\dagger$,
depending on the channel realization at this time instant ($H_1$ or $H_2$, respectively).
The successive decoding process needs to be modified as follows: The last sub-channel is interference-free, as in the ``compound scenario'',
and therefore its interference can be subtracted of the other sub-channels; however, its components in the other sub-channels, differ with the realizations
at each time instant (``off-diagonal'' coefficients differ with $H_k$, unlike the diagonal ones which are equal to all channel realizations).
The successive decoding process of the other sub-messages needs to be modified in a similar manner.

Note however that for $K>2$ channel realizations, more channel uses need to be processed together,
in general, as explained in \secref{ss:space_time_nearly}.
In the time-varying scenario, this implies that, in order to use the schemes of \secref{ss:space_time_nearly}, the channel needs to be constant in time for a number of time instants which equals the number of channel uses that are jointly processed together. 
This requirement is shared by the space--time schemes of \cite{Alamouti} and \cite{TarokhJafarkhaniCalderbank_STBC}.

%% file: upper_lower.tex
In the above sections, we discussed the simultaneous decomposition of several matrices
into \emph{upper triangular} forms. In terms of the transmission scheme described
in \secref{ss:multi_user_scheme}, all the receivers decode the messages in the same order
(starting with the last component; ending with the first one).

This scheme can be generalized,  if we allow each receiver to choose \emph{its own} order of decoding.
It turns out that this generalized scheme can achieve rates which are strictly higher than the rates achieved using
the ordinary scheme (where all the decoders use the same order of decoding).

In the case of two transmit antennas, the channel canonical matrices
\eqref{eq:G_matrix_K_users}
 are $2 \times 2$ matrices. Thus, allowing different decoding orders means that some matrices are transformed into
 \emph{upper triangular} matrices, whereas the others~--- into \emph{lower triangular} matrices,
 where all the resulting matrices have equal diagonal values.
The following theorem is proved using a similar technique to the one used for the proof of \thrmref{thm:perfect_two_on_two_complex}.
Again, as explained in \remref{rem:equal_det}, we can
assume without loss of generality that both matrices have determinants equal to $1$.

\begin{thm} \label{thm:upper_lower_complex}
Let $A_1$ and $A_2$ be \emph{complex-valued} $2 \times 2$ matrices with determinants equal to $1$.
Then, there exist three complex-valued $2 \times 2$ unitary matrices $U_1$, $U_2$, and $V$, such that
\begin{align}
      {{  \left( U_1  \right)^\dagger    }} A_1 V =  \left( \begin{array}{cc}
		      1 & * \\ 0 & 1
                   \end{array} \right)
\end{align}
and
\begin{align}
      {{  \left( U_2  \right)^\dagger    }} A_2 V =  \left( \begin{array}{cc}
		      1 & 0 \\ * & 1
                   \end{array} \right)\,,
\end{align}
if and only if the following inequality is satisfied:
\begin{align} \label{eq:upper_lower_complex_2}
      F_2 \left(  A_1^\dagger A_1 - I ,  A_2^\dagger A_2 - I  \right) \geq 0 \,,
\end{align}
where
\begin{align}
    F_2(S_1,S_2) \triangleq \mathrm{det} \big(
	    S_1 S_2  - \adj S_2 \adj S_1
\big) \,.
\end{align}
\end{thm}
The proof is given in \appref{app:upper_lower_proof}.

This result can be easily generalized, as stated in the following corollary.

\begin{corol}
Let $A_1$ and $A_2$ be \emph{complex-valued} $2 \times 2$ matrices with determinants equal to $1$, and let $r>0$.
Then there exist three complex-valued $2 \times 2$ unitary matrices $U_1$, $U_2$, and $V$, such that
\begin{align}
      {{  \left( U_1  \right)^\dagger    }} A_1 V =  \left( \begin{array}{cc}
		      r & * \\ 0 & 1/r
                   \end{array} \right)
\end{align}
and
\begin{align}
      {{  \left( U_2  \right)^\dagger    }} A_2 V =  \left( \begin{array}{cc}
		      r & 0 \\ * & 1/r
                   \end{array} \right)
\end{align}
if and only if the following conditions are satisfied:
\begin{align}
      \det \left( A_1^\dagger A_1 - r^2 I \right)  & \leq 0 \\
      \det \left( A_2^\dagger A_2 - 1/r^2 I \right)  & \leq 0 \\
      F_2 \left(  A_1^\dagger A_1 - r^2 I ,  A_2^\dagger A_2 - 1/r^2 I  \right) & \geq 0 \,.
\end{align}
The proof of the corollary follows along the same lines as that of \thrmref{thm:upper_lower_complex} with obvious modifications.
\end{corol}

Recall the ``degrees-of-freedom mismatch'' scenario of Examples \ref{DOF-example} and \ref{DOF-example_3users}.
The compound capacity in this case is achieved by a white input covariance matrix.
The corresponding channel canonical matrices \eqref{eq:G_matrix_K_users}, are
\begin{align}
    G_1 & = \left( \begin{array}{cc}
	2^{\Cptp/2} & 0 \\ 0 & 1 \end{array}
\right) \,, \\
G_2  & = \left(\begin{array}{cc}
	1 & 0 \\ 0 & 2^{\Cptp/2} \end{array}
\right) \,,
\\
G_3  &= \left(\begin{array}{cc}
	2^{\Cptp/4} & 0 \\ 0 & 2^{\Cptp/4} \end{array}
\right)
\,.
\end{align}
Since $G_3$ is a scaled identity matrix, performing $3$-GMD on these three matrices is in fact equivalent to $2$-GMD of $G_1$ and $G_2$, which is not possible according to  \thmref{thm:perfect_two_on_two_complex}.
However, if we allow generalized triangularization ~--- namely, receiver $1$ transforms the channel into upper triangular form, whereas receiver $2$ transforms it into lower triangular form ~--- then the decomposition is possible according to \thmref{thm:upper_lower_complex}, using the following precoding matrix:
        \begin{align}
            V =
            \sqrt{\frac{1}{2^{\Cptp/2}+1}}
            \left(
              \begin{array}{cc}
                1 & 2^{\Cptp/4} \\
                2^{\Cptp/4} & -1 \\
              \end{array}
            \right)  \,,
        \end{align}
which gives rise, in turn, to the following triangular matrices:
\begin{align}
    T_1 & = \left( \begin{array}{cc}
	2^{\Cptp/4} & \frac{2^{\Cptp}-1}{2^{\Cptp/2}+1} \\ 0 & 2^{\Cptp/4} \end{array}
\right) \,, \\
T_2  & = \left(\begin{array}{cc}
	2^{\Cptp/4} & 0 \\ - \frac{2^{\Cptp}-1}{2^{\Cptp/2}+1} & 2^{\Cptp/4} \end{array}
\right) \,,
\\
T_3  &= \left(\begin{array}{cc}
	2^{\Cptp/4} & 0 \\ 0 & 2^{\Cptp/4} \end{array}
\right)
\,.
\end{align}

%% file: block_gmd.tex
There are certain cases, where triangularity of the resulting matrices is not necessary and block-triangular forms, with
blocks satisfying certain relations between their determinants, suffice.
In these cases we are interested primarily in deriving information-theoretic bounds,
rather than constructing practical communication schemes.

This is the case for the Gaussian MIMO joint source--channel coding (JSCC) problem,
where we wish to convey a scalar Gaussian source over Gaussian MIMO links,
having different capacities.
In this case, pure digital transmission, as in Sections \ref{s:multi_user} and \ref{s:space_time}, is not optimal, as it is restricted to the minimum of the capacities of the different MIMO links.
Indeed, better performance may be achieved,
using a scheme which better adapts to the different capacities of the different channel links.
For more information see \cite[Sec. IV]{STUD:SP}.

For this purpose, we first extend the GTD, discussed in \secref{ss:decomposition},
for a block-triangular form,
after which we apply this result in the derivation of a block joint triangularization.

\begin{thm}[Block GTD]
\label{thm:BlockGTD}
    Let $A$ be an $n \times n$ full-rank matrix. Then, it can be decomposed into a
block upper triangular form ($1 \leq M \leq n$):
    \begin{align}
    \label{eq:BlockTriangular}
        \tilde R = \left(
              \begin{array}{cccc}
                \tilde R_{11} & \tilde R_{12} & \cdots & \tilde R_{1M} \\
                0 & \tilde R_{22} & \cdots & \tilde R_{2M} \\
                \vdots & \ & \ddots & \vdots \\
                0 & \cdots & 0 & \tilde R_{MM} \\
              \end{array}
            \right) \,,
    \end{align}
    where $\tilde R_{j\ell}$ are $n_j \times n_\ell$ blocks, and the matrices $\tilde{R}_{mm}$ have prescribed determinants $\det\left( \tilde R_{mm} \right)$,
    such that \mbox{$\sum_{m=1}^M n_m = n$},
    if and only if
    \begin{align}
    \label{eq:BlockDiag_WeylConds}
        \prod_{m=1}^q \left| \det \left( \tilde R_{{p_m} {p_m}} \right) \right| &\leq \prod_{j=1}^{\sum_{m=1}^q n_{p_m}} \sigma_j
    \end{align}
    for all $q=1,2,...,M$, and
    \begin{align}
    \label{eq:BlockDiag_WeylCondsEquality}
        \prod_{m=1}^M \left| \det \left( \tilde R_{{p_m} {p_m}} \right) \right| &= \prod_{j=1}^{n} \sigma_j \,,
    \end{align}
    where $\sigma_j$ are the singular values of $A$ ordered non-increasingly,
    $\left\{ {p_m} \right\}_{m=1}^M$ are the indices satisfying
    \begin{align}
        d_{p_1} \geq d_{p_2} \geq \cdots \geq d_{p_M} \,,
    \end{align}
    and
    \begin{align}
        d_{m} \triangleq \sqrt[n_m]{\left| \det \left( \tilde R_{mm} \right) \right|} \,, \qquad m=1,...,M \,.
    \end{align}
\end{thm}

\ver{
Before we prove this theorem, we need the following lemma.
}{
The proof of this theorem is provided in \cite[Appendix K]{JET:TechReport:SeveralUsers2013}
and is based upon the following lemma, which is proved in \cite[Sec.~VIII-C]{JET:TechReport:SeveralUsers2013}.
}
\begin{lemma}[GTD with Multiplicities]
\label{lem:GTD_with_multiplicities}
    Let $A$ be an $n \times n$ full-rank matrix with singular values $\left\{ \sigma_j \right\}$, ordered non-increasingly.
    Then, it can be decomposed as
    \begin{align}
    \label{eq:append:A=URV'}
        A = U R V^\dagger \,,
    \end{align}
    where $R$ is upper triangular and $U,V$ are unitary,
    if and only if
    \begin{align}
    \label{eq:WeylMultiplicityCond}
        \prod_{m=1}^q r_m^{n_m} &\leq \prod_{j=1}^{\sum_{m=1}^q {n_m}} \sigma_j
    \end{align}
    for every $q$ ($q=1,2,...,M$), and
    \begin{align}
    \label{eq:WeylMultiplicityCondEqual}
        \prod_{m=1}^M r_m^{n_m} &= \prod_{j=1}^n \sigma_j \,,
    \end{align}
    where the absolute values of the diagonal of $R$ take $M$ ($1 \leq M \leq n$) distinct values; these values, ordered non-decreasingly, are denoted by $r_{m}$ (\mbox{$m=1,2,...,M$}), and the number of occurrences (``multiplicity'') of each value~--- by $n_m$.
\end{lemma}
\ver{
The proof of this lemma is given in \appref{app:Proof_GTD_with_multiplicities}.
}{}

Note that this lemma suggests that in case of multiplicities of the absolute values of the desired diagonal entries of the triangular matrix, if those entries take only $M$ different values,
then it suffices to verify only $M$ conditions (1 condition per distinct value),
instead of the $n$ conditions of general GTD.

\begin{proof}[Proof of \thrmref{thm:BlockGTD}]
    Decompose, according to the GMD, every block matrix $\tilde R_{mm}$ in \eqref{eq:BlockTriangular} laying on the main diagonal, as
    \begin{align}
       \tilde R_{mm} = U_{mm} T_{mm} V_{mm}^\dagger \,, \qquad m=1,2,...,K\,,
    \end{align}
    where $U_{mm}$ and $V_{mm}$ are unitary and $T_{mm}$ is upper triangular with constant diagonal entries which are equal to
    \begin{align}
        \left[ T_{mm} \right]_j = \sqrt[n_m]{\left| \det \left( \tilde R_{mm} \right) \right|} \triangleq d_{m} \,, \qquad j=1,2,...,n_m \,.
    \end{align}
    Hence, applying the unitary matrices $U^\dagger$ on the left and $V$ on the right, given by
    \begin{align}
        U &= \left(
              \begin{array}{cccc}
                U_{11} & 0 & \cdots & 0 \\
                0 & U_{22} & \cdots & 0 \\
                \vdots & \vdots & \ddots & \vdots \\
                0 & 0 & \cdots & U_{MM} \\
              \end{array}
            \right) \,,
            \\
        V &= \left(
        \begin{array}{cccc}
            V_{11} & 0 & \cdots & 0 \\
            0 & V_{22} & \cdots & 0 \\
            \vdots & \vdots & \ddots & \vdots \\
            0 & 0 & \cdots & V_{MM} \\
        \end{array}
        \right) \,,
    \end{align}
    gives rise to an upper triangular matrix whose diagonal equals to the concatenation of the diagonals of $\left\{ T_{mm} \right\}$.
    Therefore, the task of constructing the decomposition \eqref{eq:BlockTriangular} is equivalent to decomposing $A$ into triangular form with a diagonal that is equal to the concatenation of the diagonals of
    $\left\{ T_{mm} \right\}$.
    Denote the entries of this diagonal, reordered non-increasingly, by $\br(A)$ and the singular values of $A$ by
    $\bsigma(A)$.
    Then, the aforementioned decomposition is possible if and only if Weyl's condition
    \begin{align}
        \bsigma(A) \succeq \br(A)
    \end{align}
    is satisfied, which in turn is satisfied if and only if \eqref{eq:BlockDiag_WeylConds} and \eqref{eq:BlockDiag_WeylCondsEquality} hold, according to \lemref{lem:GTD_with_multiplicities}.
\end{proof}

\begin{corol}[Joint Block Triangularization]
\label{corol:Block-JET}
    Let $A_1$ and $A_2$ be two full-rank $n \times n$ complex-valued matrices. Then $A_1$ and $A_2$ can be jointly decomposed into block-triangular forms
    \begin{align}
        A_1 &= U_1 \tilde R_1 V^\dagger \\
        A_2 &= U_2 \tilde R_2 V^\dagger \,,
    \end{align}
    where $U_k$ and $V$ are unitary, and $\tilde R_k$ are block-triangular:
    \begin{align}
        \tilde R_k = \left(
              \begin{array}{cccc}
                \tilde R^{(k)}_{11} & \tilde R^{(k)}_{12} & \cdots & \tilde R^{(k)}_{1M} \\
                0 & \tilde R^{(k)}_{22} & \cdots & \tilde R^{(k)}_{2M} \\
                \vdots & \ddots & \ddots & \vdots \\
                0 & \cdots & 0 & \tilde R^{(k)}_{MM} \\
              \end{array}
            \right) \,, \qquad k=1,2 \,,
    \end{align}
    where corresponding blocks $\tilde R^{(1)}_{j\ell}$ and $\tilde R^{(2)}_{j\ell}$ have the same dimensions $n_j \times n_\ell$, such that $\sum_m n_m = n$, and prescribed determinant ratios of the blocks on the main diagonal,
    $\det \left( \tilde R^{(1)}_{mm} \right) \big/ \det \left( \tilde R^{(2)}_{mm} \right)$
    if and only if
    \begin{align}
    \label{eq:BlockDiag_WeylConds1}
        \prod_{m=1}^q \left| \det \left( \tilde R^{(1)}_{p_m p_m} \right) \middle/ \det \left( \tilde R^{(2)}_{p_m p_m} \right) \right|
        &\leq \prod_{j=1}^{\sum_{\ell=1}^q n_{k_\ell}} \mu_j
    \end{align}
    for all $q=1,2,...,M$, and
    \begin{align}
    \label{eq:BlockDiag_WeylCondsEquality1}
        \prod_{m=1}^M \left| \det \left( \tilde R^{(1)}_{p_m p_m} \right) \middle/ \det \left( \tilde R^{(2)}_{p_m p_m} \right) \right|
        &= \prod_{j=1}^{n} \mu_j \,,
    \end{align}
    where $\mu_j$ are the generalized singular values \cite{VanLoan76,GolubVanLoan3rdEd} of $(\tilde R_1,\tilde R_2)$ ordered non-increasingly, $\left\{ p_m \right\}_{m=1}^K$ are the indices satisfying
    \begin{align}
        d_{p_1} \geq d_{p_2} \geq \cdots \geq d_{p_M} \,,
    \end{align}
    and
    \begin{align*}
        d_{m} \triangleq \sqrt[n_m]{\left| \det \left( \tilde R^{(1)}_{p_m p_m} \right) \middle/ \det \left( \tilde R^{(2)}_{p_m p_m} \right) \right|} \,, \: m=1,2,...,M \,.
    \end{align*}
\end{corol}

\ver{
\begin{proof}
     The proof is similar to the proof of ~\cite[Theorem~1]{STUD:SP}, by replacing the GTD by the block-GTD of \thrmref{thm:BlockGTD} and using the fact that the inverse of a square block-triangular matrix is a matrix of the same block-triangular form with blocks on its main diagonal which are equal to the inverses of the original matrix, and the fact that multiplying two square block-triangular matrices with the same block dimensions results in a matrix of the same block-triangular form with blocks on its main diagonal which are equal to the product of the corresponding blocks of the multiplied matrices.
\end{proof}
}{
Again, the proof of this theorem is provided in \cite[Sec.~VIII-C]{JET:TechReport:SeveralUsers2013}.
}

%% file: discussion.tex
In this work, we derived new joint triangularizations of several matrices.
Specifically, we were interested in designing triangular matrices having equal or constant diagonals, by applying unitary operations,
for the construction of a practical scheme for the common-message BC problem, that approaches its capacity.
We derived conditions for the existence of such decompositions, for specific cases; conditions for general matrices~---
remain unknown.

For the general case (even when such exact decompositions are not possible),
we introduced a decomposition that nearly achieves this goal for time-extended variants of the channel matrices.
However, the number of time extensions required, for this proposed decomposition,
grows rapidly with the number of jointly-decomposed matrices.
Nonetheless, numerical evidence suggests that this number of required time extensions, can be greatly reduced,
and calls for further research.

%% file: det_adj_complex_proof.tex
Before we turn to the proof of the lemma, we
introduce the following lemma, the proof of which is relegated to \appref{app:F_invariant}.

\begin{lemma} \label{lem:F_invariant}
Let $S_1$ and $S_2$ be $n \times n$ complex-valued matrices,
and let $U$ be an $n \times n$ unitary matrix.
Then,
\begin{align}
	F_1 \left(  U^\dagger S_1 U ,  U^\dagger S_2 U  \right)
 = F_1(S_1,S_2) \,.
\end{align}
\end{lemma}

Now, let $S_1$ and $S_2$ be two complex-valued $2 \times 2$ Hermitian matrices.
Without loss of generality, we can restrict ourselves to vectors $\bv \in \mathbb{C}^2$ that have
a Euclidean norm of $1$. Namely, we are looking for a necessary and sufficient condition for
the existence of a solution $v \in \mathbb{C}^2$ to the following three equations:
\begin{subequations} \label{eq:lemma_decomposition}
\noeqref{eq:lemma_decomposition_S1_S2_v=1}
\begin{align}
 & \bv^\dagger S_1 \bv = 0   \label{eq:lemma_decomposition_S1}      \\
 & \bv^\dagger S_2 \bv = 0   \label{eq:lemma_decomposition_S2}       \\
 & \left\| \bv \right\| = 1 \,. \label{eq:lemma_decomposition_S1_S2_v=1}
\end{align}
\end{subequations}
First, note that if $\det(S_1)>0$ then $S_1$ is either positive definite or negative definite, and in both cases
there is no non-zero solution $\bv$ to \eqref{eq:lemma_decomposition_S1}. Similarly, if $\det(S_2)>0$ there is no non-zero
solution to \eqref{eq:lemma_decomposition_S2}. Therefore, from now on we can assume that $\det(S_1) \leq 0$ and $\det(S_2) \leq 0$.

Note that for any $2 \times 2$ unitary matrix $U$, the decomposition
\eqref{eq:lemma_decomposition} is equivalent to
\begin{subequations}
\label{eq:the_three_equations}
\noeqref{eq:the_three_equations:1,eq:the_three_equations:2,eq:the_three_equations:3}
\begin{align}
\label{eq:the_three_equations:1}
 \tilde{\bv}^\dagger \tilde{S}_1 \tilde{\bv} & = 0  \\
\label{eq:the_three_equations:2}
 \tilde{\bv}^\dagger \tilde{S}_2 \tilde{\bv} & = 0  \\
\label{eq:the_three_equations:3}
 \left\| \tilde{\bv} \right\| & = 1 \,,
\end{align}
\end{subequations}
where
\begin{align}
  \tilde{\bv} &\triangleq U^\dagger \bv \\
  \tilde{S}_1 &\triangleq U^\dagger S_1 U \\
  \tilde{S}_2 &\triangleq U^\dagger S_2 U \,.
\end{align}
Since $S_k$ are Hermitian, so are $\tilde{S}_k$.

Also, according to \lemref{lem:F_invariant}, \eqref{eq:lemma_tnayal} is equivalent to
\begin{align}
      \det(\tilde{S}_1) & \leq 0 \\*
      \det(\tilde{S}_2) & \leq 0 \\*
      F_1 \left(  \tilde{S}_1 , \tilde{S}_2  \right) & \geq 0 \,.
\end{align}
Thus, by choosing $U$ that diagonalizes $S_1$, we can assume without loss of generality that $S_1$ is  real valued and diagonal matrix:
\begin{align}
 S_1 &= \left( \begin{array}{cc}
		a_1 & 0 \\ 0 & c_1
             \end{array} \right)
\\
 S_2 &= \left( \begin{array}{cc}
		a_2 & b_2+i\beta_2 \\ b_2-i\beta_2 & c_2
             \end{array} \right) \,,
\end{align}
where $a_1,c_1,a_2,c_2,b_2,\beta_2$ are real-valued.
Denoting
\begin{align}
    \bv = \left( \begin{array}{c}
	    x_1 + i x_2 \\ y_1 + i y_2
      \end{array}  \right) \,,
\end{align}
the three equations \eqref{eq:the_three_equations} become:
\begin{align}
\!\!\!\!\!\!\!\!\!\!
    \left( \begin{array}{cccc}
		1 & 0 & 0 & 1 \\
		a_1 & 0 & 0 & c_1 \\
		a_2 & b_2 & \beta_2 & c_2
           \end{array} \right)
    \left( \begin{array}{c}
	    x_1^2 + x_2^2 \\ 2(x_1y_1 + x_2y_2) \\
2(x_2 y_1 - x_1 y_2)
   \\ y_1^2 + y_2^2
           \end{array} \right)
&=
    \left( \begin{array}{c}
	    1 \\ 0 \\ 0
           \end{array} \right) \!\!\!.
 \label{eq:linear_equation_system}
\end{align}
We now consider the following cases.

\paragraph{Case 1} \label{tnayal:case1}
Assume first that $a_1 \neq c_1$ and $b_2 \neq 0$.
Thus, \eqref{eq:linear_equation_system} is equivalent to:
\col{
\begin{align*}
    \left( \begin{array}{c}
	    x_1^2 + x_2^2 \\ 2(x_1y_1 + x_2y_2) \\ 2(x_2 y_1 - x_1 y_2) \\ y_1^2 + y_2^2
           \end{array} \right)
&=
    \overbrace{ \left( \begin{array}{cccc}
		1 & 0 & 0 & 1 \\
		a_1 & 0 & 0 & c_1 \\
		a_2 & b_2 & \beta_2 & c_2 \\
		  0 & 0 & 1 & 0
           \end{array} \right)^{-1}}^{B^{-1}}
    \left( \begin{array}{c}
	    1 \\ 0 \\ 0 \\ t
           \end{array} \right)
\\*
& \triangleq
    \frac{1}{\Delta}
    \left( \begin{array}{c}
	      f_1(t) \\ f_2(t) \\ f_3(t) \\ f_4(t)
           \end{array} \right) \,,
\end{align*}
}{
\begin{align*}
    \left( \begin{array}{c}
	    x_1^2 + x_2^2 \\ 2(x_1y_1 + x_2y_2) \\ 2(x_2 y_1 - x_1 y_2) \\ y_1^2 + y_2^2
           \end{array} \right)
&=
    \overbrace{ \left( \begin{array}{cccc}
		1 & 0 & 0 & 1 \\
		a_1 & 0 & 0 & c_1 \\
		a_2 & b_2 & \beta_2 & c_2 \\
		  0 & 0 & 1 & 0
           \end{array} \right)^{-1}}^{B^{-1}}
    \left( \begin{array}{c}
	    1 \\ 0 \\ 0 \\ t
           \end{array} \right)
\triangleq
    \frac{1}{\Delta}
    \left( \begin{array}{c}
	      f_1(t) \\ f_2(t) \\ f_3(t) \\ f_4(t)
           \end{array} \right) \,,
\end{align*}
}
where $t$ is some real-valued parameter,
$f_1(t)$,$f_2(t)$,$f_3(t)$,$f_4(t)$ are four first-degree polynomials in $t$
(with coefficients that depend on the matrices $S_1,S_2$,
and
where
\begin{align}
	\Delta \triangleq \det B = b_2(c_1 - a_1) \neq 0\,.
\end{align}
Thus, finding a solution $\bv$ to the original problem is \emph{equivalent} to finding a solution
$(x_1,x_2,y_1,y_2,t)$ to the following equations:
\begin{subequations}
\label{eq:appendix_equations1}
\noeqref{eq:appendix_equations1:1,eq:appendix_equations1:2,eq:appendix_equations1:3,eq:appendix_equations1:4}
\begin{align}
\label{eq:appendix_equations1:1}
  x_1^2 + x_2^2    & = \frac{1}{\Delta} f_1(t) \\
\label{eq:appendix_equations1:2}
  2(x_1 y_1 + x_2 y_2)    & = \frac{1}{\Delta} f_2(t) \\
\label{eq:appendix_equations1:3}
  2(x_2 y_1 - x_1 y_2)    & = \frac{1}{\Delta} f_3(t) \\
\label{eq:appendix_equations1:4}
  y_1^2 + y_2^2    & = \frac{1}{\Delta} f_4(t)
\end{align}
\end{subequations}

\begin{assert}
\label{assert:appendix1}
  A solution to \eqref{eq:appendix_equations1} exists
  if and only if the following conditions hold for some $t \in \mathbb{R}$:
  \begin{subequations}
  \label{eq:jcase1}
  \noeqref{eq:jcase1:1,eq:jcase1:2,eq:jcase1:3}
  \begin{align}
\label{eq:jcase1:1}
      \frac{1}{\Delta} f_1(t) & \geq 0 \\*
\label{eq:jcase1:2}
      \frac{1}{\Delta} f_4(t) & \geq 0 \\*
\label{eq:jcase1:3}
      4 f_1(t) f_4(t) & = f_2^2(t) + f_3^2(t)  \,.
  \end{align}
  \end{subequations}
\end{assert}

\begin{proof}[Proof of \assertref{assert:appendix1}]
\label{sss:justification1}
  Construct the following three vectors: $\bp_1=(x_2,-x_1)$, $\bp_2=(y_1,y_2)$, $\bp_3=(x_1,x_2)$.
    Then, 
  \begin{subequations}
  \label{eq:appendix_p_vectors1}
  \noeqref{eq:appendix_p_vectors1:1,eq:appendix_p_vectors1:2,eq:appendix_p_vectors1:3,eq:appendix_p_vectors1:4}
  \begin{align}
\label{eq:appendix_p_vectors1:1}
  ||\bp_1||^2&=||\bp_3||^2=x_1^2+x_2^2\\
\label{eq:appendix_p_vectors1:2}
  ||\bp_2||^2&=y_1^2+y_2^2\\
\label{eq:appendix_p_vectors1:3}
  2\left<\bp_1,\bp_2\right>&= 2(x_2y_1-x_1y_2)\\
\label{eq:appendix_p_vectors1:4}
  2\left<\bp_2,\bp_3\right> &= 2(x_1y_1+x_2y_2) 
  \end{align}
  \end{subequations}
  Note that the l.h.s.\ of \eqref{eq:appendix_equations1} and the r.h.s.\ of \eqref{eq:appendix_p_vectors1} coincide.
  We note that $\bp_3$ and $\bp_1$ are orthogonal.
  Hence, the angles between these vectors satisfy
  \begin{align}
    \cos{\theta_1}&=\frac{\left<\bp_1,\bp_2\right>}{||\bp_1||||\bp_2||}\\
    \cos{\theta_2}&=\frac{\left<\bp_3,\bp_2\right>}{||\bp_3||||\bp_2||}=\frac{\left<\bp_3,\bp_2\right>}{||\bp_1||||\bp_2||} \\
    \cos{\theta_2}&=\cos{(\pm{\frac{\pi}{2}}-\theta_1)}=\pm \sin{\theta_1} \,.
  \end{align}

  Thus, a solution to \eqref{eq:appendix_equations1} exists if and only if
  \begin{subequations}
  \label{eq:pi_conds}
  \noeqref{eq:pi_conds:1,eq:pi_conds:2,eq:pi_conds:3}
  \begin{align}
  \label{eq:pi_conds:1}
  ||\bp_1||^2&\geq0\\
  \label{eq:pi_conds:2}
  ||\bp_2||^2&\geq0\\
  \label{eq:pi_conds:3}
  \!\!\!\!\!\!\!\!\!\!\!\!\!\cos^2\theta_1+\sin^2\theta_1&=\frac{\left<\bp_1,\bp_2\right>^2}{||\bp_1||^2||\bp_2||^2} + \frac{\left<\bp_3,\bp_2\right>^2}{||\bp_1||^2||\bp_2||^2}=1 
  .
  \end{align}
  \end{subequations}
  where \eqref{eq:pi_conds:3} is equivalent to 
  \begin{align}
  0&=4||\bp_1||^2||\bp_2||^2-\left(2\left<\bp_1,\bp_2\right>\right)^2-\left(2\left<\bp_3,\bp_2\right>\right)^2 ,
  \end{align}
  which is equivalent, in turn, to \eqref{eq:jcase1}.
\end{proof}

By definition, and using \eqref{eq:linear_equation_system}, we have $(f_1(t)+f_4(t)) = \Delta$. Therefore, the three conditions of \eqref{eq:jcase1} are equivalent to the single condition
\begin{align}
      4f_1(t) f_4(t) - f_2^2(t) - f_3^2(t) = 0 \,.
\end{align}
This is a quadratic equation in $t$:
\begin{align}
      a t^2 + bt + c = 0 \,,
\end{align}
where the constants $a,b,c$ depend on the matrices $S_1,S_2$ as follows:
\begin{subequations}
\label{eq:parabolic_inequality_constants}
\noeqref{eq:parabolic_inequality_constants:1,eq:parabolic_inequality_constants:2,eq:parabolic_inequality_constants:3}
\begin{align} 
\label{eq:parabolic_inequality_constants:1}
a& \triangleq -(a_1-c_1)^2(b_2^2+\beta_2^2)\\
\label{eq:parabolic_inequality_constants:2}
b& \triangleq 2 \beta_2 (a_2c_1-a_1c_2)(a_1-c_1) \\
\label{eq:parabolic_inequality_constants:3}
c& \triangleq -4 a_1 c_1 b_2^2 - (a_2c_1-a_1c_2)^2 \,.
\end{align}
\end{subequations}
Note that since $a_1 \neq c_1$ and $b_2 \neq 0$, the coefficient $a$ is \emph{strictly} negative.
Therefore, a necessary and sufficient condition for the existence of a solution is for the discriminant to be non-negative:
\begin{align}
      b^2 - 4 a c \geq 0 \,.
\end{align}
A direct calculation shows that
\begin{align}
      b^2 - 4 a c = 4 \Delta^2  F_1(S_1,S_2) \,,
\end{align}
where
\begin{align}
      F_1(S_1,S_2) \triangleq \mathrm{det} \left( S_1 \adj(S_2) - S_2 \adj(S_1)    \right) \,,
\end{align}
which completes the proof for this case.

\paragraph{Case 2} \label{tnayal:case2}
Assume now that $a_1=c_1$. Since we assumed $\det(S_1) \leq 0$, this means
that $a_1=c_1=0$, namely, $S_1=0$.
In this case we have
\begin{align}
      F_1(S_1,S_2) = F_1(0,S_2) = 0 \,.
\end{align}
Thus, condition \eqref{eq:condition3} holds.
Since we assumed that $\det(S_2) \leq 0$, $S_2$ has
one non-negative eigenvalue and one non-positive eigenvalue, therefore there
necessarily exists  $\bv$ with norm $1$ such that $\bv^\dagger S_2 \bv = 0$,
and therefore there exists a solution to the equations in \eqref{eq:lemma_decomposition}.

\paragraph{Case 3} \label{tnayal:case3}
Next, assume that $a_1 \neq c_1$, $b_2=0$, and $\beta_2 \neq 0$.
Thus, \eqref{eq:linear_equation_system} becomes
\begin{align} \label{eq:case_3_equations}
    \left( \begin{array}{ccc}
		1 &  0 & 1 \\
		a_1 &  0 & c_1 \\
		a_2 &  \beta_2 & c_2
           \end{array} \right)
    \left( \begin{array}{c}
	    x_1^2 + x_2^2 \\
2(x_2 y_1 - x_1 y_2)
   \\ y_1^2 + y_2^2
           \end{array} \right)
&=
    \left( \begin{array}{c}
	    1 \\ 0 \\ 0
           \end{array} \right) \!,
\end{align}
which reduces to
\col{
% \begin{subequations}
% \noeqref{eq:appendix_equations2:1,eq:appendix_equations2:2}
\begin{align}
\label{eq:appendix_equations2}
% \label{eq:appendix_equations2:1}
 \left( \begin{array}{c}
	    x_1^2 + x_2^2 \\
2(x_2 y_1 - x_1 y_2)
   \\ y_1^2 + y_2^2
           \end{array} \right)
&=
% \label{eq:appendix_equations2:2}
 \left( \begin{array}{c}
f_5 \\
f_6   \\
f_7           \end{array} \right)
\\ &\triangleq
\frac{1}{(a_1 - c_1)\beta_2}
 \left( \begin{array}{c}
-\beta_2 c_1 \\
a_2 c_1 - a_1 c_2   \\
a_1 \beta_2           \end{array} \right) \,.
\end{align}
% \end{subequations}
}{
\begin{align}
\label{eq:appendix_equations2}
\begin{aligned}
  & \left( \begin{array}{c}
	      x_1^2 + x_2^2 \\
  2(x_2 y_1 - x_1 y_2)
    \\ y_1^2 + y_2^2
	    \end{array} \right)
  =
  \left( \begin{array}{c}
  f_5 \\
  f_6   \\
  f_7           \end{array} \right)
  \triangleq
  \frac{1}{(a_1 - c_1)\beta_2}
  \left( \begin{array}{c}
  -\beta_2 c_1 \\
  a_2 c_1 - a_1 c_2   \\
  a_1 \beta_2           \end{array} \right) \,.
\end{aligned}
\end{align}
}

\begin{assert}
\label{assert:appendix2}
A solution to \eqref{eq:appendix_equations2} exists if and only if the following conditions holds:
\begin{subequations}
\label{eq:jcase2}
\noeqref{eq:jcase2:1,eq:jcase2:2,eq:jcase2:3}
\begin{align} 
\label{eq:jcase2:1}
f_5 & \geq 0 \\
\label{eq:jcase2:2}
f_7 & \geq 0 \\
\label{eq:jcase2:3}
4 f_5 f_7 - f_6^2 & \geq 0 \,.
\end{align}
\end{subequations}
\end{assert}

\begin{proof}[Proof of \assertref{assert:appendix2}]
\label{sss:justification2}
Construct the following two vectors: $\bp_1=(x_2,-x_1)$, $\bp_2=(y_1,y_2)$. Using the inner product definition, we have
\begin{subequations}
\label{eq:appendix_p_vectors2}
\noeqref{eq:appendix_p_vectors2:1,eq:appendix_p_vectors2:2,eq:appendix_p_vectors2:3}
\begin{align}
\label{eq:appendix_p_vectors2:1}
 ||\bp_1||^2&=x_1^2+x_2^2\\
\label{eq:appendix_p_vectors2:2}
 ||\bp_2||^2&=y_1^2+y_2^2\\
 \label{eq:appendix_p_vectors2:1}
2\left<\bp_1,\bp_2\right> &=  2(x_2y_1-x_1y_2) \,,
\end{align}
\end{subequations}
and the angle between the two vectors satisfies
\begin{align}
  \cos{\theta_1}&=\frac{\left<\bp_1,\bp_2\right>}{||\bp_1||||\bp_2||} \,.
\end{align}
Note that the l.h.s.\ of \eqref{eq:appendix_equations2} coincides with the r.h.s.\ of \eqref{eq:appendix_p_vectors2}.
Thus, a solution to \eqref{eq:appendix_equations2} exists if and only if
\begin{subequations}
\label{eq:pi_case2}
\noeqref{eq:pi_case2:1,eq:pi_case2:2,eq:pi_case2:3}
\begin{align}
\label{eq:pi_case2:1}
 ||\bp_1||^2&\geq0\\*
\label{eq:pi_case2:2}
 ||\bp_2||^2&\geq0\\*
\label{eq:pi_case2:3}
 \frac{\left<\bp_1,\bp_2\right>}{||\bp_1||||\bp_2||}&\leq1
\,,
\end{align}
\end{subequations}

where \eqref{eq:pi_case2:3} is equivalent to 
\begin{align}
 4||\bp_1||^2||\bp_2||^2-\left(2\left<\bp_1,\bp_2\right>\right)^2 \geq 0
  ,
 \end{align}
which is equivalent, in turn, to \eqref{eq:jcase2}.
\end{proof}

By definition, and using \eqref{eq:case_3_equations}, we have $f_5+f_7 = 1$.
Thus, these three equations are equivalent to the single equation
\begin{align}
4 f_5 f_7 - f_6^2 &=
\frac{
 -(a_2 c_1-a_1 c_2)^2-4
 a_1 c_1 \beta_2^2
}{\beta_2^2 (a_1-c_1)^2} \geq 0 \,.
\end{align}
Since the denominator is positive, this is equivalent to
\begin{align}
 -(a_2 c_1-a_1 c_2)^2-4
 a_1 c_1 \beta_2^2
 \geq 0 \,.
\end{align}
On the other hand, we have
\begin{align}
F_1(S_1,S_2) & =
 -(a_2 c_1-a_1 c_2)^2-4
 a_1 c_1 \beta_2^2 \,.
\end{align}
Thus, condition \eqref{eq:condition3} holds if and only if there exists a solution
to \eqref{eq:lemma_decomposition}.

\paragraph{Case 4} \label{tnayal:case4}
We are left with the case where $a_1 \neq c_1$, $b_2=0$, and $\beta_2 = 0$.
In this case,
\eqref{eq:linear_equation_system} becomes
\begin{align}
    \left( \begin{array}{cc}
		1 &  1 \\
		a_1 & c_1 \\
		a_2 & c_2
           \end{array} \right)
    \left( \begin{array}{c}
	    x_1^2 + x_2^2 \\
    y_1^2 + y_2^2
           \end{array} \right)
&=
    \left( \begin{array}{c}
	    1 \\ 0 \\ 0
           \end{array} \right) \,.
\end{align}
A necessary condition for the existence of a solution is that
the second and the third rows are linearly dependent
(or in other words, $a_1 c_2 = a_2 c_1$), in which case we have
\begin{align}
  x_1^2 + x_2^2 &= \frac{c_1}{c_1-a_1}
\\
  y_1^2 + y_2^2 &= \frac{-a_1}{c_1-a_1} \,.
\end{align}
Since we assumed $\det(S_1) \leq 0$, $a_1$ and $c_1$ have opposite signs,
and therefore $x_1^2+x_2^2$ and $y_1^2+y_2^2$ are both non-negative.
In conclusion, a necessary and sufficient condition for the existence of a solution
to \eqref{eq:lemma_decomposition} in this case is $a_2 c_1 = a_1 c_2$. 
On the other hand, we have
\begin{align}
    F_1(S_1,S_2) &=  -(a_2 c_1-a_1 c_2)^2 \,,
\end{align}
which is non-negative if and only if $a_2 c_1 = a_1 c_2$. Thus,
\eqref{eq:condition3} is a necessary and sufficient condition for the existence of a solution
to \eqref{eq:lemma_decomposition}.

This concludes the proof of the lemma.
\hfill$\blacksquare$

%% file: F_invariant.tex
Let $S_1$ and $S_2$ be $n \times n$ complex-valued matrices, and let $U$ be an $n \times n$ unitary matrix.
We have:
\col{
\begin{align*}
	& F_1 \left(  U^\dagger  S_1 U ,  U^\dagger S_2 U  \right)
\\
& =
	\det \left[
U^\dagger S_1 U \adj \left( U^\dagger S_2 U \right)
-
U^\dagger S_2 U \adj \left( U^\dagger S_1 U \right)
\right]
\\
& =
	\det \left[
U^\dagger S_1 U \adj \left( U \right) \adj \left( S_2 \right) \adj \left( U^\dagger \right)
\right.
\\
& \qqquad \left.
-
U^\dagger S_2 U \adj \left( U \right) \adj \left( S_1 \right) \adj \left( U^\dagger \right)
\right] \,.
\end{align*}
Since $U \adj \left( U \right) = \det(U)  I$, we have
\begin{align*}
	& F_1 \left(  U^\dagger  S_1 U ,  U^\dagger S_2 U  \right)
\\
& =
\left[ \det (U) \right]^n
	\det \left[
U^\dagger S_1 \adj \left( S_2 \right) \adj \left( U^\dagger \right)
\right.
\\
& \qqquad \left.
-
U^\dagger S_2 \adj \left( S_1 \right) \adj \left( U^\dagger \right)
\right]
\\
& =
\left[ \det (U) \right]^n
	\det \left[ U^\dagger \left(
S_1 \adj \left( S_2 \right) 
-
S_2 \adj \left( S_1 \right) \right) \adj \left( U^\dagger \right)
\right]
\\
& =
\left[ \det  (U) \right]^n \det \left[  U^\dagger \adj \left( U^\dagger \right)      \right]
	\det \left[
S_1 \adj \left( S_2 \right) 
-
S_2 \adj \left( S_1 \right)
\right]
\\
& =
\left( \det U \right)^n \left[ \det \left( U^\dagger \right) \right]^n
	\det \left[
S_1 \adj \left( S_2 \right)
-
S_2 \adj \left( S_1 \right)
\right]
\\
& =
\left[ \det \left( U  U^\dagger \right) \right]^n
	\det \left[
S_1 \adj \left( S_2 \right)
-
S_2 \adj \left( S_1 \right)
\right]
\\
& =
\det \left( I \right)^n
	\det \left[
S_1 \adj \left( S_2 \right)
-
S_2 \adj \left( S_1 \right)
\right]
\\
& =
	\det \left[
S_1 \adj \left( S_2 \right)
-
S_2 \adj \left( S_1 \right)
\right]
\\
& = F_1(S_1,S_2) \,. \qquad\qquad\qquad\qquad\qquad\qquad\qquad\qquad\qquad \blacksquare
\end{align*}
}
{
\begin{align*}
	F_1 \left(  U^\dagger  S_1 U ,  U^\dagger S_2 U  \right)
& = 	\det \left[
U^\dagger S_1 U \adj \left( U^\dagger S_2 U \right)
-
U^\dagger S_2 U \adj \left( U^\dagger S_1 U \right)
\right]
\\
& =
	\det \left[
U^\dagger S_1 U \adj \left( U \right) \adj \left( S_2 \right) \adj \left( U^\dagger \right)
-
U^\dagger S_2 U \adj \left( U \right) \adj \left( S_1 \right) \adj \left( U^\dagger \right)
\right] \,.
\end{align*}
Since $U \adj \left( U \right) = \det(U)  I$, we have
\begin{align*}
	F_1 \left(  U^\dagger  S_1 U ,  U^\dagger S_2 U  \right)
& =
\left[ \det (U) \right]^n
	\det \left[
U^\dagger S_1 \adj \left( S_2 \right) \adj \left( U^\dagger \right)
-
U^\dagger S_2 \adj \left( S_1 \right) \adj \left( U^\dagger \right)
\right]
\\
& =
\left[ \det (U) \right]^n
	\det \left[ U^\dagger \left(
S_1 \adj \left( S_2 \right)
-
S_2 \adj \left( S_1 \right) \right) \adj \left( U^\dagger \right)
\right]
\\
& =
\left[ \det  (U) \right]^n \det \left[  U^\dagger \adj \left( U^\dagger \right)      \right]
	\det \left[
S_1 \adj \left( S_2 \right)
-
S_2 \adj \left( S_1 \right)
\right]
\\
& =
\left( \det U \right)^n \left[ \det \left( U^\dagger \right) \right]^n
	\det \left[
S_1 \adj \left( S_2 \right)
-
S_2 \adj \left( S_1 \right)
\right]
\\
& =
\left[ \det \left( U  U^\dagger \right) \right]^n
	\det \left[
S_1 \adj \left( S_2 \right)
-
S_2 \adj \left( S_1 \right)
\right]
\\
& =
\det \left( I \right)^n
	\det \left[
S_1 \adj \left( S_2 \right)
-
S_2 \adj \left( S_1 \right)
\right]
\\
& =
	\det \left[
S_1 \adj \left( S_2 \right)
-
S_2 \adj \left( S_1 \right)
\right]
\\
& = F_1(S_1,S_2) \,.
\end{align*}
\hfill$\blacksquare$
}

%% file: rateless_reduction.tex
Recall that the original problem was to perform $2$-GMD \eqref{eq:K-GMD} to the following
two $3 \times 3$ matrices, both having a determinant equal to $1$:
\begin{align}
A_1 & =  \left( \begin{array}{ccc}
	  b^4 & 0 & 0 \\
	  0 &  b^{-2}   & 0 \\
	  0 & 0 & b^{-2}
       \end{array}
\right)
\\
A_2 & =  \left( \begin{array}{ccc}
	  b & 0 & 0 \\
	  0 &  b  & 0 \\
	  0 & 0 & b^{-2}
       \end{array}
\right)\,.
\end{align}
Since these two matrices are diagonal, we can assume, without loss of generality, that the elements in the first column of the matrix $V$ in \eqref{eq:K-GMD} are positive real-valued  (since the phase can be canceled by the matrices $U_k$). Also,
the first columns of $A_1 V$ and of $A_2 V$ must have norms equal to $1$, and thus
\begin{align}
 V = \left(
    \begin{array}{ccc}
      v_{11} & * & * \\
    v_{21} & * & * \\
      v_{31} & * & *
\end{array}
\right) \,,
\end{align}
where
\begin{align} \label{eq:rateless_abc}
\begin{aligned}
      v_{11} &= \frac{1}{\sqrt{b^8+b^4+1}}
\\
      v_{21} &= \frac{b^3}{\sqrt{b^8+b^4+1}}
\\
      v_{31} &= \frac{b^2}{\sqrt{b^4+b^2+1}} \,.
\end{aligned}
\end{align}
The remaining two columns must lay in the orthogonal complement to the subspace spanned by this vector,
which is spanned by the two vectors $(v_{12},v_{22},v_{32})^T$ and $(v_{13},v_{23},v_{33})^T$
where
\begin{align}
      v_{12} &= \frac{b^3}{\sqrt{b^6+1}}
\\
      v_{22} &= \frac{-1}{\sqrt{b^6+1}}
\\
v_{32} &= 0
\\
v_{13} & = \frac{b^2}{ \sqrt{(b^2+1)(b^8+b^4+1))} }
\\
v_{23} & = \frac{b^5}{ \sqrt{(b^2+1)(b^8+b^4+1)} }
\\
v_{33} &= - \frac{ \sqrt{1+b^6}}{ \sqrt{ b^8+b^4+1  } } \,.
\end{align}
In other words, we can represent $V$ as
\begin{align}
 V = V_0
\left(
          \begin{array}{ccc}
      1 & 0 & 0 \\
      0 & W_{11} & W_{12} \\
      0 & W_{21} & W_{22}
\end{array}
\right)
\,,
\end{align}
where
\begin{align}
      V_0 =
\left(
    \begin{array}{ccc}
      v_{11} & v_{12} & v_{13} \\
    v_{21} & v_{22} & v_{23} \\
      v_{31} & v_{32} & v_{33}
\end{array}
\right) \,,
\end{align}
and $W$ is a $2 \times 2$ unitary matrix.
Thus, the matrix
 \begin{align}
\left(
          \begin{array}{ccc}
      1 & 0 & 0 \\
      0 & W_{11} & W_{12} \\
      0 & W_{21} & W_{22}
\end{array}
\right)
\,,
\end{align}
performs $2$-GMD on the two matrices
\begin{align}
A_1 V_0 & =  \left( \begin{array}{ccc}
	  b^4 v_{11} & b^4 v_{12} & b^4 v_{13} \\
	  b^{-2} v_{21} & b^{-2} v_{22} & b^{-2} v_{23} \\
	  b^{-2} v_{31} & b^{-2} v_{32} & b^{-2} v_{33}
       \end{array}
\right)
\\
A_2 V_0 & =  \left( \begin{array}{ccc}
	  b v_{11} & b v_{12} & b v_{13} \\
	  b v_{21} & b v_{22}  & b v_{23} \\
	  b^{-2} v_{31} & b^{-2}v_{32} & b^{-2} v_{33}
       \end{array}
\right)\,.
\end{align}
or, equivalently, on the same matrices after Gram-Schmidt orthogonalization (i.e., QR decomposition):
\begin{align}
U_1^\dagger A_1 V_0 &= \left( \begin{array}{ccc}	
	  1 & * & * \\
	  0 & \frac{\sqrt{1-b^2+b^8}}{b^2} & \frac{b^6-1}{b\sqrt{(1-b^2+b^8)(1+b^2+b^4)}} \\
	  0 & 0 & \frac{b^2}{\sqrt{1-b^2+b^8}}
       \end{array}
\right)
\,, \\
U_2^\dagger A_2 V_0 &=
\left( \begin{array}{ccc}
	  1 & * & * \\
	  0 & b & 0 \\
	  0 & 0 & b^{-1}
       \end{array}
\right)\,.
\end{align}
In other words, $W$ performs $2$-GMD on the two following matrices:
\begin{align}
\tilde{A}_1 &= \left( \begin{array}{cc}	
	  \frac{\sqrt{1-b^2+b^8}}{b^2} & \frac{b^6-1}{b\sqrt{(1-b^2+b^8)(1+b^2+b^4)}} \\
	  0 & \frac{b^2}{\sqrt{1-b^2+b^8}}	
       \end{array}
\right)
\\
\tilde{A}_2 &=
\left( \begin{array}{cc}
	   b  & 0 \\
	  0 & b^{-1}
       \end{array}
\right)\,,
\end{align}
which is what we wanted to prove.
\hfill$\blacksquare$

%% file: extention_does_not_help.tex
Let $A_1$ and $A_2$ be two complex-valued $2 \times 2$ matrices with determinants equal to $1$.
Define:
\begin{align}
 S_1 & \triangleq A_1^\dagger A_1 - I \\
 S_2 & \triangleq A_2^\dagger A_2 - I \,.
\end{align}
Let $N \geq 2$, and define the following extended matrices:
\begin{align}
  \begin{array}{c}
    \AAA_k  \triangleq \blkmat{A_k}{N} \\
    \SSS_k  \triangleq \blkmat{S_k}{N}
  \end{array}
\quad k=1,2 \,.
\end{align}
Now, assume that there exist complex-valued unitary matrices $\UUU_1,\UUU_2,\VVV$ such that
\begin{align} \label{eq:ext_does_not_help_decomposition}
      \UUU_k^\dagger \AAA_k \VVV = \TTT_k \,, \quad k=1,2 \,,
\end{align}
where $\TTT_k$ are upper triangular with all the diagonal values equal $1$.
In particular, if we denote the first column of $\VVV$ by $\vvv$, then necessary (although not sufficient) conditions for the existence of the decomposition \eqref{eq:ext_does_not_help_decomposition}
are
\begin{align}
	\left\| \AAA_1 \vvv \right\|^2 &= 1
\\
	\left\| \AAA_2 \vvv \right\|^2 &= 1
\\
	\left\| \vvv \right\|^2 &= 1 \,,
\end{align}
or equivalently,
\begin{subequations}
\label{eq:space_time_cond}
\noeqref{eq:space_time_cond:1,eq:space_time_cond:2,eq:space_time_cond:3}
\begin{align} 
\label{eq:space_time_cond:1}
	\vvv^\dagger \SSS_1 \vvv &= 0
\\
\label{eq:space_time_cond:2}
	\vvv^\dagger \SSS_2 \vvv &= 0
\\
\label{eq:space_time_cond:3}
	\vvv^\dagger \vvv &= 1 \,.
\end{align}
\end{subequations}
As in the proof of \lemref{lem:det_adj_complex}, we can assume, without loss of generality, that
$S_1$ is real-valued and diagonal.
Denoting
\begin{align}
	\vvv &= \left(
    \begin{array}{c}
	x_1 + i x_2 \\ y_1+ i y_2 \\ \vdots \\ x_{2N-1} + ix_{2N} \\ y_{2N-1} + iy_{2N}
    \end{array}
\right) \,,
\\
 S_1 &= \left( \begin{array}{cc}
		a_1 & 0 \\ 0 & c_1
             \end{array} \right)
\\
 S_2 &= \left( \begin{array}{cc}
		a_2 & b_2+i\beta_2 \\ b_2-i\beta_2 & c_2
             \end{array} \right) \,,
\end{align}
the three equations \eqref{eq:space_time_cond} become
\col{
\begin{align}
\!\!\!\!
\!\!\!\!
    \left[ \begin{array}{cccc}
		1 & 0 & 0 & 1 \\
		a_1 & 0 & 0 & c_1 \\
		a_2 & b_2 & \beta_2 & c_2
           \end{array} \right]
    \left[ \begin{array}{c}
	     X_1+\cdots+X_N\\
         2(W_1+\cdots+W_N)\\
         2(Z_1+\cdots+Z_N)\\
         Y_1+\cdots+Y_N\\
           \end{array} \right]
&=
    \left[ \begin{array}{c}
	    1 \\ 0 \\ 0
           \end{array} \right]
\label{eq:extentions_linear_equation_system}
\end{align}
}
{
\begin{align}
    \left( \begin{array}{cccc}
		1 & 0 & 0 & 1 \\
		a_1 & 0 & 0 & c_1 \\
		a_2 & b_2 & \beta_2 & c_2
           \end{array} \right)
    \left( \begin{array}{c}
	     X_1+\cdots+X_N\\
         2(W_1+\cdots+W_N)\\
         2(Z_1+\cdots+Z_N)\\
         Y_1+\cdots+Y_N\\
           \end{array} \right)
&=
    \left( \begin{array}{c}
	    1 \\ 0 \\ 0
           \end{array} \right),
\label{eq:extentions_linear_equation_system}
\end{align}
}
where we define
\begin{align}
    X_j &\triangleq x_{2j-1}^2+x_{2j}^2 \\*
    W_j &\triangleq x_{2j-1}y_{2j-1} + x_{2j}y_{2j} \\*
    Z_j &\triangleq x_{2j}y_{2j-1} - x_{2j-1}y_{2j} \\*
    Y_j &\triangleq y_{2j-1}^2+y_{2j}^2 \,.
\end{align}
We now consider the following cases.

\paragraph{Case 1}
Assume first that $a_1 \neq c_1$ and $b_2 \neq 0$.
Thus, \eqref{eq:extentions_linear_equation_system} is equivalent to:
\col{
\begin{align*}
\left[ \begin{array}{c}
	     X_1+\cdots+X_N\\
         2(W_1+\cdots+W_N)\\
         2(Z_1+\cdots+Z_N)\\
         Y_1+\cdots+Y_N\\
           \end{array} \right]
&=
    \overbrace{ \left[ \begin{array}{cccc}
		1 & 0 & 0 & 1 \\
		a_1 & 0 & 0 & c_1 \\
		a_2 & b_2 & \beta_2 & c_2 \\
		  0 & 0 & 1 & 0
           \end{array} \right]^{-1}}^{B^{-1}}
    \left[ \begin{array}{c}
	    1 \\ 0 \\ 0 \\ t
           \end{array} \right]
\\
& \triangleq
    \frac{1}{\Delta}
    \left[ \begin{array}{c}
	      f_1(t) \\ f_2(t) \\ f_3(t) \\ f_4(t)
           \end{array} \right] \,, \label{eq:does_not_help_f1f2f3f4}
\end{align*}
}
{
\begin{align}
\left( \begin{array}{c}
	     X_1+\cdots+X_N\\
         2(W_1+\cdots+W_N)\\
         2(Z_1+\cdots+Z_N)\\
         Y_1+\cdots+Y_N\\
           \end{array} \right)
&=
    \overbrace{ \left( \begin{array}{cccc}
		1 & 0 & 0 & 1 \\
		a_1 & 0 & 0 & c_1 \\
		a_2 & b_2 & \beta_2 & c_2 \\
		  0 & 0 & 1 & 0
           \end{array} \right)^{-1}}^{B^{-1}}
    \left( \begin{array}{c}
	    1 \\ 0 \\ 0 \\ t
           \end{array} \right)
\triangleq
    \frac{1}{\Delta}
    \left( \begin{array}{c}
	      f_1(t) \\ f_2(t) \\ f_3(t) \\ f_4(t)
           \end{array} \right) \,, \label{eq:does_not_help_f1f2f3f4}
\end{align}
}
where $t$ is some real-valued parameter,
$f_1(t)$,$f_2(t)$,$f_3(t)$,$f_4(t)$ are first-degree polynomials in $t$
(with coefficients that depend on the matrices $S_1,S_2$),
and
\begin{align}
	\Delta \triangleq \det B = b_2(c_1 - a_1) \neq 0\,.
\end{align}
Thus, finding a solution $\vvv$ to the original problem is \emph{equivalent} to finding a solution
$(x_1,\cdots,x_{2N},y_1,\cdots,y_{2N},t)$ to the following equations:
\begin{subequations}
\label{eq:appendix_equations1_spacetime}
\noeqref{eq:appendix_equations1_spacetime:1,eq:appendix_equations1_spacetime:2,eq:appendix_equations1_spacetime:3,eq:appendix_equations1_spacetime:4}
\begin{align}
\label{eq:appendix_equations1_spacetime:1}
X_1+\cdots+X_N   & = \frac{1}{\Delta} f_1(t) \\
\label{eq:appendix_equations1_spacetime:2}
 2(W_1+\cdots+W_N)    & = \frac{1}{\Delta} f_2(t) \\
\label{eq:appendix_equations1_spacetime:3}
 2(Z_1+\cdots+Z_N)    & = \frac{1}{\Delta} f_3(t) \\
\label{eq:appendix_equations1_spacetime:4}
 Y_1+\cdots+Y_N    & = \frac{1}{\Delta} f_4(t) \,.
\end{align}
\end{subequations}

\begin{assert}
\label{assert:appendix1_spacetime}
  A solution to \eqref{eq:appendix_equations1_spacetime} exists
  if and only if the following conditions hold for some $t \in \mathbb{R}$:
  \begin{subequations}
  \label{eq:jcase1_spacetime}
  \noeqref{eq:jcase1_spacetime:1,eq:jcase1_spacetime:2,eq:jcase1_spacetime:3}
  \begin{align}
\label{eq:jcase1_spacetime:1}
     \frac{1}{\Delta} f_1(t) & \geq 0 \\
\label{eq:jcase1_spacetime:2}
    \frac{1}{\Delta} f_4(t) & \geq 0 \\
\label{eq:jcase1_spacetime:3}
    4 f_1(t) f_4(t) & \geq f_2^2(t) + f_3^2(t)  \,.
  \end{align}
  \end{subequations}
\end{assert}

\begin{proof}[Proof of \assertref{assert:appendix1_spacetime}]
\label{sss:justification1_spacetime}
  Construct the following three vectors:
\begin{align}
 \bp_1&=(x_2,-x_1,x_4,-x_3,\cdots,x_{2N},x_{2N-1}) \\*
 \bp_2&=(y_1,y_2,y_3,y_4,\cdots,y_{2N-1},y_{2N})	   \\*
 \bp_3&=(x_1,x_2,x_3,x_4,\cdots,x_{2N-1},x_{2N})
  \,.
\end{align}
  Using the inner product definition, we have
  \begin{align}
  \label{eq:appendix_p_vectors1_spacetime}
  \begin{aligned}
||\bp_1||^2&=||\bp_3||^2 
\\ &= x_1^2+x_2^2+x_3^2+x_4^2+\cdots+x_{2N-1}^2+x_{2N}^2\\*
&= X_1+X_2\cdots+X_N\\
||\bp_2||^2&=y_1^2+y_2^2+y_3^2+y_4^2+\cdots+y_{2N-1}^2+y_{2N}^2\\*
&=Y_1+Y_2+\cdots+Y_N\\
2\left<\bp_1,\bp_2\right> &= 2(x_2y_1-x_1y_2+\cdots+x_{2N}y_{2N-1}-x_{2N-1}y_{2N})\\*
&=2(Z_1+Z_2+\cdots+ Z_N)\\
2\left<\bp_2,\bp_3\right> &= 2(x_1y_1+x_2y_2+\cdots+x_{2N-1}y_{2N-1}+x_{2N}y_{2N})\\*
&=2(W_1+W_2+\cdots+W_N)
\,,
  \end{aligned}
  \end{align}
    and the angles between these vectors satisfy
\begin{align}
 \cos{\theta_1}&=\frac{\left<\bp_1,\bp_2\right>}{||\bp_1||||\bp_2||}\\*
 \cos{\theta_2}&=\frac{\left<\bp_3,\bp_2\right>}{||\bp_3||||\bp_2||}=\frac{\left<\bp_3,\bp_2\right>}{||\bp_1||||\bp_2||}
\,.
\end{align}

Note that the l.h.s.\ of \eqref{eq:appendix_equations1_spacetime} and the r.h.s.\ of \eqref{eq:appendix_p_vectors1_spacetime} coincide, and that $\left<\bp_3,\bp_1\right>=0$. Therefore the angle between them is $\pi/2$. One verifies that the maximum of $\cos^2{\theta_1}+\cos^2{\theta_2}$ is achieved when all three vectors are on the same plane, in which case $\cos{\theta_2}=\cos{(\pm\pi/2-\theta_1)}=\pm \sin{\theta_1}$, which implies that $\cos^2{\theta_1}+\cos^2{\theta_2}=1$.
When the three vectors do not lay on the same plane, $\cos^2{\theta_1}+\cos^2{\theta_2} < 1$.

  Thus, a solution to \eqref{eq:appendix_equations1_spacetime} exists if and only if
\begin{subequations}
\label{eq:ext_no_help:p_i}
\noeqref{eq:ext_no_help:p_i:1,eq:ext_no_help:p_i:2,eq:ext_no_help:p_i:3}
 \begin{align}
\label{eq:ext_no_help:p_i:1}
 ||\bp_1||^2&\geq0\\*
\label{eq:ext_no_help:p_i:2}
 ||\bp_2||^2&\geq0\\*
\label{eq:ext_no_help:p_i:3}
 \!\!\!\!\!\!\!\!\!\!\!\!\cos^2\theta_1+\cos^2\theta_2&=\frac{\left<\bp_1,\bp_2\right>^2}{||\bp_1||^2||\bp_2||^2} + \frac{\left<\bp_3,\bp_2\right>^2}{||\bp_1||^2||\bp_2||^2}\leq1
,
\end{align}
\end{subequations}
where \eqref{eq:ext_no_help:p_i:3} is equivalent to 
\begin{align}
 4||\bp_1||^2||\bp_2||^2-\left(2\left<\bp_1,\bp_2\right>\right)^2-\left(2\left<\bp_3,\bp_2\right>\right)^2 \geq 0 
\,,
\end{align}
  which is equivalent, in turn, to \eqref{eq:jcase1_spacetime}.
\end{proof}

By definition, and using \eqref{eq:does_not_help_f1f2f3f4}, $(f_1(t)+f_4(t)) = \Delta$. Therefore, these three conditions are equivalent to the following single condition:
\begin{align}
      4f_1(t) f_4(t) - f_2^2(t) - f_3^2(t) \geq 0 \,.
\end{align}
This is a quadratic inequality in $t$,
\begin{align} \label{eq:parabolic_inequality}
      a t^2 + bt + c \geq 0 \,,
\end{align}
where the constants $a,b,c$ are as in \eqref{eq:parabolic_inequality_constants}.
Note that since
 $a_1 \neq c_1$ and $b_2 \neq 0$, the coefficient $a$ is \emph{strictly} negative.
Therefore, a necessary and sufficient condition for the existence of a (real-valued) solution $t$ to the inequality in \eqref{eq:parabolic_inequality} is for the discriminant to be non-negative:
\begin{align}
      b^2 - 4 a c \geq 0 \,.
\end{align}
A direct calculation shows that
\begin{align}
      b^2 - 4 a c = 4 \Delta^2 F_1(A_1^\dagger A_1 - I ,  A_2^\dagger A_2 - I)  \,,
\end{align}
where $F_1$ is defined as in \eqref{eq:definition_of_F}.
This condition is the same as the condition in \eqref{eq:condition3} which completes the proof of \thrmref{thm:extention_does_not_help} for this case.

\paragraph{Case 2}
Assume now that $a_1=c_1$. As in case 2 in the proof of \lemref{lem:det_adj_complex},
condition \eqref{eq:condition3} holds, and thus this case is not possible under the assumptions of the theorem.

\paragraph{Case 3}
Next, assume that $a_1 \neq c_1$, $b_2=0$, and $\beta_2 \neq 0$.
Thus, \eqref{eq:extentions_linear_equation_system} becomes
\begin{align}
    \left( \begin{array}{ccc}
		1 &  0 & 1 \\
		a_1 &  0 & c_1 \\
		a_2 &  \beta_2 & c_2
           \end{array} \right)
\left( \begin{array}{c}
	     X_1+\cdots+X_N\\
         2(Z_1+\cdots+Z_N)\\
         Y_1+\cdots+Y_N\\
           \end{array} \right)
&=
    \left( \begin{array}{c}
	    1 \\ 0 \\ 0
           \end{array} \right) \,,
\end{align}
which reduces to 
\col{
\begin{subequations}
\label{eq:appendix_equations2_spacetime}
\noeqref{eq:appendix_equations2_spacetime:1,eq:appendix_equations2_spacetime:2}
\begin{align}
\label{eq:appendix_equations2_spacetime:1}
&\left( \begin{array}{c}
	     X_1+\cdots+X_N\\
         2(Z_1+\cdots+Z_N)\\
         Y_1+\cdots+Y_N\\
           \end{array} \right)
\\*
\label{eq:appendix_equations2_spacetime:2}
&=
 \left( \begin{array}{c}
f_5 \\
f_6   \\
f_7           \end{array} \right)
\triangleq
\frac{1}{(a_1 - c_1)\beta_2}
 \left( \begin{array}{c}
-\beta_2 c_1 \\
a_2 c_1 - a_1 c_2   \\
a_1 \beta_2           \end{array} \right) \,.
\end{align}
\end{subequations}
}{
\begin{align}
 \label{eq:appendix_equations2_spacetime}
\begin{aligned}
&\left( \begin{array}{c}
	     X_1+\cdots+X_N\\
         2(Z_1+\cdots+Z_N)\\
         Y_1+\cdots+Y_N\\
           \end{array} \right)
=
 \left( \begin{array}{c}
f_5 \\
f_6   \\
f_7           \end{array} \right)
\triangleq
\frac{1}{(a_1 - c_1)\beta_2}
 \left( \begin{array}{c}
-\beta_2 c_1 \\
a_2 c_1 - a_1 c_2   \\
a_1 \beta_2           \end{array} \right) \,.
\end{aligned}
\end{align}
}

\begin{assert}
\label{assert:appendix2_spacetime}
A solution to \eqref{eq:appendix_equations2_spacetime} exists if and only if the following conditions holds:
\begin{subequations}
\label{eq:jcase2_spacetime}
\noeqref{eq:jcase2_spacetime:1,eq:jcase2_spacetime:2,eq:jcase2_spacetime:3}
\begin{align} 
\label{eq:jcase2_spacetime:1}
f_5 & \geq 0 \\
\label{eq:jcase2_spacetime:2}
f_7 & \geq 0 \\
\label{eq:jcase2_spacetime:3}
4 f_5 f_7 - f_6^2 & \geq 0 \,.
\end{align}
\end{subequations}
\end{assert}

\begin{proof}[Proof of \assertref{assert:appendix2_spacetime}]
\label{sss:justification2_spacetime}
Construct the following two vectors:
\begin{align}
 \bp_1&=(x_2,-x_1,x_4,-x_3,\cdots,x_{2N},x_{2N-1}) \\*
 \bp_2&=(y_1,y_2,y_3,y_4,\cdots,y_{2N-1},y_{2N})	
  \,.
\end{align}
 Using the inner product definition, we have
\begin{align}
\begin{aligned}
\label{eq:appendix_p_vectors2_spacetime}
||\bp_1||^2 &= x_1^2+x_2^2+x_3^2+x_4^2+\cdots+x_{2N-1}^2+x_{2N}^2\\*
&= X_1+X_2\cdots+X_N\\*
||\bp_2||^2&=y_1^2+y_2^2+y_3^2+y_4^2+\cdots+y_{2N-1}^2+y_{2N}^2\\*
&=Y_1+Y_2+\cdots+Y_N\\*
2\left<\bp_1,\bp_2\right> &= 2(x_2y_1-x_1y_2+\cdots+x_{2N}y_{2N-1}-x_{2N-1}y_{2N})\\*
&=2(Z_1+Z_2+\cdots+ Z_N)
\,.
\end{aligned}
\end{align}
and the angle between the two vectors satisfies
\begin{align}
  \cos{\theta_1}&=\frac{\left<\bp_1,\bp_2\right>}{||\bp_1||||\bp_2||} \,.
\end{align}
Note that the l.h.s.\ of \eqref{eq:appendix_equations2_spacetime} and the r.h.s.\ of \eqref{eq:appendix_p_vectors2_spacetime} coincide.
Thus, a solution to \eqref{eq:appendix_equations2_spacetime} exists if and only if
\begin{subequations}
\label{eq:ext_no_hhelp:pi_case_3}
\noeqref{eq:ext_no_hhelp:pi_case_3:1,eq:ext_no_hhelp:pi_case_3:2,eq:ext_no_hhelp:pi_case_3:3}
\begin{align}
\label{eq:ext_no_hhelp:pi_case_3:1}
 ||\bp_1||^2&\geq0\\*
\label{eq:ext_no_hhelp:pi_case_3:2}
 ||\bp_2||^2&\geq0\\*
\label{eq:ext_no_hhelp:pi_case_3:3}
 \frac{\left<\bp_1,\bp_2\right>}{||\bp_1||||\bp_2||}&\leq1
\end{align}
\end{subequations}
where \eqref{eq:ext_no_hhelp:pi_case_3:3} is equivalent to 
\begin{align}
 4||\bp_1||^2||\bp_2||^2-\left(2\left<\bp_1,\bp_2\right>\right)^2 \geq 0 \,,
 \end{align}
which is equivalent, in turn, to \eqref{eq:jcase2_spacetime}.
\end{proof}

From \assertref{assert:appendix2_spacetime}, a necessary condition for the existence of a solution to  \eqref{eq:appendix_equations2_spacetime} is
 \begin{align}
4 f_5 f_7 - f_6^2 & \geq 0 \,,
\end{align}
which is equivalent, in turn, to
\begin{align}
 -(a_2 c_1-a_1 c_2)^2-4
 a_1 c_1 \beta_2^2
 \geq 0 \,.
\end{align}
On the other hand,
\begin{align}
F_1(S_1,S_2) & =
 -(a_2 c_1-a_1 c_2)^2-4
 a_1 c_1 \beta_2^2 \,.
\end{align}
Thus condition \eqref{eq:condition3} must hold true, since otherwise
no solution to
\eqref{eq:space_time_cond} exists.

\paragraph{Case 4}
We are left with the case where $a_1 \neq c_1$, $b_2=0$, and $\beta_2 = 0$.
In this case,
\eqref{eq:extentions_linear_equation_system} reduces to
\begin{align}
    \left( \begin{array}{cc}
		1 &  1 \\
		a_1 & c_1 \\
		a_2 & c_2
           \end{array} \right)
\left( \begin{array}{c}
	     X_1+\cdots+X_N\\
         Y_1+\cdots+Y_N\\
           \end{array} \right)
&=
    \left( \begin{array}{c}
	    1 \\ 0 \\ 0
           \end{array} \right) \,.
\end{align}
A necessary condition for the existence of a solution in this case, is that
the second and the third rows are linearly dependent, i.e., $a_1 c_2 = a_2 c_1$.
On the other hand,
\begin{align}
    F_1(S_1,S_2) &=  -(a_2 c_1-a_1 c_2)^2 \,.
\end{align}
Thus if condition \eqref{eq:condition3} does not hold, no solution to
\eqref{eq:space_time_cond} exists.

This concludes the proof of the theorem.
\hfill$\blacksquare$

%% file: proof_gmd_jet.tex
First, assume that statement 2 holds.
Namely, There exist $K+2$ matrices with orthonormal columns
$U_1,\ldots,U_{K+1},V$, of dimensions $n \times \nnn$, such that
\begin{align}
	U_k^\dagger A_k V = R_k \,,\qquad k=1,\ldots,K+1 \,,
\end{align}
where $\left\{ R_k \right\}$ are $\nnn \times \nnn$ upper triangular with equal diagonals.
Now, arbitrarily extend $V$ to an $n \times n$ unitary matrix:
\begin{align}
      \tilde{V} = \left( \begin{array}{c|c}
 V & V^\perp \end{array} \right) \,.
\end{align}
Then, $U_k$  can also be extended  to $n \times n$ unitary matrices, by performing
Gram-Schmidt process on the columns of $A_k \tilde{V}$:
\begin{align}
      \tilde{U}_k = \left( \begin{array}{c|c} U_k & U_k^\perp \end{array} \right) \,,
\end{align}
such that
\begin{align}
      \tilde{U}_k^\dagger A_k \tilde{V} = \tilde{R}_k = 
      \left(
        \begin{array}{c|c}
                R_k & * \\ \hline
                0   & \hat{R}_k
        \end{array}
      \right)
    \,,
\end{align}
and $\hat{R}_k$ are upper triangular (with diagonal elements that depend on $k$).
Thus, we have:
\col{
\begin{align}
	\tilde{U}_k^\dagger B_k \tilde{U}_{K+1} &=
      \tilde{U}_k^\dagger A_k A_{K+1}^{-1} \tilde{U}_{K+1}
\\
      &=
      \tilde{U}_k^\dagger A_k \tilde{V} \tilde{V}^\dagger A_{K+1}^{-1} \tilde{U}_{K+1}
\\
      &=
      \tilde{R}_k  \tilde{R}_{K+1}^{-1}
\\
      &=
	\tilde{T}_k \,,
\end{align}
}
{
\begin{align}
	\tilde{U}_k^\dagger B_k \tilde{U}_{K+1} &=
      \tilde{U}_k^\dagger A_k A_{K+1}^{-1} \tilde{U}_{K+1}
\\*
      &=
      \tilde{U}_k^\dagger A_k \tilde{V} \tilde{V}^\dagger A_{K+1}^{-1} \tilde{U}_{K+1}
\\*
      &=
      \tilde{R}_k  \tilde{R}_{K+1}^{-1}
\\*
      &=
	\tilde{T}_k \,,
\end{align}
}
where $\tilde{T}_k$ is of the form
\begin{align}
      \tilde{T}_k = \left(
      \begin{array}{c|c}
	    T_k & * \\ \hline
	    0   & \hat{T}_k
      \end{array}
\right) \,,
\end{align}
where $T_k$ is upper triangular with all the diagonal elements equal to $1$,
and $\hat{T}_k$ is upper triangular (with diagonal elements that depend on $k$).
By substitution:
\begin{align}
      \left( \begin{array}{c}
		      U_k^\dagger \\ \hline \left(U_k^\perp\right)^\dagger
             \end{array} \right)
      B_k
      \left( \begin{array}{c|c}
		      U_{K+1} & U_{K+1}^\perp
             \end{array} \right)
=\left(
      \begin{array}{c|c}
	    T_k & * \\ \hline
	    0   & \hat{T}_k
      \end{array}
\right) \,.
\end{align}
By taking only the first $\nnn$ rows and the first $\nnn$ columns of this equality, we obtain
\begin{align}
    U_k^\dagger B_k U_{K+1} = T_k \,,
\end{align}
which results in statement 1.

Now, assume that statement 1 holds.
Perform the QR decomposition on the matrix $A_{K+1}^{-1} U_{K+1}$:
\begin{align}
	A_{K+1}^{-1} U_{K+1} = V R  \,,
\end{align}
where $V$ is of dimensions $n \times \nnn$ with orthonormal columns, and $R$ is an \mbox{$\nnn \times \nnn$}  upper triangular matrix.
Thus, using \eqref{eq:definition_of_a}, we obtain the following equalities:
\begin{align}
  U_k^\dagger A_k V R & = U_k^\dagger A_k A_{K+1}^{-1} U_{K+1}
\\*
 & = U_k^\dagger B_k U_{K+1}  \,, \qquad k=1,\ldots,K \,,
\end{align}
which, according to \eqref{eq:jet_of_a}, suggest
\begin{align} \label{eq:miffy1}
&   U_k^\dagger A_k V R = T_k  \,, \qquad \qquad k=1,\ldots,K \,.
\end{align}
On the other hand, we have
\col{
\begin{subequations}
\label{eq:miffy2}
\noeqref{eq:miffy2:1,eq:miffy2:2}
\begin{align}
 U_{K+1}^\dagger A_{K+1} V R    &= U_{K+1}^\dagger A_{K+1} A_{K+1}^{-1} U_{K+1} 
\label{eq:miffy2:1}
\\*
 &= U_{K+1}^\dagger U_{K+1}  =  I \,. 
\label{eq:miffy2:2}
\end{align}
\end{subequations}
}
{
\begin{align}
 U_{K+1}^\dagger A_{K+1} V R    &= U_{K+1}^\dagger A_{K+1} A_{K+1}^{-1} U_{K+1} \\*
 &= U_{K+1}^\dagger U_{K+1}   \\* 
 &=  I \,. 
\label{eq:miffy2}
\end{align}
}
Multiplying \eqref{eq:miffy1} and \eqref{eq:miffy2} by $R^{-1}$ on the right yields:
\begin{align}
  U_k^\dagger A_k V         &= T_k R^{-1}  \,, \qquad k=1,\ldots,K \\
  U_{K+1}^\dagger A_{K+1} V &= R^{-1} \,.
\end{align}
Since $T_k$ are upper triangular with only $1$s on the diagonal, the matrices
$ R_k  \triangleq T_k R^{-1} $  ($k=1,\ldots,K$) and
$ R_{K+1} \triangleq R^{-1} $
have equal diagonals, thus statement 2 holds.

This completes the proof. \hfill$\blacksquare$

%% file: proof_K_3_n_2.tex
  The proof will be based on $K=3$ steps.

  Denote by $\left\{ \mA_k \right\}$ the extended matrices corresponding to $N = 4$ channel uses.

  \underline{Step 1:}

  Start by applying a 1-GMD for each block (corresponding to a single channel use) of the first matrix $A_1$:
\begin{align}
	\Ul11 A_1 \Vl1 = \left( \begin{array}{cc}
	                     1 & * \\ 0 & 1
	                 \end{array} \right)
	\,,
\end{align}
\noindent which corresponds, in turn, to applying the following extended unitary matrices
(recall the definition of the embedding operation)
\begin{align}
    \UUUU11 &\triangleq \embb{8}{\Ul11}{\subm{1}{2}\subm{3}{4}\subm{5}{6}\subm{7}{8}}\,,\\
    \VVVV1  &\triangleq \embb{8}{\Vl1} {\subm{1}{2}\subm{3}{4}\subm{5}{6}\subm{7}{8}}
    \,,
\end{align}
and results in the following extended triangular matrix
  \begin{align}
      &\TTT_1^{(1)} = \UUUU11\AAA_1 \VVVV1\\
      &\quad = \left(
	\begin{array}{ccccccccc} \cline{1-2}
		\multicolumn{1}{|c}{1} & \multicolumn{1}{c|}{*} & 0 & 0 & 0 & 0 & 0 & 0 \\
		\multicolumn{1}{|c}{0}    &  \multicolumn{1}{c|}{\cellcolor[gray]{0.8} 1} & 0 & 0 & \multicolumn{1}{c}{\cellcolor[gray]{0.8} 0} & 0 & 0 & 0 \\ \cline{1-4}
		  0 & 0 & \multicolumn{1}{|c}{1} & \multicolumn{1}{c|}{*}  & 0 & 0 & 0 & 0  \\
		  0 & 0 &  \multicolumn{1}{|c}{0}    & \multicolumn{1}{c|}{1} & 0 & 0 & 0 & 0 \\ \cline{3-6}
		  0 & \multicolumn{1}{c}{\cellcolor[gray]{0.8} 0} & 0 & 0 & \multicolumn{1}{|c}{\cellcolor[gray]{0.8} 1} & \multicolumn{1}{c|}{*} & 0 & 0 \\
		  0 & 0 & 0 & 0 &  \multicolumn{1}{|c}{0}    & \multicolumn{1}{c|}{1} & 0 & 0 \\ \cline{5-8}
		  0 & 0 & 0 & 0 & 0 & 0 & \multicolumn{1}{|c}{1} & \multicolumn{1}{c|}{*} \\
		  0 & 0 & 0 & 0 & 0 & 0 & \multicolumn{1}{|c}{0} & \multicolumn{1}{c|}{1} \\ \cline{7-8}
	\end{array}
      \right)\,.
  \end{align}

  Note that the same matrix $\VVVV1$ has to be applied to all matrices (since the encoder is shared by all users).
  We decompose the resulting matrices (after multiplying them by $\VVVV1$) according to the QR decomposition,
  resulting in unitary matrices $\UUUU{k}{1}$ such that:
  \begin{align}
      &\TTT_k^{(1)} = \UUUU{k}1 \AAA_k \VVVV1  \\
      &\quad = \left(
	\begin{array}{ccccccccc} \cline{1-2}
		\multicolumn{1}{|c}{r_1^k} & \multicolumn{1}{c|}{*} & 0 & 0 & 0 & 0 & 0 & 0 \\
		\multicolumn{1}{|c}{0}    &  \multicolumn{1}{c|}{\cellcolor[gray]{0.8}r_2^k} & 0 & 0 & \multicolumn{1}{c}{\cellcolor[gray]{0.8} 0} & 0 & 0 & 0 \\ \cline{1-4}
		  0 & 0 & \multicolumn{1}{|c}{r_1^k} & \multicolumn{1}{c|}{*}  & 0 & 0 & 0 & 0  \\
		  0 & 0 &  \multicolumn{1}{|c}{0}    & \multicolumn{1}{c|}{r_2^k} & 0 & 0 & 0 & 0 \\ \cline{3-6}
		  0 & \multicolumn{1}{c}{\cellcolor[gray]{0.8}0} & 0 & 0 & \multicolumn{1}{|c}{\cellcolor[gray]{0.8}r_1^k} & \multicolumn{1}{c|}{*} & 0 & 0 \\
		  0 & 0 & 0 & 0 &  \multicolumn{1}{|c}{0}    & \multicolumn{1}{c|}{r_2^k} & 0 & 0 \\ \cline{5-8}
		  0 & 0 & 0 & 0 & 0 & 0 & \multicolumn{1}{|c}{r_1^k} & \multicolumn{1}{c|}{*} \\
		  0 & 0 & 0 & 0 & 0 & 0 & \multicolumn{1}{|c}{0} & \multicolumn{1}{c|}{r_2^k} \\ \cline{7-8}
	\end{array}
      \right) \,,
  \end{align}
where $r_1^k \, r_2^k = 1$ and $k = 2,3$.

  \underline{Step 2:} \\
  In the second step we apply the 1-GMD to the matrices $\TTT_2^{(1)}\subm{2}{5}$ and $\TTT_2^{(1)}\subm{4}{7}$.
  In both cases the two-by-two matrices are of the same form:
  \begin{align}
    \Ul22\left(
    \begin{array}{cc}
      r_2^2 & 0 \\
      0	    & r_1^2 \\
    \end{array}
    \right)
    \Vl2
    =
    \left(
    \begin{array}{cc}
      1 & * \\
      0	    & 1 \\
    \end{array}
    \right)
   \,.
  \end{align}

  Now note that the matrix corresponding to these elements in $\TTT_1^{(1)}$ have the identity matrix form $I_2$.
Thus, by \propertyref{prop:lenI}, applying $\Vl2$ on the right and $\left( \Vl2 \right)^\dagger$ on the left results in the identity matrix, \ie, 
$\TTT_1^{(1)}\subm{2}{5}$ and $\TTT_1^{(1)}\subm{4}{7}$ remain unchanged.

  For the third matrix, we apply the QR decomposition with $\UUUU32$ (assuming no special structure).\\
  Define
 \begin{align}
    \UUUU22 &\triangleq \embb{8}{\Ul22}{\subm{2}{5}\subm{4}{7}} \,, \\
    \VVVV2  &\triangleq \embb{8}{\Vl2} {\subm{2}{5}\subm{4}{7}}
    \,.
  \end{align}

 Thus, we attain the following matrices after the completion of the second step:
\begin{align*}
      &\TTT_2^{(2)} = \UUUU22 \TTT_2^{(1)} \VVVV2  \\*
      &\quad = \left(
	\begin{array}{ccccccccc}
		  r_1^2 & * & 0 & 0 & * & 0 & 0 & 0 \\
		  0 &  \multicolumn{1}{c}{\cellcolor[gray]{0.8} 1} & 0 & 0 & \multicolumn{1}{c}{\cellcolor[gray]{0.8} *} & * & 0 & 0 \\
		  0 & 0 & r_1^2 & *  & 0 & 0 & * & 0  \\
		  0 & 0 & 0 & \multicolumn{1}{c}{\cellcolor[gray]{0.5} 1} & 0 & 0 & \multicolumn{1}{c}{\cellcolor[gray]{0.5} *} & * \\
		  0 & \multicolumn{1}{c}{\cellcolor[gray]{0.8} 0} & 0 & 0 & \multicolumn{1}{c}{\cellcolor[gray]{0.8 }1} & * & 0 & 0 \\
		  0 & 0 & 0 & 0 & 0 & r_2^2 & 0 & 0 \\
		  0 & 0 & 0 & \multicolumn{1}{c}{\cellcolor[gray]{0.5} 0} & 0 & 0 & \multicolumn{1}{c}{\cellcolor[gray]{0.5 }1} & * \\
		  0 & 0 & 0 & 0 & 0 & 0 & 0 & r_2^2 \\
	\end{array}
      \right) \,,
\\
      &\TTT_1^{(2)} = \left( \VVVV2 \right)^\dagger \TTT_1^{(1)} \VVVV2  \\*
      &\quad = \left(
	\begin{array}{ccccccccc}
		  1 & * & 0 & 0 & * & 0 & 0 & 0 \\
		  0 &  \multicolumn{1}{c}{\cellcolor[gray]{0.8} 1} & 0 & 0 & \multicolumn{1}{c}{\cellcolor[gray]{0.8} 0} & * & 0 & 0 \\
		  0 & 0 & 1 & *  & 0 & 0 & * & 0  \\
		  0 & 0 & 0 & \multicolumn{1}{c}{\cellcolor[gray]{0.5} 1} & 0 & 0 & \multicolumn{1}{c}{\cellcolor[gray]{0.5} 0} & * \\
		  0 & \multicolumn{1}{c}{\cellcolor[gray]{0.8} 0} & 0 & 0 & \multicolumn{1}{c}{\cellcolor[gray]{0.8} 1} & * & 0 & 0 \\
		  0 & 0 & 0 & 0 & 0 & 1 & 0 & 0 \\
		  0 & 0 & 0 & \multicolumn{1}{c}{\cellcolor[gray]{0.5} 0} & 0 & 0 & \multicolumn{1}{c}{\cellcolor[gray]{0.5} 1} & * \\
		  0 & 0 & 0 & 0 & 0 & 0 & 0 & 1 \\
	\end{array}
      \right) \,,
\\
      &\TTT_3^{(2)} = \UUUU32 \TTT_3^{(1)} \VVVV2  \\*
      &\quad = \left(
	\begin{array}{ccccccccc}
		  r_1^3 & * & 0 & 0 & * & 0 & 0 & 0 \\
		  0 &  \multicolumn{1}{c}{\cellcolor[gray]{0.8} d_2} & 0 & 0 & \multicolumn{1}{c}{\cellcolor[gray]{0.8} *} & * & 0 & 0 \\
		  0 & 0 & r_1^3 & *  & 0 & 0 & * & 0  \\
		  0 & 0 & 0 & \multicolumn{1}{c}{\cellcolor[gray]{0.5} d_2} & 0 & 0 & \multicolumn{1}{c}{\cellcolor[gray]{0.5} *} & * \\
		  0 & \multicolumn{1}{c}{\cellcolor[gray]{0.8} 0} & 0 & 0 & \multicolumn{1}{c}{\cellcolor[gray]{0.8} d_1} & * & 0 & 0 \\
		  0 & 0 & 0 & 0 & 0 & r_2^3 & 0 & 0 \\
		  0 & 0 & 0 & \multicolumn{1}{c}{\cellcolor[gray]{0.5} 0} & 0 & 0 & \multicolumn{1}{c}{\cellcolor[gray]{0.5} d_1} & * \\
		  0 & 0 & 0 & 0 & 0 & 0 & 0 & r_2^3 \\
	\end{array}
      \right) \,,
  \end{align*}
where $d_1 \, d_2 = 1$.

  \underline{Step 3:} \\
  Finally, apply the 1-GMD to $\submat{\TTT_3^{(2)}}{4}{5}$:
  \begin{align}
    \Ul33
    \left(
    \begin{array}{cc}
      d_2 & 0 \\
      0	    & d_1 \\
    \end{array}
    \right)
    \Vl3
    =
    \left(
    \begin{array}{cc}
      1 & * \\
      0	    & 1 \\
    \end{array}
    \right)
    \,.
  \end{align}
  Again, note that the corresponding sub-matrices of $\TTT_1^{(3)}$ and $\TTT_2^{(3)}$ are equal to $I_2$.
  Thus by \propertyref{prop:lenI}, multiplying them by $\Vl3$ on the right and $\left( \Vl3 \right)^\dagger$ on the left, gives rise to the identity matrix $I_2$.
  By defining
\begin{align}
&\UUUU33\triangleq\embb{8}{\Ul33}{\subm{4}{5}}\,,\\
&\VVVV3 \triangleq\embb{8}{\Vl3} {\subm{4}{5}} \,,\\
&\UUUU13=\UUUU23\triangleq \left( \VVVV3 \right)^\dagger \,,
\end{align}

  we arrive to the following three triangular matrices:
 \begin{align}
      &\TTT_3^{(3)} = \UUUU33 \TTT_3^{(2)} \VVVV3  \\*
      &\quad = \left(
	\begin{array}{ccccccccc}
		  r_1^3 & * & 0 & * & * & 0 & 0 & 0 \\
		  0 &  d_2 & 0 & * & * & * & 0 & 0 \\
		  0 & 0 & r_1^3 & *  & * & 0 & * & 0  \\ \cline{4-5}
		  0 & 0 & 0 & \multicolumn{1}{|c}{1} & \multicolumn{1}{c|}{*} & * & * & * \\
		  0 & 0 & 0 & \multicolumn{1}{|c}{0} & \multicolumn{1}{c|}{1} & * & * & * \\ \cline{4-5}
		  0 & 0 & 0 & 0 & 0 & r_2^3 & 0 & 0 \\
		  0 & 0 & 0 & 0  & 0 & 0 & d_1 & * \\
		  0 & 0 & 0 & 0 & 0 & 0 & 0 & r_2^3
	\end{array}
      \right) \,, \\
      &\TTT_2^{(3)} = \left( \VVVV3 \right)^\dagger \TTT_2^{(2)} \VVVV3  \\*
      &\quad = \left(
	\begin{array}{ccccccccc}
		  r_1^2 & * & 0 & * & * & 0 & 0 & 0 \\
		  0 &  1 & 0 & * & * & * & 0 & 0 \\
		  0 & 0 & r_1^2 & *  & * & 0 & * & 0  \\ \cline{4-5}
		  0 & 0 & 0 & \multicolumn{1}{|c}{1} & \multicolumn{1}{c|}{0} & * & * & * \\
		  0 & 0 & 0 & \multicolumn{1}{|c}{0} & \multicolumn{1}{c|}{1} & * & * & * \\ \cline{4-5}
		  0 & 0 & 0 & 0 & 0 & r_2^2 & 0 & 0 \\
		  0 & 0 & 0 & 0  & 0 & 0 & 1 & * \\
		  0 & 0 & 0 & 0 & 0 & 0 & 0 & r_2^2
	\end{array}
      \right) \,,  \\*
  &\TTT_1^{(3)} = \left( \VVVV3 \right)^\dagger \TTT_1^{(2)} \VVVV3  \\
  &\quad = \left(
\begin{array}{ccccccccc}
	  1 & * & 0 & * & * & 0 & 0 & 0 \\
	  0 &  1 & 0 & 0 & 0 & * & 0 & 0 \\
	  0 & 0 & 1 & *  & * & 0 & * & 0  \\ \cline{4-5}
	  0 & 0 & 0 & \multicolumn{1}{|c}{1} & \multicolumn{1}{c|}{0} & * & 0 & * \\
	  0 & 0 & 0 & \multicolumn{1}{|c}{0} & \multicolumn{1}{c|}{1} & * & 0 & * \\ \cline{4-5}
	  0 & 0 & 0 & 0 & 0 & 1 & 0 & 0 \\
	  0 & 0 & 0 & 0  & 0 & 0 & 1 & * \\
	  0 & 0 & 0 & 0 & 0 & 0 & 0 & 1
\end{array}
  \right) \,.
\end{align}

   By taking the middle rows and columns (rows and columns 4 and 5) we achieve the desired decomposition with diagonal elements equaling to 1 in all three triangular matrices,  simultaneously. Formally, we do so by multiplying $\left(\extt8{4,5}\right)^\dagger$ on the left and by $\extt8{4,5}$ on the right (see \remref{remark:extract}) to achieve:

  \begin{align}
  \nonumber
      \left(\extt8{4,5}\right)^\dagger \TTT_{k}^{(3)} \extt8{4,5}
                = \left(
                  \begin{array}{cc}
                    1 & * \\
                    0 & 1 \\
                  \end{array}
                \right)
       \,.
  \end{align}

  Thus, by defining 
  \begin{align}
    \mV &= \mV^{(1)} \mV^{(2)} \mV^{(3)}\extt8{4,5} \\
    \left(\UUU_{k}\right)^\dagger &= \left(\extt8{4,5}\right)^\dagger\UUUU{k}3\UUUU{k}2\UUUU{k} \,,
    \quad k = 1,2,3
    \,,
  \end{align}
we arrive at the desired result.
\hfill$\blacksquare$

%% file: proof_n_2.tex
For $K$ users, we use the same idea, i.e., applying two-by-two 1-GMD operations sequentially on the different channel matrices.
Thus, stating the indices of the four-tuples for which 1-GMD is applied at each step (for each matrix), suffices to establish the desired construction. \\
The proof will be based on $K$ steps.

Denote by $\left\{ \mA_k \right\}$ the extended matrices corresponding to $N$ channel uses.\\
  \underline{Step 1:} \\
  Perform 1-GMD (corresponding to a single channel use) on the matrix $A_1$: $\Ul11 A_1 \Vl1$.\\
  Then, we apply this decomposition to each block separately, using:
\begin{align*}
\UUUU11 &\triangleq \embb{2N}{\Ul11}{\subm{1}{2}\subm{3}{4}\cdots \subm{2N-1}{2N}} \,,\\
\VVVV1  &\triangleq \embb{2N}{\Vl1} {\subm{1}{2}\subm{3}{4}\cdots \subm{2N-1}{2N}}
\,.
\end{align*}
Then, we need to apply the same matrix $\VVVV1$  to all matrices (since the encoder is shared by all users).
  We decompose the resulting matrices (after multiplying them by $\VVVV1$) according to the QR decomposition,
  resulting in unitary matrices $\UUUU{k}{1}$. We denote the resulting extended triangular matrices by 
   $\TTT_k^{(1)} = \UUUU{k}1\mA_k\VVVV1$.

    \underline{Step 2:} \\
  Perform 1-GMD on the matrix $\submat{\TTT_2^{(1)}}{2}{2^{K-1}+1}$:
\begin{align}
  \Ul22 \left(\submat{\TTT_2^{(1)}}{2}{2^{K-1}+1}\right) \Vl2 \,.
\end{align}
  Then, apply this decomposition to each of the matrices, using:
\begin{align*}
\UUUU22 &\triangleq \embb{2N}{\Ul22} {\subunion{q}{2q}{2^{K-1}+2q-1}} \,, \\
\VVVV2  &\triangleq \embb{2N}{\Vl2}  {\subunion{q}{2q}{2^{K-1}+2q-1}} \,,
\end{align*}
for all $q \in \{1\,,2\,,\ldots\,, N-2^{K-2}\}$.

Note that the submatrices of $\TTT_1^{(1)}$ in these indices, 
$\submat {\TTT_1^{(1)}}{2}{2^{K-1}+1} \cdots \submat {\TTT_1^{(1)}}{2N-2^{K-1}}{2N-1}$ are equal to $I_2$; 
by \propertyref{prop:lenI}, multiplying them by $\Vl2$ on the right and $\left( \Vl2 \right)^\dagger$ on the left,
leaves them unchanged.

Then, we need to apply the same matrix $\VVVV2$  to all matrices (since the encoder is shared by all users).
  We decompose the resulting matrices (after multiplying them by $\VVVV2$) according to the QR decomposition,
  resulting in unitary matrices $\UUUU{k}{2}$. We denote the resulting extended triangular matrices by 
  $\TTT_k^{(2)} = \UUUU{k}2 \TTT_k^{(1)}\VVVV2$. 

\vspace{.5\baselineskip}
    \underline{Step $3 \leq l \leq K$:}

  Perform 1-GMD on the matrix \\
  $\submat{\TTT_l^{(l-1)}}{2^{K-1} - 2^{K-(l-1)} + 2}{2^{K-1}+1}$:
\begin{align}
  \Ul{l}{l} \left(\submat{\TTT_l^{(l-1)}}{2^{K-1} - 2^{K-(l-1)} + 2}{2^{K-1}+1}\right) \Vl{l} \,.
\end{align}
  Then, apply this decomposition to each of the extended matrices, using:
\col{
\begin{align*}
& \UUUU{l}{l} \triangleq \\*
& \embb{2N}{\Ul{l}{l}}{\subunion{q}{2^{K-1} - 2^{K-(l-1)} + 2q}{2^{K-1}+2q-1}}\\*
& \VVVV{l} \triangleq \\*
& \embb{2N}{\Vl{l}}{\subunion{q}{2^{K-1} - 2^{K-(l-1)} + 2q}{2^{K-1}+2q-1}}
\end{align*}
}{
\begin{align}
 \UUUU{l}{l} &\triangleq 
 \embb{2N}{\Ul{l}{l}}{\subunion{q}{2^{K-1} - 2^{K-(l-1)} + 2q}{2^{K-1}+2q-1}}\,, \\*
 \VVVV{l} &\triangleq 
 \embb{2N}{\Vl{l}}{\subunion{q}{2^{K-1} - 2^{K-(l-1)} + 2q}{2^{K-1}+2q-1}}\,,
\end{align}
}
for all $q \in \{1\,,2\,,\ldots\,,N-2^{K-2}\}$.

Note that the submatrices of the matrices $\TTT_j^{(l-1)}$ ($j=1\,,...\,,l-1$) in the same indices are all equal to $I_2$;
by \propertyref{prop:lenI}, multiplying them by $\Vl{l}$ on the right and $\left( \Vl{l} \right)^\dagger$ on the left, leaves them unchanged.

Then, we need to apply the same matrix $\VVVV{l}$  to all matrices (since the encoder is shared by all users).
  We decompose the resulting matrices (after multiplying them by $\VVVV{l}$) according to the QR decomposition,
  resulting in unitary matrices $\UUUU{k}{l}$. We denote the resulting extended triangular matrices by 
  $\TTT_k^{(l)} = \UUUU{k}{l} \TTT_k^{(l-1)}\VVVV{l}$.

 \underline{Step $K$:} \\
 After performing the last step (step $l=K$), we are left with $K$ matrices, $\TTT_k^{(K)}$, the central submatrices of which, $\submatt{\TTT_k^{(K)}}{2^{K-1}}{2N-2^{K-1}+1}$, 
 have diagonals equal to $1$. We extract these matrices using the following matrix (see \remref{remark:extract}):
\begin{align}
  \OOO \triangleq \extt{2N}{2^{K-1}:2N-2^{K-1}+1} \,.
\end{align}

  Thus, by defining 
  \begin{align*}
    \left(\UUU_1\right)^\dagger & \triangleq \OOO^\dagger
\left( \VVVV{K} \right)^\dagger\cdots\left( \VVVV2 \right)^\dagger\UUUU11 &\\
    \left(\UUU_k\right)^\dagger &\triangleq \OOO^\dagger \left( \VVVV{K} \right)^\dagger\cdots\left( \VVVV{k+1} \right)^\dagger \UUUU{k}{k} \cdots \UUUU{k}1 &\\
    \left(\UUU_K\right)^\dagger &\triangleq \OOO^\dagger \UUUU{K}{K} \cdots \UUUU{K}1\\
    \mV &\triangleq \mV^{(1)} \mV^{(2)} \cdots \mV^{(K)} \OOO \,, 
  \end{align*}
we arrive at the desired result.
\hfill$\blacksquare$

%% file: proof_K_2.tex
The proof is composed of $K=2$ steps, where, in the case of general $n$, the second step consists of two stages.

\underline{Step 1:} \\
We start by performing 1-GMD (corresponding to a single channel use) on the first matrix $A_1$:
  \begin{align}
    \Ul11 A_1 \Vl1 = \left( \begin{array}{ccccc}
		    1 & * & \cdots & * & * \\
		    0 & 1 & \cdots & * & *\\
		    \vdots & \vdots & \ddots & \vdots & \vdots \\
		    0 & 0 & \cdots & 1 & * \\
		    0 & 0 & \cdots & 0 & 1
		\end{array} \right) \,.
  \end{align}
Apply this decomposition to each block separately, on the first extended matrix, $\AAA_1$, using:
\col{
\begin{align*}
&\UUUU11 \triangleq \\
& \embb{2N}{\Ul11}{\submm{1}{n}\submm{n+1}{2n}\ldots\submm{(N-1)n+1}{Nn}} \\
&\VVVV1  \triangleq \\
& \embb{2N}{\Vl1} {\submm{1}{n}\submm{n+1}{2n}\ldots\submm{(N-1)n+1}{Nn}} \,.
\end{align*}
}{
\begin{align}
\UUUU11 &\triangleq 
 \embb{2N}{\Ul11}{\submm{1}{n}\submm{n+1}{2n}\ldots\submm{(N-1)n+1}{Nn}} \,, \\
\VVVV1  &\triangleq 
 \embb{2N}{\Vl1} {\submm{1}{n}\submm{n+1}{2n}\ldots\submm{(N-1)n+1}{Nn}} \,.
\end{align}
}
Note that the same matrix $\mV^{(1)}$ has to be applied to all matrices (since the encoder is shared by all users).
  We decompose the resulting matrices (after multiplying them by $\mV^{(1)}$) according to the QR decomposition,
  resulting in unitary matrices~$\UUUU21$:
\begin{align*}
      \mT_k^{(1)} & \triangleq   \UUUU{k}1 \AAA_k \VVVV1
\\*
& = \left(
\begin{array}{c:c:c:c:c}
T_{k}^{(1)} & 0 & \cdots & 0 & 0 \\ \hdashline
	      0 & T_{k}^{(1)} & \cdots & 0 & 0 \\ \hdashline	
              \vdots &  \vdots &   \vdots        &    \ddots    &  \vdots \\ \hdashline
	      0 & 0        & \cdots      & T_{k}^{(1)} & 0 \\              \hdashline
		0 & 0 & \cdots         &0     & T_{k}^{(1)}
	\end{array} \right)\,,\qquad k=1,2\,,
\end{align*}
where,
  \begin{align}
    T_{1}^{(1)} &\triangleq \left( \begin{array}{ccccc}
		    1 & * & \cdots & * & * \\
		    0 & 1 & \cdots & * & *\\
		    \vdots & \vdots & \ddots & \vdots & \vdots \\
		    0 & 0 & \cdots & 1 & * \\
		    0 & 0 & \cdots & 0 & 1
		\end{array} \right) \,, \\
    T_{2}^{(1)} &\triangleq \left( \begin{array}{ccccc}
		    r_1 & * & \cdots & * & * \\
		    0 & r_2 & \cdots & * & *\\
		    \vdots & \vdots & \ddots & \vdots & \vdots \\
		    0 & 0 & \cdots & r_{n-1} & * \\
		    0 & 0 & \cdots & 0 & r_n
		\end{array} \right) \,.
  \end{align}

\noindent
\underline{Step 2:}\\
This step consists of 2 stages: the first is the reordering stage and the second is application of 1 1-GMD to each block.\\
\qquad \underline{Stage 1: Reordering}

It is convenient to reorder the columns of $\mT_{k}^{(1)}$ such that the columns 
\begin{align}
kn,kn+(n-1),kn+2(n-1), \cdots, kn+(n-1)^2
\end{align}
are ``grouped together'' for every $k$.\footnote{Note that this set includes exactly one symbol from each of $n$ consecutive channel uses.} 
Formally, we do so by applying the $nN \times n(N-n+1)$ reordering matrix 
\begin{align}
~\OOO = \extt{nN}{kn,kn+(n-1),kn+2(n-1), \cdots, kn+(n-1)^2}.
\end{align}

 The reordering stage gives rise to the following matrices of dimensions \mbox{$n(N-n+1) \times n(N-n+1)$}:
 \begin{align*}
 \mT_k^{(2)(1)} &\triangleq   \left(\OOO\right)^{\dagger} \mT_k^{(1)} \OOO \\
  &=\left(
  \begin{array}{c:c:c:c:c}
  T_{k}^{(2)(1)} & * & \cdots & * & * \\ \hdashline
	      0 & T_{k}^{(2)(1)} & \cdots & * & * \\ \hdashline	
              \vdots &  \vdots &   \vdots        &    \ddots    &  \vdots \\ \hdashline
	      0 & 0        & \cdots      & T_{k}^{(2)(1)} & * \\              \hdashline
		0 & 0 & \cdots         &0     & T_{k}^{(2)(1)}
	\end{array} \right) \,,\\
    & \qqqquad \qqqquad \qqqquad \qqqquad \qquad k=1\,,2\,,
\end{align*}
where,
  \begin{align}
    T_{1}^{(2)(1)} &\triangleq \left( \begin{array}{ccccc}
		    1 & 0 & \cdots & 0 & 0 \\
		    0 & 1 & \cdots & 0 & 0\\
		    \vdots & \vdots & \ddots & \vdots & \vdots \\
		    0 & 0 & \cdots & 1 & 0 \\
		    0 & 0 & \cdots & 0 & 1
		\end{array} \right) \,, \\
    T_{2}^{(2)(1)} &\triangleq \left( \begin{array}{ccccc}
		    r_n & 0 & \cdots & 0 & 0 \\
		    0 & r_{n-1} & \cdots & 0 & 0\\
		    \vdots & \vdots & \ddots & \vdots & \vdots \\
		    0 & 0 & \cdots & r_{2} & 0 \\
		    0 & 0 & \cdots & 0 & r_1
		\end{array} \right) \,,
  \end{align}
  the superscripts denote the step and stage number, and the subscripts denote the user number.

\qquad \underline{Stage 2: 1-GMD} \\
Perform 1-GMD on the matrix $T_{2}^{(2)(1)}$:
  \begin{align}
    \Ul22 T_{2}^{(2)(1)} \Vl2 = \left( \begin{array}{ccccc}
		    1 & * & \cdots & * & * \\
		    0 & 1 & \cdots & * & *\\
		    \vdots & \vdots & \ddots & \vdots & \vdots \\
		    0 & 0 & \cdots & 1 & * \\
		    0 & 0 & \cdots & 0 & 1
		\end{array} \right) \,.
  \end{align}

Note that the matrix $T_{1}^{(2)(1)}$ is equal to $I_n$; by \propertyref{prop:lenI}, multiplying it by $\Vl2$ on the right and $\left( \Vl2 \right)^\dagger$ on the left,
leaves it unchanged.

We now apply this decomposition to each block separately, using 
\begin{align*}
\UUUU22 &\triangleq \embb{n(N-n+1)}{\Ul22}{\subunionn{q}{1+n(q-1)}{qn}}\,,\\
\VVVV2  &\triangleq \embb{n(N-n+1)}{\Vl2} {\subunionn{q]}{1+n(q-1)}{qn}}\,,
\end{align*}
for all $q \in \{1\,,2\,,\ldots\,,N-n+1\}$, 
which results in the extended triangular matrices 
 \begin{align}
      \mT_k^{(2)} &\triangleq   \UUUU{k}2 \mT_k^{(2)(1)} \VVVV2 \\
&=  \left(
\begin{array}{c:c:c:c:c}
 T_{k}^{(2)} & * & \cdots & * & * \\ \hdashline
	      0 & T_{k}^{(2)} & \cdots & * & * \\ \hdashline	
              \vdots &  \vdots &   \vdots        &    \ddots    &  \vdots \\ \hdashline
	      0 & 0        & \cdots      & T_{k}^{(2)} & * \\              \hdashline
		0 & 0 & \cdots         &0     & T_{k}^{(2)}
	\end{array} \right)\,, k=1\,,2\,,
\end{align}
where,
  \begin{align}
    T_{1}^{(2)} &= \left( \begin{array}{ccccc}
		    1 & 0 & \cdots & 0 & 0 \\
		    0 & 1 & \cdots & 0 & 0\\
		    \vdots & \vdots & \ddots & \vdots & \vdots \\
		    0 & 0 & \cdots & 1 & 0 \\
		    0 & 0 & \cdots & 0 & 1
		\end{array} \right) \,,
\\
    T_{2}^{(2)} &= \left( \begin{array}{ccccc}
		    1 & * & \cdots & * & * \\
		    0 & 1 & \cdots & * & *\\
		    \vdots & \vdots & \ddots & \vdots & \vdots \\
		    0 & 0 & \cdots & 1 & * \\
		    0 & 0 & \cdots & 0 & 1
		\end{array} \right) \,.
  \end{align}
  Thus, by defining 
\begin{align}
\qqquad\VVV &\triangleq \VVVV1\OOO\VVVV2 \\
\left(\UUU_1\right)^\dagger &\triangleq \left(\VVVV2\right)^\dagger  \left(\OOO\right)^{\dagger} \UUUU11 \\
\left(\UUU_2\right)^\dagger &\triangleq \UUUU22  \left(\OOO\right)^{\dagger} \UUUU21
\end{align}
we arrive at the desired result.
\hfill$\blacksquare$

%% file: proof_general.tex
  \noindent
   The proof for the case of $K$ users, follows the same principles of the special cases presented in \secref{ss:space_time_proof_K_2_n_2} and Appendices \ref{app:proof_K_3_n_2}, \ref{app:proof_n_2}, \ref{app:proof_K_2}. The proof is composed of $K$ steps, each of which consists of 2 stages (except for the first step): a reordering stage and a 1-GMD stage.

  Denote by $\left\{ \mA_k \right\}$ the extended matrices corresponding to $N$ channel uses.

  \underline{Step 1:} \\
  Perform 1-GMD on the first matrix matrix $A_1$ (corresponding to to a single channel use): $\Ul11 A_1 \Vl1$.
  Apply this decomposition to each block separately, on the first extended matrix $\AAA_1$, using:
\col{
    {{\small
    \begin{align}
     &\UUUU11 \triangleq  
     \\& \embb{nN}{\Ul11}{\submm{1}{n}\submm{n+1}{2n}\ldots\submm{(N-1)n+1}{Nn}} , 
     \\& \VVVV1  \triangleq  
     \\& \embb{nN}{\Vl1}{\submm{1}{n}\submm{n+1}{2n}\ldots\submm{(N-1)n+1}{Nn}} .
    \end{align}
    }}
}{
\begin{align}
 \UUUU11 &\triangleq  \embb{nN}{\Ul11}{\submm{1}{n}\submm{n+1}{2n}\ldots\submm{(N-1)n+1}{Nn}} \,, \\
 \VVVV1  &\triangleq  \embb{nN}{\Vl1}{\submm{1}{n}\submm{n+1}{2n}\ldots\submm{(N-1)n+1}{Nn}} \,.
\end{align}
}
Note that the same matrix $\VVVV1$ has to be applied to all matrices (since the encoder is shared by all users).
  We decompose the resulting matrices (after multiplying them by $\VVVV1$) according to the QR decomposition,
  resulting in unitary matrices $\UUUU{k}{1}$. The resulting extended triangular matrices are denoted by
   $\mT_k^{(1)} \triangleq \UUUU{k}{1}\mA_k\VVVV1$.

\vspace{.3\baselineskip}
    \underline{Step $2\leq{l}\leq K$:}

\vspace{.3\baselineskip}
\underline{Stage 1: Reordering} \\
We perform the ordering stage using the following ordering matrix, for all $q_1 \in \{1\,,2\,,\ldots\,,N-n^{K-1}+n^{K-{l}}\}$ and \mbox{$q_2 \in \{1\,,2\,,\ldots\,, n\}$}:
\begin{align}
&~\OOO^{l} \triangleq \extt{nN-n^{(K-{l}+1)}\left(n^{({l}-1)}-1\right)}{  \left\{ \left\{n+(q_1-1)n+(q_2-1)\Delta\right\}_{q_2}\right\}_{q_1}}
\,,
\end{align}
where $\Delta=n^{K-{l}+1}-1$. Note that the range of $q_2$ is equal to the dimension $n$ of each block, whereas the range of $q_2$ is determined by the number of blocks, which depends on ${l}$.

Thus, at the end of the first stage, we are left with
$\mT_k^{({l})(1)}=\left(\OOO^{l}\right)^\dagger\mT_k^{({l}-1)}\OOO^{l} $.
Note that in each step the size of $\mT_k^{({l})(1)}$ is decreasing.\\

\qquad \underline{Stage 2: 1-GMD} \\
  Perform 1-GMD on the matrix $\submatt{\mT_{l}^{({l})(1)}}{1}{n}$ using:
\begin{align}
  \Ul{{l}}{{l}} \left(\submatt{\mT_{l}^{({l})(1)}}{1}{n}\right) \Vl{{l}} \,.
\end{align}
  Then, apply this decomposition to each of the extended matrices, using:
\begin{align*}
\UUUU{{l}}{{l}}&\triangleq \embb{nN-n^{(K-{l}+2)}}{\Ul{{l}}{{l}}}{\subunion{q}{1+n(q-1)}{nq}}\\*
\VVVV{{l}}      &\triangleq \embb{nN-n^{(K-{l}+2)}}{\Vl{{l}}}      {\subunion{q}{1+n(q-1)}{nq}}
\,,
\end{align*}
for all $q \in \{1\,,2\,,\ldots\,,(N-n^{K-1}+n^{K-{l}})\}$.

Note that the submatrices of $\mT_k^{({l})(1)}$ ($k=1\,,...\,,{l}-1$) in the same indices are all equal $I_n$;
by \propertyref{prop:lenI}, multiplying them by $\Vl{{l}}$ on the right and $\left( \Vl{{l}} \right)^\dagger$ on the left,
leave them unchanged.

The same matrix $\VVVV{{l}}$  has to be applied to all matrices (since the encoder is shared by all users).
  We decompose the resulting matrices (after multiplying them by $\VVVV{l}$) according to the QR decomposition,
  resulting in unitary matrices $\UUUU{k}{{l}}$. The resulting extended triangular matrices will be denoted as
   $\mT_k^{({l})} \triangleq \UUUU{k}{{l}} \mT_k^{({l})(1)}\VVVV{l}$.

\vspace{.5\baselineskip}
 \underline{Step $K$:} \\
 After performing the last step (step ${l}=K$) we attain $K$ matrices $\mT_k^{(K)}$
 which all have 1s on theirs diagonals.

  Thus, by defining 
\col{  
\begin{align}
    \left(\UUU_1\right)^\dagger &\triangleq \left( \VVVV{K} \right)^\dagger\left( \OOO^K \right)^\dagger\cdots\left( \VVVV2 \right)^\dagger\left( \OOO^2 \right)^\dagger\UUUU11 &\\
    \left(\UUU_k\right)^\dagger &\triangleq \left( \VVVV{K} \right)^\dagger\left( \OOO^K \right)^\dagger\cdots\left( \VVVV{k+1} \right)^\dagger\left( \OOO^{k+1} \right)^\dagger \cdot \\
& \phantom{\triangleq} \cdot \UUUU{k}{k} \left( \OOO^{k} \right)^\dagger\cdots \UUUU{k}1 &\\
    \left(\UUU_K\right)^\dagger &\triangleq \UUUU{K}{K} \left( \OOO^K \right)^\dagger\cdots \UUUU{K}1\\
    \mV &\triangleq \mV^{(1)} \OOO^2 \mV^{(2)} \cdots \OOO^K\mV^{(K)} \,,
  \end{align}
}
{
\begin{align}
    \left(\UUU_1\right)^\dagger &\triangleq \left( \VVVV{K} \right)^\dagger\left( \OOO^K \right)^\dagger\cdots\left( \VVVV2 \right)^\dagger\left( \OOO^2 \right)^\dagger\UUUU11 &\\
    \left(\UUU_k\right)^\dagger &\triangleq \left( \VVVV{K} \right)^\dagger\left( \OOO^K \right)^\dagger\cdots\left( \VVVV{k+1} \right)^\dagger\left( \OOO^{k+1} \right)^\dagger \cdot 
\UUUU{k}{k} \left( \OOO^{k} \right)^\dagger\cdots \UUUU{k}1 &\\
    \left(\UUU_K\right)^\dagger &\triangleq \UUUU{K}{K} \left( \OOO^K \right)^\dagger\cdots \UUUU{K}1\\
    \mV &\triangleq \mV^{(1)} \OOO^2 \mV^{(2)} \cdots \OOO^K\mV^{(K)} \,,
  \end{align}
}
we arrive at the desired result.
\hfill$\blacksquare$

%% file: upper_lower_proof.tex
We can assume without loss of generality that the matrix $V$ is of the following form:
\begin{align}
    V = \left( \begin{array}{cc}
                  x_1 + ix_2 & y_1 - iy_2 \\
		  y_1 + iy_2 & -x_1+ix_2
               \end{array}
\right) \,,
\end{align}
where $x_1,x_2,y_1,y_2$ are real numbers satisfying
\begin{align}
 x_1^2+x_2^2 + y_1^2 + y_2^2 = 1 \,. 
\end{align}
Denote the first column of $V$ by $\bv_1$ and the second column by $\bv_2$.
Then, there exist unitary matrices $U_1,U_2$ such that
\begin{align}
      {{  \left( U_1  \right)^\dagger    }} A_1 V =  \left( \begin{array}{cc}
		      1 & * \\ 0 & 1
                   \end{array} \right)
\end{align}
and
\begin{align}
      {{  \left( U_2  \right)^\dagger    }} A_2 V =  \left( \begin{array}{cc}
		      1 & 0 \\ * & 1
                   \end{array} \right) \,,
\end{align}
if and only if the following two vectors have an Euclidean norm of $1$:
\begin{align}
   A_1 \bv_1 &= A_1 \left( \begin{array}{c}
x_1 + ix_2 
\\ y_1 + iy_2
\end{array}     \right)
\\
A_2 \bv_2 &= A_2 \left( \begin{array}{c}
y_1 - iy_2
\\ -x_1 + ix_2
\end{array}     \right) \,,
\end{align}
or equivalently,
\col{
\begin{subequations}
\label{eq:upper_lower_three_equations}
\noeqref{eq:upper_lower_three_equations:1,eq:upper_lower_three_equations:2,eq:upper_lower_three_equations:3,eq:upper_lower_three_equations:4}
\begin{align} 
\label{eq:upper_lower_three_equations:1}
  \left\| \bv_1 \right\|^2 &= 1 \\
\label{eq:upper_lower_three_equations:2}
  \left\| \bv_2 \right\|^2 &= 1 \\
\label{eq:upper_lower_three_equations:3}
  \bv_1^\dagger  \left( A_1^\dagger A_1 - I \right) \bv_1 &= 0 \\
\label{eq:upper_lower_three_equations:4}
  \bv_2^\dagger  \left( A_2^\dagger A_2 - I \right) \bv_2 &= 0 \,.
\end{align}
\end{subequations}
}{
\begin{align} \label{eq:upper_lower_three_equations}
\begin{aligned}
  \left\| \bv_1 \right\|^2 &= \left\| \bv_2 \right\|^2 = 1 \\
  \bv_1^\dagger  \left( A_1^\dagger A_1 - I \right) \bv_1 &= \bv_2^\dagger  \left( A_2^\dagger A_2 - I \right) \bv_2 = 0 \\
\end{aligned}
\end{align}}
By definition, $\left\| \bv_1 \right\|^2 = \left\| \bv_2 \right\|^2$,
and also that for any Hermitian $2 \times 2$ matrix $S$:
\begin{align}
     \bv_2^\dagger  S \bv_2  =  \bv_1^\dagger \adj \left( S \right) \bv_1 \,.
\end{align}
Thus, \eqref{eq:upper_lower_three_equations} is equivalent to
\col{
\begin{align}
\begin{aligned}
  \left\| \bv_1 \right\|^2 &= 1 \\
  \bv_1^\dagger S_1 \bv_1 &= 0 \\
  \bv_1^\dagger S_2 \bv_1 &= 0 \,,
\end{aligned}
\end{align}
}{
\begin{align}
\begin{aligned}
  \left\| \bv_1 \right\|^2 &= 1 \\
  \bv_1^\dagger S_1 \bv_1 &=  \bv_1^\dagger S_2 \bv_1 = 0 \\
\end{aligned}
\end{align}}
where
\begin{align}
  S_1 &\triangleq A_1^\dagger A_1 - I \\
  S_2 &\triangleq \adj \left(  A_2^\dagger A_2 - I \right) \,.
\end{align}
Since $\det(A_1)=\det(A_2)=1$, we have
\begin{align}
  \det(S_1) & \leq 0 \\
  \det(S_2) & \leq 0 \,.
\end{align}
Thus, from \lemref{lem:det_adj_complex} it follows that a solution exists if and only if
\begin{align} \label{eq:qqq124}
    \det \left( S_1 \mathrm{adj} (S_2) - S_2 \mathrm{adj} (S_1) \right) \geq 0 \,.
\end{align}
Note that for any $2 \times 2$ matrix $A$,
\begin{align}
      \mathrm{adj} (\adj (A)) =  A \,.
\end{align}
Hence, the left hand side of condition \eqref{eq:qqq124} can be written as
\col{
\begin{align}
 & \det \left(   S_1 \adj (S_2) - S_2 \adj (S_1) \right)
\\
& =
\det \left(
(A_1^\dagger A_1 - I)(A_2^\dagger A_2 - I) \right.
\\
& \qqquad \left. - \adj(A_2^\dagger A_2 - I) \adj (A_1^\dagger A_1 - I)
\right)
\\
&= F_2 \left( A_1^\dagger A_1-I , A_2^\dagger A_2 - I \right)\,,
\end{align}
}
{
\begin{align}
  \det \left(   S_1 \adj (S_2) - S_2 \adj (S_1) \right)
&=
\det \left(
(A_1^\dagger A_1 - I)(A_2^\dagger A_2 - I) \right.
\left. - \adj(A_2^\dagger A_2 - I) \adj (A_1^\dagger A_1 - I)
\right)
\\
&= F_2 \left( A_1^\dagger A_1-I , A_2^\dagger A_2 - I \right)\,,
\end{align}
}
which completes the proof of the theorem.
\hfill$\blacksquare$

%% file: proof_GTD_with_multiplicities.tex
Denote the vector consisting of $\left\{ r_{m} \right\}$ with their multiplicities, ordered non-increasingly, by $\br$ and the vector whose entries are the singular values of $A$, $\left\{ \sigma_j \right\}$, ordered non-increasingly, by~$\bsigma$.
According to the GTD \cite{GTD}, the decomposition \eqref{eq:append:A=URV'} is possible if and only if Weyl's condition \cite{WeylCondition,WeylConditionInverse_ByHorn},
\begin{align}
\label{eq:Major_Append}
    \bsigma \succeq \br \,,
\end{align}
holds true. Namely, $n$ conditions need to be evaluated.
We shall show next that when at least some of the absolute values of the desired diagonal $R_{jj}$ are of multiplicity greater than 1, such that there are $M<n$ distinct such (absolute) values, only $M$ of these conditions, \eqref{eq:WeylMultiplicityCond}-\eqref{eq:WeylMultiplicityCondEqual}, need to be evaluated.
The necessity of \eqref{eq:WeylMultiplicityCond}-\eqref{eq:WeylMultiplicityCondEqual} is apparent since they constitute the $n_1,n_1+n_2,...,n$ conditions in \eqref{eq:Major_Append}.

We shall prove the sufficiency of these conditions by induction.

\textbf{Basis:}
We shall show first that the $n_1$ condition in \eqref{eq:Major_Append} is sufficient for all the first $n_1$ conditions in \eqref{eq:Major_Append} to hold:
Assume that
\begin{align}
    r_1^{n_1} &\leq \prod_{j=1}^{n_1} \sigma_j \,,
\end{align}
holds true. This condition can be rewritten as
\begin{align}
    r_1 \leq \sqrt[n_1]{\prod_{j=1}^{n_1} \sigma_j} \,.
\end{align}
Using the fact that the geometric-mean of a set of size $n_1$ cannot be larger than the geometric-mean of its largest $q$ values ($q=1,...,n_1-1$), we have
\begin{align}
    r_1 \leq \sqrt[n_1]{\prod_{j=1}^{n_1} \sigma_j} \leq \sqrt[q]{\prod_{j=1}^{q} \sigma_j} \,, \quad q=1,...,n_1-1 \,,
\end{align}
or equivalently,
\begin{align}
    r_1^q \leq \prod_{j=1}^{q} \sigma_j \,, \quad q=1,...,n_1-1 \,,
\end{align}
which are exactly equivalent to the first $n_1$ conditions of \eqref{eq:Major_Append}.

\textbf{Inductive step:} Assume that the conditions \eqref{eq:WeylMultiplicityCond}-\eqref{eq:WeylMultiplicityCondEqual} guarantee that the first $\sum_{m=1}^{k-1} n_m$ conditions in \eqref{eq:Major_Append} are satisfied.
We shall prove that all the first $\sum_{m=1}^k n_m$ conditions in \eqref{eq:Major_Append} hold true.
We shall now show that if the $\sum_{m=1}^{k} n_m$ condition in \eqref{eq:Major_Append} holds true (which is the $k$-th condition in \eqref{eq:WeylMultiplicityCond}),
then so do the $n_k-1$ conditions that precede it.
Let $q$ be some integer between $1$ and $M$, and assume that
\begin{align}
    \prod_{m=1}^q r_m^{n_m} &\leq \prod_{j=1}^{\sum_{m=1}^q {n_m}} \sigma_j \,,
\end{align}
which can be equivalently written as
\begin{align}
\label{eq:r^n<=gamma*Pi_lambda_i}
    r_q^{n_q} \leq \gamma \prod_{j = \left( \sum_{m=1}^{q-1} {n_m} \right) + 1}^{\sum_{m=1}^q {n_m}} \sigma_j \,,
\end{align}
where $\gamma$ is defined as
\begin{align}
    \gamma \triangleq \prod_{m=1}^{q-1} r_m^{-n_m} \prod_{j=1}^{\sum_{m=1}^{q-1} {n_m}} \sigma_j
\end{align}
and is equal or larger than $1$.

Let $l$ be some integer between $1$ and $n_q-1$, and assume, to contradict, that
\begin{align}
    \left( \prod_{m=1}^{q-1} r_m^{n_m} \right) r_q^l &> \prod_{j=1}^{\left( \sum_{m=1}^{q-1} n_m \right) + l} \sigma_j \,,
\end{align}
or equivalently,
\begin{align}
\label{eq:r^q>gamma*Pi_lambda_i}
    r_q^l > \gamma \prod_{j = \left( \sum_{m=1}^{q-1} {n_m} \right) + 1}^{\left( \sum_{m=1}^{q-1} n_m \right) + l} \sigma_j \,.
\end{align}
Dividing \eqref{eq:r^n<=gamma*Pi_lambda_i} by \eqref{eq:r^q>gamma*Pi_lambda_i} gives rise to
\begin{align}
    r_q^{n_q-l} < \prod_{j = \left( \sum_{m=1}^{q-1} n_m \right) + l + 1}^{\sum_{m=1}^q n_m} \sigma_j \,,
\end{align}
which can be written as
\begin{align}
    r_q < \sqrt[n_q-l]{\prod_{j = \left( \sum_{m=1}^{q-1} n_m \right) + l + 1}^{\sum_{m=1}^q n_m} \sigma_j} \,.
\end{align}
Using the fact that the geometric-mean of the smallest $n_q-l$ values of a set of positive numbers is equal or smaller than the geometric mean of its $l$ largest values,
and the fact that $\gamma \geq 1$, we have
\begin{align}
    r_q < \sqrt[n_q-l]{\prod_{j = \left( \sum_{m=1}^{q-1} n_m \right) + l + 1}^{\sum_{m=1}^q n_m} \sigma_j}
    \leq \sqrt[l]{\gamma \prod_{j = \left( \sum_{m=1}^{q-1} n_m \right) + 1}^{\left( \sum_{m=1}^{q-1} n_m \right) + l} \sigma_j} \,.
\end{align}
i.e.,
\begin{align}
    r_q^l < \gamma \prod_{j = \left( \sum_{m=1}^{q-1} n_m \right) + 1}^{\left( \sum_{m=1}^{q-1} n_m \right) + l} \sigma_j \,,
\end{align}
in contradiction to \eqref{eq:r^q>gamma*Pi_lambda_i}.
\hfill $\blacksquare$

%% file: bios.tex
\begin{IEEEbiographynophoto}{Anatoly Khina}
    was born in Moscow, USSR, on September 10, 1984.
    He received the B.Sc.\ and M.Sc.\ degrees in electrical engineering (both \emph{summa cum laude}) from Tel Aviv University in 2006 and 2010, respectively, where he is currently working towards completing his  Ph.D.\ degree. 
    His research interests include information theory, signal processing, digital communications and matrix analysis.

    In parallel to his studies, Anatoly has been working as an engineer in various algorithms, software and hardware R\&D positions. 
    He is a recipient of the Rothschild fellowship, Clore scholarship, Trotsky Award, Weinstein Prize for research in signal processing, and the first prize for outstanding research work of the Advanced Communication Center, Israel.
\end{IEEEbiographynophoto}

\begin{IEEEbiographynophoto}{Idan Livni}
    was born in Tel-Aviv, Israel, on July 30, 1984. He received the B.Sc.\ and M.Sc.\ degrees (both \emph{cum laude}) in electrical engineering from Tel Aviv University in 2006 and 2013, respectively. 
    His research interests are in digital communications, signal processing and information theory.
\end{IEEEbiographynophoto}

\begin{IEEEbiographynophoto}{Ayal Hitron}
    received a B.Sc.\ (\emph{summa cum laude}) in electrical engineering and a B.A.\ in Physics (\emph{summa cum laude}), both from the Technion~--- Israel Institute of Technology in 2003, and an M.Sc.\ in electrical engineering (\emph{summa cum laude}), from Tel Aviv University in 2012.

    Ayal is the recipient of a bronze medal in the International Physics Olympiad (IPhO), the Knesset award for outstanding undergraduate student achievements, and the Weinstein Prize for research in signal processing.
\end{IEEEbiographynophoto}

\begin{IEEEbiographynophoto}{Uri Erez}
    (M'09) was born in Tel-Aviv, Israel, on October 27, 1971.
    He received the B.Sc.\ degree in mathematics and physics and the M.Sc.\ and
    Ph.D.\ degrees in electrical engineering from Tel-Aviv University in 1996,
    1999, and 2003, respectively. 
    During 2003--2004, he was a Postdoctoral Associate at the Signals, Information and Algorithms Laboratory at the
    Massachusetts Institute of Technology (MIT), Cambridge. 
    Since 2005, he has been with the Department of Electrical Engineering--Systems at Tel-Aviv
    University. 
    His research interests are in the general areas of information theory and digital communications. 
    He served in the years 2009--2011 as Associate Editor for Coding Techniques for the 
    {\sc IEEE Transactions on Information Theory}.
\end{IEEEbiographynophoto}